\documentclass[prx,twocolumn,aps]{revtex4-2}
\usepackage{pdfpages}
\makeatletter
\AtBeginDocument{\let\LS@rot\@undefined}
\makeatother

\usepackage{amsmath}
\usepackage[utf8]{inputenc}
\usepackage[T1]{fontenc}
\usepackage{pdfpages}
\usepackage{graphicx}
\usepackage{epstopdf}
\usepackage{amssymb}
\usepackage{mathtools}
\usepackage{amsthm}
\usepackage{dsfont}
\usepackage[normalem]{ulem}
\usepackage[toc,page]{appendix}
\usepackage{soul}
\usepackage{framed}
\usepackage[most]{tcolorbox}
\usepackage{tabularx}
\usepackage{diagbox}
\usepackage{amsthm}
\usepackage[qm]{qcircuit}
\usepackage{braket}
\usepackage{subfigure}
\usepackage{extarrows}
\usepackage{booktabs}
\newcommand\headercell[1]{
   \smash[b]{\begin{tabular}[t]{@{}c@{}} #1 \end{tabular}}}
\usepackage{multirow}
\usepackage[colorlinks]{hyperref}
\AtBeginDocument{
  \hypersetup{
    citecolor=blue,
    linkcolor=blue,
    urlcolor=blue}}

\usepackage[capitalize,noabbrev]{cleveref}

\theoremstyle{plain}
\newtheorem{theorem}{Theorem}
\newtheorem{definition}[theorem]{Definition}
\newtheorem{corollary}[theorem]{Corollary}
\newtheorem{lemma}[theorem]{Lemma}

\newcounter{rowItemCount}
\newcounter{subRowItemCount}

\usepackage{newfloat}
\usepackage{algcompatible}
\usepackage[size=small,justification=centerlast]{caption}
\usepackage{etoolbox}
\AtEndEnvironment{algorithm}{\noindent\hrulefill\par\nobreak\vskip-8pt}

\DeclareFloatingEnvironment[
fileext=loa,
listname=List of Algorithms,
name=ALGORITHM,
placement=tbhp,
]{algorithm}
\DeclareCaptionFormat{algorithms}{\vskip-6pt\hrulefill\par#1#2#3\vskip-6pt\hrulefill}

\algblock[Input]{Input}{EndInput}
\algblockdefx[Input]{Input}{EndInput}
[1]{\textbf{Input} #1}
{}

\newcommand{\diag}{\mathrm{diag}~}

\newcommand{\I}{\mathrm{i}}

\newcommand{\mc}[1]{\mathcal{#1}}
\newcommand{\mf}[1]{\mathfrak{#1}}
\newcommand{\wt}[1]{\widetilde{#1}}

\newcommand{\abs}[1]{\left\lvert#1\right\rvert}
\newcommand{\norm}[1]{\left\lVert#1\right\rVert}

\newcommand{\ud}{\,\mathrm{d}}
\newcommand{\Or}{\mathcal{O}}

\newcommand{\RR}{\mathbb{R}}
\newcommand{\CC}{\mathbb{C}}

\newcommand{\II}{\mathds{1}}
\newcommand{\bP}{\mathds{P}}

\newcommand{\rd}{\mathrm{d}}

\newcommand{\ugrp}{\mathsf{U}}

\newcommand{\expt}[1]{\mathbb{E}\left( #1 \right)}
\newcommand{\exptwrt}[2]{\mathbb{E}_{#1}\left( #2 \right)}
\newcommand{\expl}{\text{exp}}

\newcommand{\nsys}{{n}}
\newcommand{\Nsys}{{N}}
\newcommand{\nbe}{{m}}
\newcommand{\Nbe}{{M}}

\newcommand{\svt}{\vartriangleright}

\newcommand{\eig}{{\text{eig}}}

\newcommand{\maxl}{\text{max}}

\newcommand{\argminl}{\mathop{\text{argmin}}}
\newcommand{\SP}{{Supplementary Note}}

\newcommand{\circqsvt}{\mc{U}}

\crefname{equation}{Eq.}{Eqs.}
\crefname{section}{}{}

\begin{document}

\title{A quantum hamiltonian simulation benchmark}

\newcommand{\DeptMath}{Department of Mathematics, University of California, Berkeley, California 94720 USA}
\newcommand{\LBLMath}{Computational Research Division, Lawrence Berkeley National Laboratory, Berkeley, CA 94720, USA}
\newcommand{\BQIC}{Berkeley Center for Quantum Information and Computation, Berkeley, California 94720 USA}
\newcommand{\DeptChem}{Department of Chemistry, University of California, Berkeley, California 94720 USA}
\newcommand{\CIQC}{Challenge Institute of Quantum Computation, University of California, Berkeley, California 94720 USA}

\author{Yulong Dong$^{1,2}$} 
\author{K. Birgitta Whaley$^{1,3,4}$}
\author{Lin Lin$^{2,4,5}$}
\email{Electronic address: linlin@math.berkeley.edu}

\affiliation{$^1$\BQIC}
\affiliation{$^2$\DeptMath}
\affiliation{$^3$\DeptChem}
\affiliation{$^4$\CIQC}
\affiliation{$^5$\LBLMath}

\begin{abstract}
Hamiltonian simulation is one of the most important problems in quantum computation, and quantum singular value transformation (QSVT) is an efficient way to simulate a general class of Hamiltonians. However, the QSVT circuit typically involves multiple ancilla qubits and multi-qubit control gates. In order to simulate a certain class of $n$-qubit random Hamiltonians, we propose a drastically simplified quantum circuit that we refer to as the minimal QSVT circuit, which uses only one ancilla qubit and no multi-qubit controlled gates. We formulate a simple metric called the quantum unitary evolution score (QUES), which is a scalable quantum benchmark and can be verified without any need for classical computation. Under the globally depolarized noise model, we demonstrate that QUES is directly related to the circuit fidelity, and the potential classical hardness of an associated quantum circuit sampling problem. Under the same assumption, theoretical analysis suggests  there exists an `optimal' simulation time $t^{\text{opt}}\approx 4.81$, at which even a noisy quantum device may be sufficient to demonstrate the potential classical hardness.

\end{abstract}

\maketitle

\renewcommand*{\thesubfigure}{\alph{subfigure}}

\noindent {\large \textbf{Introduction}}

Recent years have witnessed tremendous progress in quantum hardware and quantum algorithms.
As near-term quantum devices become increasingly accessible, the need for holistic benchmarking of such devices is also rapidly growing.
Indeed, while most of the frequently used quantum benchmarks, such as randomized benchmarking \cite{MagesanGambettaEmerson2011} and gateset tomography \cite{Blume-KohoutGambleNielsenEtAl2017}, still focus on the performance of one or a few qubits, 
over the past three years a number of `whole machine' benchmarks have been proposed that aim at assessing the performance of quantum devices beyond a small number of qubits~\cite{BoixoIsakovSmelyanskiyEtAl2018,CrossBishopSheldonEtAl2019,ErhardWallmanPostlerEtAl2019,AruteAryaBabbushEtAl2019,ProctorRudingerYoungEtAl2022,CornelissenBauschGilyen2021,DongLin2021,ProctorSeritanRudingerEtAl2021}.

While results from such generic benchmarks certainly provide important characteristics of the quantum devices themselves, we are ultimately interested in applying the devices to carry out specific computational tasks. However, the circuit structure of quantum algorithms can be vastly different for different algorithms. Generic quantum benchmarks can miss structural information that is specific to a particular algorithm and which may amplify either quantum errors of certain types or errors amongst a certain group of qubits, and/or reduce errors elsewhere.
In this work we address the benchmarking of quantum simulations for time-independent Hamiltonians.  Such a simulation can be stated as follows: given an initial state $\ket{\psi_0}$ and a Hamiltonian $\mf{H}$, evaluate the quantum state at time $t$ according to $\ket{\psi(t)}=\exp(-\I t \mf{H})\ket{\psi_0}$. Hamiltonian simulation is of immense importance in characterizing quantum dynamics for a diverse range of systems and situations in quantum physics, chemistry and materials science. Simulation of one quantum Hamiltonian by another quantum system was also one of the motivations of Feynman's 1982 proposal for design of quantum computers~\cite{Feynman1982}. Hamiltonian simulation is also used as a subroutine in numerous other quantum algorithms, such as quantum phase estimation~\cite{Kitaev1995} and solving linear systems of equations~\cite{HarrowHassidimLloyd2009}. 

Following the conceptualization of a universal quantum simulator using a Trotter decomposition of the time evolution operator $e^{-it\mf{H}}$ \cite{Lloyd1996}, a number of quantum algorithms for Hamiltonian simulation have been proposed~\cite{BerryAhokasCleveEtAl2007,BerryChildsCleveEtAl2015,LowChuang2017,LowWiebe2019,Campbell2019}.  Detailed assessment of these algorithms, with continued improvement of theoretical error bounds, has since emerged as a very active area of research~\cite{BerryChilds2012,BerryCleveGharibian2014,BerryChildsKothari2015,ChildsMaslovNamEtAl2018,ChildsOstranderSu2019,Low2019,ChildsSu2019,ChildsSuTranEtAl2020,ChenHuangKuengEtAl2020,SahinogluSomma2020,AnFangLin2021,SuBerryWiebeEtAl2021}.
In this context, one of the most significant developments in recent years is the quantum signal processing (QSP) method \cite{LowChuang2017}, and its generalization, the quantum singular value transformation (QSVT) method~\cite{GilyenSuLowEtAl2019}.
For sparse Hamiltonian simulation, the query complexity of QSVT matches the complexity lower bound with respect to all parameters~\cite{LowChuang2017,GilyenSuLowEtAl2019}. The QSVT method also enjoys another advantage, namely that the quantum circuit is relatively simple, and requires very few ancilla qubits.
QSVT allows one to use essentially the same parameterized quantum circuit to perform a wide range of useful computational tasks, including Hamiltonian simulation~\cite{DongMengWhaleyEtAl2021}, solution of linear systems~\cite{GilyenSuLowEtAl2019,LinTong2020,TongAnWiebeEtAl2020}, and finding eigenstates of quantum Hamiltonians~\cite{LinTong2020a}.
In this sense, it provides a `grand unification' of a large class of known quantum algorithms~\cite{MartynRossiTanEtAl2021}.

Despite these advantages, QSVT is generally not viewed as a suitable technique for near-term quantum devices today. 
This is largely because these techniques rely on an input model called `block encoding' which views the Hamiltonian $\mf{H}$ as a submatrix of an enlarged unitary matrix $U_{\mf{H}}$.
For Hamiltonians arising from realistic applications (e.g., linear combination of products of Pauli or fermionic operators, and sparse matrices in general), the construction of $U_{\mf{H}}$ often involves multiple ancilla qubits and multi-qubit control gates. Taken together, these requirements can make QSVT very difficult to implement with high fidelity and to date there has been no QSVT based Hamiltonian simulation on realistic devices.

\begin{figure*}
\centering
\includegraphics[width=\textwidth]{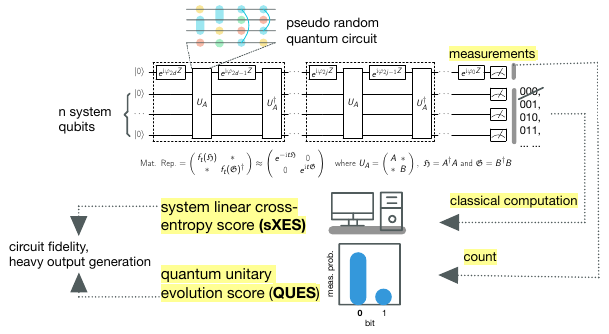}
\caption{{Illustration of the minimal quantum singular value transformation (mQSVT) circuit for the Hamiltonian simulation benchmark. 
The overall circuit implements a complex matrix polynomial $f_t(\mf{H})$ of degree $d$ on the Hamiltonian $\mf{H}$ that is defined in terms of a pseudo random quantum circuit $U_A$.
The circuit acts on $n+1$ qubits, consisting of $n$ system qubits and $1$ ancilla qubit.
After measuring the top ancilla qubit and post-selecting on the $0$ outcome of this, the action on the bottom $n$ system qubits accurately approximates $\exp(-\I t\mf{H})\ket{0^n}$. }} 
\label{fig:qsvt_circuit}
\end{figure*}

\vspace{1em}
\noindent {\large \textbf{Results}}\\
\noindent\textbf{Overview}\\

In this work we remedy this situation by identifying and demonstrating an application for QSVT on near term quantum devices that allows benchmarking of Hamiltonian simulation for a class of Hamiltonians that are relevant to recent efforts to demonstrate supremacy of quantum computation over classical computation~\cite{AruteAryaBabbushEtAl2019}.  This is the class of random Hamiltonians generated from block encoding of random unitary operators that correspond to random unitary circuits. We show that for this class of Hamiltonians it is possible to formulate a simple metric, called the quantum unitary evolution score (QUES), for the success of quantum unitary evolution. This metric is the primary output from the Hamiltonian simulation benchmark, and is directly related to the circuit fidelity. This allows verification of Hamiltonian simulation on near-term quantum devices without any need for classical computation, and the approach can be scaled to a large number of qubits.

The main result of this paper is a \emph{very simple} quantum circuit (\cref{fig:qsvt_circuit}), called the minimal QSVT (mQSVT) circuit. 
With proper parameterization, the mQSVT circuit is able to 
propagate a certain class of random Hamiltonians $\mf{H}$ to any given target accuracy. In fact, we argue that the mQSVT circuit is not only \emph{the simplest} quantum circuit for carrying out a QSVT based Hamiltonian simulation, but that it is actually the simplest possible circuit for all tasks based on QSVT.
Here $\mf{H}$ is not a Hamiltonian corresponding to a given physical system, but a random Hamiltonian generated using a simple random unitary circuit, called a Hermitian random circuit block encoded matrix (H-RACBEM)~\cite{DongLin2021}. However, for the purpose of benchmarking the capability of a quantum device to perform arbitrary Hamiltonian simulations, averaging over a distribution of the underlying arbitrary Hamiltonians is precisely what is required to generate a holistic benchmark protocol that samples from all possible instantiations.

The quantum circuit in \cref{fig:qsvt_circuit} consists of two components: an arbitrary random unitary matrix $U_A$ that implicitly defines the Hamiltonian $\mf{H}$, together with its Hermitian conjugate $U_A^{\dag}$ and a series of $R_z$ gates with carefully chosen phase factors $\{\varphi_{i}\}_{i=0}^{2d}$ (see \cref{sec:opt_phase}). 
The mQSVT circuit makes $d$ queries to $U_A$ and $U_A^\dagger$, two of which are shown explicitly in \cref{fig:qsvt_circuit}.
For an $n$-qubit matrix $\mf{H}$, the total number of qubits needed is always $n+1$, i.e., only $1$ ancilla qubit, hereafter referred to as the signal qubit, is required. This is even smaller than the simplest QSVT circuit~\cite{LowChuang2017}, which requires at least 2 ancilla qubits. 
However, more important than the reduction of the number of qubits is the fact that \cref{fig:qsvt_circuit} removes all two-qubit and multi-qubit gates outside of the unitary $U_A$. 
This means that one can choose any convenient entangling two-qubit gate (e.g.. CZ, CNOT, $\sqrt{\mathrm{iSWAP}}$) and any coupling map that is native to a quantum device to construct the random $U_A$. Combining this with the sequence of single qubit $R_z$ gates then makes the resulting benchmarking quantum circuit of \cref{fig:qsvt_circuit} readily executable.

\vspace{1em}
\noindent\textbf{Quantum unitary evolution score (QUES)}\\

\cref{fig:qsvt_circuit} implements $f_t(\mf{H})\ket{0^n}$ on the system qubits, where $f_t(\mf{H})$ is a matrix polynomial (see~\SP \ for details), with approximation error in the operator norm upper bounded by $\norm{f_t(\mf{H}) - e^{-\I t \mf{H}}}_2 \le \epsilon$. 
    Therefore in the absence of quantum errors, after applying the circuit to the input state $\ket{0^{n+1}}$, the probability $P_t(U_A) := \norm{f_t(\mf{H})\ket{0^n}}^2$ of measuring the top ancilla qubit with outcome $0$ will be close to $1$, indicating that the underlying Hamiltonian evolution is unitary. 

    From now on, we will primarily consider mQSVT circuits with a fixed set of phase factors $\{\varphi_{i}\}$ and hence fixed simulation time $t$. 
For notational simplicity, we will drop the $t$-dependence in quantities such as $P_t(U_A)$, unless specified otherwise. 

On a real quantum device, the probability $P(U_A)$ should be replaced by $P_\expl(U_A)$, which is the experimentally measured probability. 
We define the quantum unitary evolution score (QUES) by
\begin{equation}\label{eqn:ques}
    \text{QUES}(n,d) := \expt{P_\expl(U_A)},
\end{equation}
where the expectation is taken over the ensemble of random quantum circuit instances $U_A$. 
The deviation of QUES from $1$ then measures the average performance of the quantum computer under a Hamiltonian simulation task.

There is no unique prescription for constructing random quantum circuits. To fix the choice of $U_A$, we employ here the model random quantum circuit construction used to analyze the concept of quantum volume in~\cite{CrossBishopSheldonEtAl2019}. Here, 
given a number of qubits $n$, $U_A$ is constructed to contain $n$ layers, each consisting of a random permutation of the qubit labels followed by random two-qubit gates between the $n$ qubits. Given this construction, the QUES in \cref{eqn:ques} will then depend only on $n$ and $d$, and the overall depth of the circuit is approximately $2d$ times the circuit depth of $U_A$. Note that given the basic quantum gate set of a particular quantum device, alternative constructions of $U_A$ using  random choices of specific one- and two-qubit gates are possible.

\cref{fig:ques-ibmq} shows the results of computing the QUES  across $8$ different IBM Q quantum devices~\footnote{https://quantum-computing.ibm.com}, each having $5$ qubits and one of three distinct coupling maps (panel b).   When the number of qubits $n\le 3$, the QUES on all devices is relatively high ($\gtrsim 0.7$) but it decreases sharply for $n\ge 4$.
In contrast, the QUES decreases only relatively mildly as $d$ increases.
This is particularly noticeable for $n=2$, which may indicate that the quantum circuit transpiler provided by the IBM Q may be particularly effective for this device with very small qubit number.
We emphasize that compared to generic benchmark measures such as the quantum volume, the QUES is specific to the computational task of the Hamiltonian simulation, and any information specific to this is not diluted by additional averaging over output distributions from other computational tasks.
In particular, we find that even for quantum devices with relatively small quantum volume (8-QV), the performance in terms of QUES is only mildly worse than for those with a larger quantum volume (32-QV).

\begin{figure*}[htbp]
    \centering
    \includegraphics[width=0.975\textwidth]{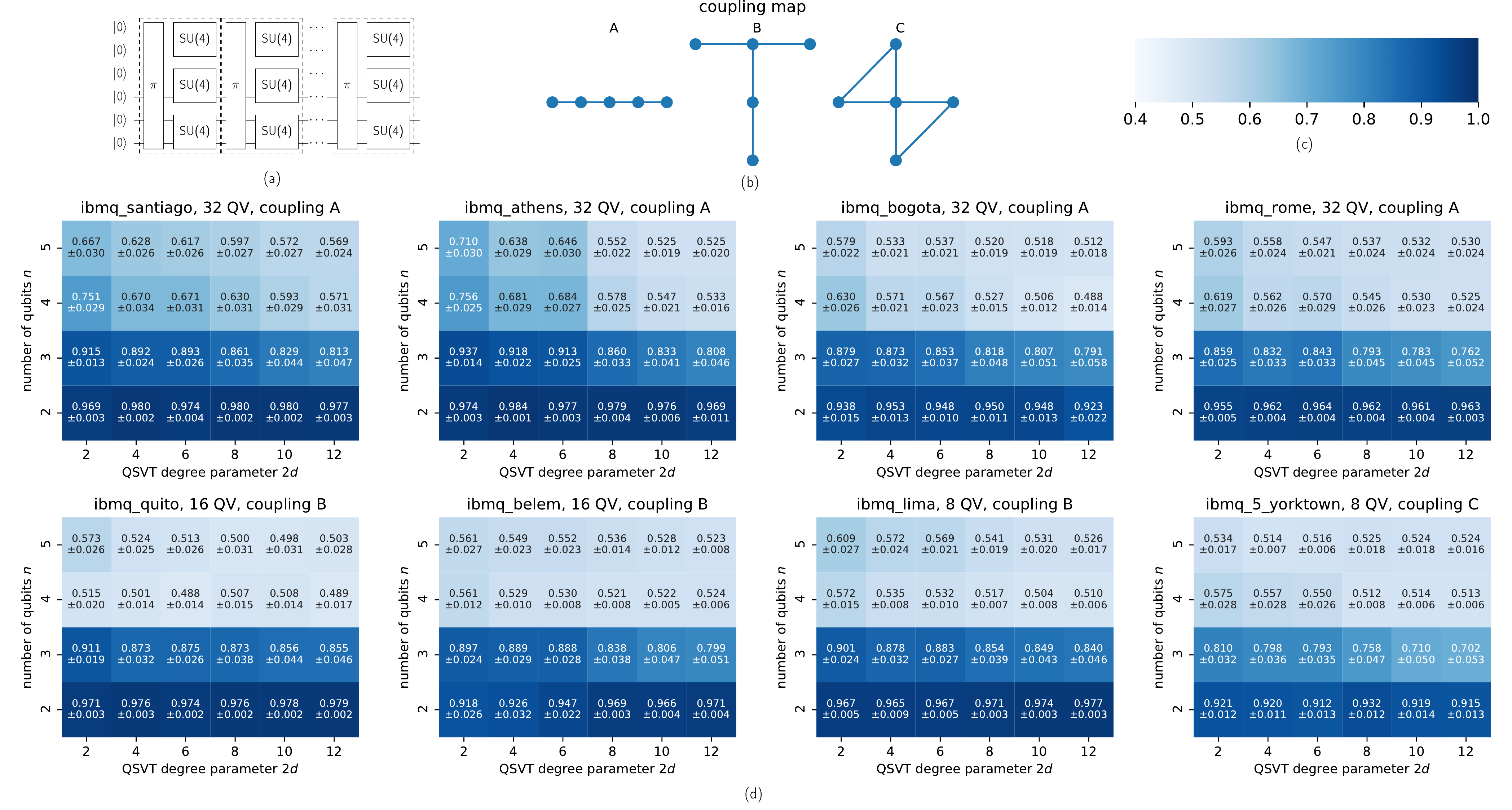}
    \caption{Quantum unitary evolution score (QUES) of the $5$-qubit quantum devices provided by the IBM Q platform~\cite{IBMQ}. (a) Visualization of the quantum circuit $U_A$ used in computing QUES. When the number of qubits is $n$, there are $n$ layers of the dashed boxes consists of the random permutation of the qubits labels followed by random two-qubit gates. After calling the transpiler, the circuit $U_A$ is decomposed with respect to the basic gate set $\Gamma = \{\text{Rz}, \sqrt{\text{X}}, \text{X}, \text{CNOT}\}$ and the coupling map which indicates the available qubit pairs on which CNOT can act. (b) Layouts of coupling maps. (c) Color bar of the heatmap. (d) Each heatmap displays the benchmarking result of a specific quantum device, with the title showing the name of the device, its quantum volume, and its coupling map. Each QUES is estimated from $50$ circuit instances. Each circuit instance is measured with $1,000$ measurement shots. The number displayed in each heatmap is the QUES value and its $95\%$ confidence interval.}
    \label{fig:ques-ibmq}
\end{figure*}

\vspace{1em}
\noindent\textbf{Circuit fidelity and system linear cross-entropy score (sXES)}\\

The quality of a noisy implementation of a quantum circuit is often characterized by the circuit fidelity.
Loosely speaking, the output quantum state of a noisy circuit can be characterized as a convex combination of the correct result and the result obtained under noise, i.e., $\text{`output'} = \alpha \times \text{`correct result'} + (1-\alpha) \times \text{`noise'}$, where $0\le \alpha\le 1$ is the circuit fidelity.
Let $p(U_A,x)$ be the noiseless bitstring probability of measuring the mQSVT circuit with outcome $0$ in the ancilla qubit and an $n$-bit binary string $x$ in the $n$ system qubits. Let $p_\expl(U_A, x)$ be the corresponding experimental bitstring probability, which can be estimated from the frequency of occurrence of the bitstring $0x$ in the measurement outcomes. 
Since the random circuit $U_A$ is approximately drawn from the Haar measure, we make analogous assumptions to those in~\cite{BoixoIsakovSmelyanskiyEtAl2018,AruteAryaBabbushEtAl2019}, and assume the following global depolarized error model:
\begin{equation}\label{eqn:succ-prob-fidelity}
    p_\expl(U_A, x) = \alpha p(U_A,x) + \frac{1-\alpha}{2^{n+1}}.
\end{equation}
We discuss the justification and potential generalization of such an error model in \cref{sec:error_model}.

Under the global depolarized error model, we now analyze the effect of noise on the circuit and show how measuring the QUES allows the circuit fidelity to be extracted. 
Prior work has made use of a combination of quantum and classical computation to obtain the circuit fidelity $\alpha$.  Such analysis relies on the possibility of evaluating the noiseless bitstring probability $p(U_A,x)$ classically, given $U_A$ and $x$,  e.g., via tensor network contraction~\cite{VillalongaLyakhBoixoEtAl2020}. This enabled the estimation of $\alpha$ from measurements of cross-entropy, referred to as XEB in this setting~\cite{BoixoIsakovSmelyanskiyEtAl2018,AruteAryaBabbushEtAl2019}. 
We adapt this approach to the Hamiltonian simulation problem by defining a system linear cross-entropy score (sXES):
\begin{equation}
\text{sXES}(U_A) := \sum_{x\ne 0^n} p(U_A, x) p_\expl(U_A,x).
\label{eqn:sXES}
\end{equation}
The prefix `system' is added because the ancilla qubit is fixed to be the $\ket{0}$ state in the definition of $p(U_A, x),p_\expl(U_A,x)$, and the $\ket{x}$ state belongs to the system register. 
In order to connect to the problem of generating heavy weight samples later, our definition of sXES excludes the bitstring $0^n$. 
This is necessary also since the statistical properties of the bitstring $0^n$ are different from those of the bitstrings in the system register.
Taking the expectation with respect to the distribution of $U_A$, and rearranging \cref{eqn:succ-prob-fidelity} then gives an expression for the circuit fidelity:
\begin{equation}\label{eqn:circuit-fidelity-XES}
    \alpha = \frac{\expt{\text{sXES}(U_A)} - \frac{1}{2^{n+1}}\expt{\sum_{x\ne 0^n} p(U_A, x)}}{\expt{\sum_{x\ne 0^n} p(U_A,x)^2} - \frac{1}{2^{n+1}}\expt{\sum_{x\ne 0^n} p(U_A, x)}}.
\end{equation}
This expression holds for any ensemble of random matrices, and relies only on the assumption that the noise model is depolarizing.  

        Once the probability distribution of $U_A$ is specified (e.g., the Haar measure~\cite{MehtaRM2004}), the only term in $\alpha$ that requires a quantum computation is $\text{sXES}(U_A)$, and all other terms in Eq.~\eqref{eqn:circuit-fidelity-XES} can be evaluated classically.
 However, evaluation of the right-hand side of \cref{eqn:circuit-fidelity-XES} often requires a significant amount of classical computation when $n$ becomes large~\cite{BoixoIsakovSmelyanskiyEtAl2018}.

\vspace{1em}
\noindent\textbf{Inferring circuit fidelity from QUES}\\

Based on the discussion so far, it might seem surprising that an alternative, very good approximation to the circuit fidelity can readily be obtained  from the QUES metric in \cref{eqn:ques}. This is arrived at by first defining $P_\expl(U_A) = \sum_x p_\expl(U_A, x)$, i.e., the average over all possible output bit strings $x$ of the probability of measuring a given bit string as outcome of the action of $U_A$ on the input state $\ket{0^{n+1}}$. Then summing both sides of \cref{eqn:succ-prob-fidelity} with respect to all bit strings $x$, further taking the expectation value of both sides over all possible $U_A$ yields a fidelity estimate $\alpha_\text{QUES}$ that can be obtained directly from the measured QUES value, namely
\begin{equation}\label{eqn:circuit-fidelity-QUES}
    \alpha_\text{QUES} = 2 \times \text{QUES} - 1.
\end{equation}
The approximation error $\epsilon$ is defined as the maximal error for simulating a bounded Hamiltonian using the mQSVT circuit, namely $\epsilon := \max_{\norm{\mf{H}}_2 \le 1} \norm{f_t(\mf{H}) - e^{-\I t \mf{H}}}_2$. It determines the extent of deviation of $\alpha_\text{QUES}$ from $\alpha$. Specifically, under the globally depolarized noise model, we have the following bound (\cref{sec:ques_fidelity})
\begin{equation}
\abs{\alpha_\text{QUES} - \alpha} \le 16 \epsilon+\Or(\epsilon^2).
\label{eqn:alpha_relation}
\end{equation}
Here the error bound is derived without including the Monte Carlo measurement error due to the finite number of measurement shots. 
The analysis of the resulting statistical error is given in \cref{sec:meas_error}.

It is evident that, unlike \cref{eqn:circuit-fidelity-XES}, there is no classical overhead for evaluating $\alpha_\text{QUES}$ for any $n$. 
Since the circuit fidelity $\alpha$ should be non-negative, combining \cref{eqn:circuit-fidelity-QUES} and \cref{eqn:alpha_relation} also indicates that under the assumption of the depolarizing noise model, we have $\text{QUES}\ge 0.5-8\epsilon+\Or(\epsilon^2)$.

To numerically verify the relation between QUES and circuit fidelity, we make use of the digital error model of \cite{BoixoIsakovSmelyanskiyEtAl2018} in which each quantum gate in the circuit is subject to a depolarizing error channel with a certain error rate.  We test the resulting noisy quantum circuit with different two-qubit gate error rates $r_2$ and set the one-qubit gate error rate to $r_1 = r_2/10$. We also discard the rotation gate with phase factor $\varphi_{2d}$, since this just adds a global phase to the exact Hamiltonian simulation. Then, given $U_A$ with a total of $g_1$ one-qubit gates and $g_2$ two-qubit gates, the reference value of the circuit fidelity can be set to $\alpha_\text{ref} := (1-r_1)^{2d(g_1+1)} (1-r_2)^{2d g_2}$ \cite{AruteAryaBabbushEtAl2019,BoixoIsakovSmelyanskiyEtAl2018}. We assume $U_A$ is Haar-distributed (numerically verified in \cref{sec:converge-to-Haar}) to simplify the computation of classical expectations.

\begin{figure*}[htbp]
    \centering
    \includegraphics[width=0.975\textwidth]{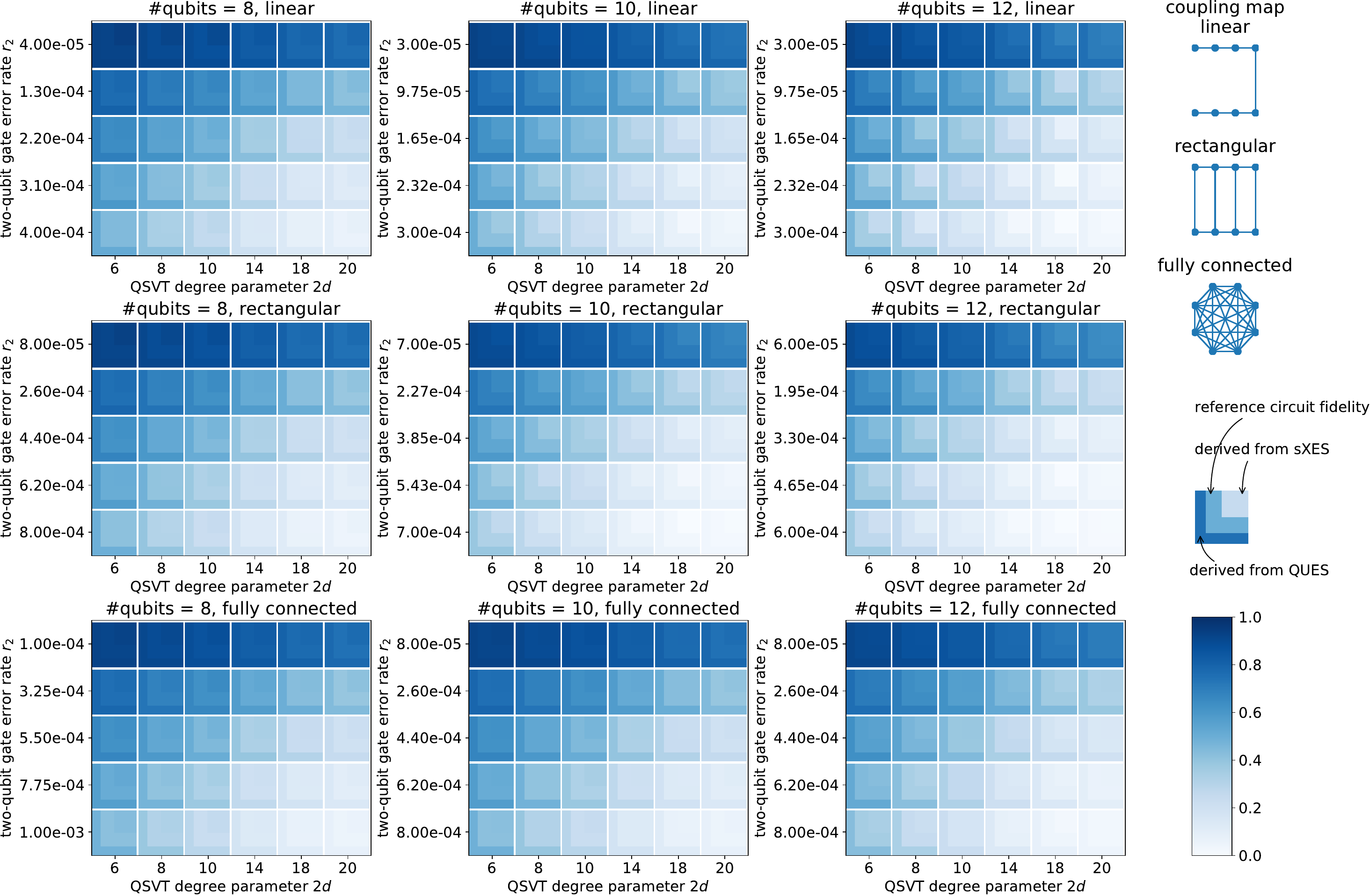}
    \caption{Circuit fidelity estimated from the quantum Hamiltonian simulation benchmark. Colored grids represent the circuit fidelity estimated from $\sim 100$ circuit repetitions. The benchmarking is performed for circuits with a range of number of system qubits, having also variable types of couplings and a range of error parameters. The depth of the random circuit instances is set to the convergent depth deduced from the convergence to Haar measure (see~\cref{sec:converge-to-Haar}).  The right column contains graphical depictions of the coupling maps, the layout of each grid, and the color bar.}
    \label{fig:hsbenchmark}
\end{figure*}

\cref{fig:hsbenchmark} summarizes the estimated circuit fidelity for random quantum circuits with different depth parameter $d$, variable coupling maps, and a range of error parameters. In all cases, we find that the derived circuit fidelity from QUES ($\alpha_\text{QUES}$), the circuit fidelity $\alpha$ obtained from sXES, and the reference value $\alpha_\text{ref}$ are generally consistent with each other. Numerical results also show that $\alpha_\text{QUES}$ exhibits a trend that slightly overestimates the value of the fidelity $\alpha$ (see \cref{tab:HSBenchmark} for numerical values of the fidelities). We also see that for a given set of error rates $r_1, r_2$, the circuits with highest connectivity show the best performance. This is because random circuits on these architectures converge faster to the Haar measure, which reduces the circuit depth  (see \cref{sec:converge-to-Haar}).

In the next two subsections we show how to assess and evaluate whether the Hamiltonian simulation with the mQSVT circuit can be a classically hard task.  We first define the analog of XHOG for Hamiltonian simulation, which we refer to as sXHOG, and give conditions for the hardness of this.  We then show that potential classical hardness can be inferred directly from the value of the circuit fidelity obtained from the QUES, i.e. from $\alpha_\text{QUES}$.

\vspace{1em}
\noindent\textbf{Classical hardness and system linear cross-entropy heavy output generation (sXHOG) }\\

The complexity-theoretic foundation of the Google claim of `quantum supremacy' in~\cite{AruteAryaBabbushEtAl2019} is based on a computational task called linear cross-entropy heavy output generation (XHOG) with Haar-distributed unitaries~\cite{AaronsonGunn2019,AaronsonChen2016,BoixoIsakovSmelyanskiyEtAl2018,AruteAryaBabbushEtAl2019}. Specifically, given a number $b>1$ and a random $n$-qubit unitary $U$, the task is to generate $k$ nonzero bitstrings $x_1, x_2, \cdots, x_k \in \{0,1\}^n$ such that $\frac{1}{k} \sum_{j=1}^k q(U, x_j) \ge b \times 2^{-n}$, where $q(U,x)=|\braket{x|U|0^n}|^2$. Here we use $U$ without the subscript to distinguish the XHOG problem and the sXHOG problem which will be defined later.
For $k$ randomly generated bitstrings, we expect that $\frac{1}{k} \sum_{j=1}^k q(U, x_j)\approx 2^{-n}$.
Therefore any value $b>1$ will correspond to a `heavy weight' output. 
When $k$ is large enough, successful solution of the XHOG problem is considered to be classically hard for every value $b>1$~\cite{AaronsonChen2016,AaronsonGunn2019}.
This holds for every circuit fidelity estimate $\alpha>0$ obtained from the XEB metric, leading to the claim of supremacy in \cite{AruteAryaBabbushEtAl2019} based on extraction of a value $\alpha\approx 0.002$ from the experiments. 

For the Hamiltonian simulation benchmark, we can define an analogous linear cross-entropy heavy output generation problem for the $n$ system qubits.  Note that the heavy weight samples are now defined only for the system qubits. We shall refer to this heavy output generation problem for Hamiltonian simulation as the sXHOG problem, to emphasize this important feature and the difference from the standard XHOG problem.  
Specifically, given a number $b > 1$, a Hamiltonian simulation benchmark circuit with sufficiently small approximation error $\epsilon$, and a random $(n+1)$-qubit unitary $U_A$ defining a random Hamiltonian on the $n$ qubits, the task is to generate $k$ nonzero bitstrings $x_1, x_2, \cdots, x_k \in \{0,1\}^n\backslash\{0^n\}$ such that $\frac{1}{k} \sum_{j=1}^k p(U_A, x_j) \ge b \times 2^{-n}$.  
Now for the case of Hamiltonian simulation, $p(U_A,x)=\Or(2^{-n})$ for any $x\ne 0^n$ at all $t$, but $p(U_A,0^n)$ can be much larger (for more details see \cref{fig:mc-justify-topt}(b) in \cref{sec:topt_derive}). The state $0^n$ is then by definition `heavy' and we must therefore exclude this from the measure in order to avoid a trivial outcome. This is what distinguishes the sXHOG problem from the original XHOG problem.

The potential classical hardness of the XHOG problem is justified by a reduction to a complexity-theoretic conjecture, called linear cross-entropy quantum threshold assumption (XQUATH) \cite{AaronsonGunn2019}. 
For completeness, we give a similar variant of the reduction of sHOG problem to a conjecture named system linear cross-entropy quantum threshold assumption (sXQUATH) in \cref{thm:reduction-XHOG-XQUATH} of \cref{sec:hardness_sxhog}. The concept of sXQUATH directly parallelizes that of XQUATH, with a similar restriction as above to exclude the output bit string $0^n$ (for more details see \cref{sec:hardness_sxhog}). Similar to the construction in Ref. \cite{AaronsonGunn2019}, the classically efficient solution to sXHOG problem yields a violation to sXQUATH, which assumes that $p(U_A,x)$ for $x\ne 0^n$ cannot be efficiently estimated on classical computers to sufficient precision.

\vspace{1em}
\noindent\textbf{Inferring classical hardness from QUES}\\

In order to decide whether a noisy implementation of the Hamiltonian simulation benchmark is potentially in the classically hard regime, we need to establish whether or not the sXHOG problem can be solved for $b>1$.  

Under the assumption that $U_A$ is drawn from the Haar measure, and that the approximation error $\epsilon$ of the mQSVT circuit is sufficiently small, we derive the following relation between $b$ and the circuit fidelity $\alpha$:
\begin{equation}
b=1+\frac{\gamma(\alpha-\alpha^*)}{\alpha+1}.
\label{eqn:b_fidelity_simplify}
\end{equation}
Here $\alpha^*$ is a fidelity threshold (not the complex conjugation of $\alpha$) and $\gamma$ a constant.  Explicit expressions for the threshold value $\alpha^*$ and  the constant $\gamma$ are given in \cref{sec:fidelity_sxhog}.  Both quantities are independent of the circuit fidelity $\alpha$ and depend only on the number of system qubits $n$ and the simulation time $t$.
\cref{eqn:b_fidelity_simplify} thus shows that when $\gamma>0$ and $\alpha>\alpha^*$, we will have $b>1$ so that the sXHOG problem solved by the mQSVT circuit might be classically hard. This is qualitatively different from the situation for XEB experiments, for which every $\alpha>0$ implies $b>1$~\cite{AruteAryaBabbushEtAl2019}.

Using the relation between QUES and $\alpha$ in \cref{eqn:circuit-fidelity-QUES,eqn:alpha_relation}, and assuming that $\epsilon$ is sufficiently small, we immediately arrive at the conclusion that when 
\begin{equation}
\text{QUES}\ge (1+\alpha^*)/2, \quad \gamma>0,
\label{eqn:ques_condition}
\end{equation}
the corresponding sXHOG problem might be classically hard for a sufficiently large value of $n$. This is a surprising result, since as noted above, the estimation of QUES does not require  intensive classical computation.  In fact it is not even necessary to actually generate any heavy weight samples - instead we just need to measure the value of QUES, Eq.~\eqref{eqn:ques}, which is readily done by repeatedly running the circuit in~\cref{fig:qsvt_circuit} with random circuit parameters as described above. Of course, should one wish to actually solve the sXHOG problem itself, the heavy weight samples would need to be generated using a quantum computer and intensive classical computation for computation of $\frac{1}{k} \sum_{j=1}^k p(U_A, x_j)$ would then also be required. But in order to demonstrate the potential regime of classical hardness for Hamiltonian simulation, i.e., the minimal values of $n$ and $d$ to reach this regime, this is not required.

To further investigate  the implications of \cref{eqn:b_fidelity_simplify}, we now explicitly indicate the time dependence of all quantities (i.e., we employ the notation $\gamma\to \gamma_t,\alpha^*\to \alpha^*_t$). In \cref{fig:supremacy-region-HS} we plot the values of $\gamma_t,\alpha^*_t$ according to the expressions given in \cref{sec:fidelity_sxhog} as a function of the simulation time $t$, for $n=4, 8, 12$ qubits. \cref{fig:supremacy-region-HS} shows that $\gamma_t>0$ for all $t$, so then we only need to determine whether it is possible to have fidelity $\alpha\ge \alpha^*_t$.  
It is evident from \cref{fig:supremacy-region-HS} that both $\alpha^*_t$ (panel (a)) and $\gamma_t$ (panel (b)) show oscillatory behavior. We now analyze this behavior to identify an optimal time at which the potential classical hardness of Hamiltonian simulation in this random Hamiltonian setting can be demonstrated for a sufficiently large number of qubits $n$.

For very short times, i.e., when $t\approx 0$, we have $\alpha^*_t>1$. This means that we cannot have $b>1$ for any value of the circuit fidelity $0\le \alpha\le 1$.
To see why this is the case, consider the limit $t=0$. Here $p_t(U_A,0^n)=1$, and $p_t(U_A,x)=0$ for any $x\ne 0^n$. 
By continuity, when $t$ is very small, the magnitude of $p_t(U_A,x)$ for most bitstrings $x\ne 0^n$ is still very small and cannot reach the heavy output regime. \cref{fig:supremacy-region-HS}(a) also shows that there is a critical simulation time $t^{\text{thr}}\approx 2.26$, for which $\alpha^*_t<1$ for any $t>t^{\text{thr}}$.

When $t>t^{\text{thr}}$, ideally we would like to have $\alpha^*_t \approx 0$, so that a very low experimental circuit fidelity $\alpha$ is sufficient to reach the heavy output regime. 
To this end we investigate what happens at the vanishing circuit fidelity, i.e., $\alpha=0$.
Detailed analysis shows that in the large $n$ limit, we have $\gamma_t\alpha^*_t=\expt{p_t(U_A,0^n)}$, and
\cref{eqn:b_fidelity_simplify} can be simplified as (see \cref{sec:fidelity_sxhog})
\begin{equation}
b|_{\alpha = 0}=1-\expt{p_t(U_A,0^n)},
\label{eqn:b_zeroalpha}
\end{equation} 
where the expectation value is taken with respect to the random unitaries $U_A$ as before.
Thus when the expectation value is positive, i.e., $\expt{p_t(U_A,0^n)} > 0$, in the large $n$ limit we have $b|_{\alpha = 0}<1$ and the task should not be classically hard.
Moreover, since $b$ is a continuous function of $\alpha$, even if we now have finite circuit fidelity $\alpha$, when this is small enough we can still find $b<1$. This provides an alternative explanation of \cref{eqn:b_fidelity_simplify}, namely, that the circuit fidelity $\alpha$ needs to be larger than the finite positive threshold value $\alpha^*_t>0$ for \emph{most} values of $t>t^{\text{thr}}$.

As a result of these considerations, when $n$ is large enough, it is important to focus  on the regimes where the expectation value $\expt{p_t(U_A,0^n)} \approx 0$, which from \cref{eqn:b_fidelity_simplify} implies that the threshold fidelity $\alpha^*_t\approx 0$. The numerical results shown in \cref{fig:supremacy-region-HS} indicate that this can happen in two different scenarios.
The first is when the simulation time $t\to \infty$ (see the analytic justification of this statement in \cref{sec:asymptotic_larget}). Of course this requires a very large circuit depth and is a physically `trivial' limit that is impractical on near-term quantum devices. The second scenario, which is much more relevant in practice, is when $\alpha^*_t$ reaches its first minimum, which defines an optimal time $t=t^{\text{opt}}$. In the large $n$ limit, the value of $t^{\text{opt}}$ can be rationalized as the first node of the Bessel function $J_0(t/2)$ (see \cref{sec:topt_derive}). 
\cref{fig:supremacy-region-HS} (a) shows that for $t^{\text{opt}}\approx 4.81$, we already have $\expt{p_t(U_A,0^n)}\approx 0$ and $\alpha^*_t\approx 0$. 
Therefore simulating to the time $t=t^{\text{opt}}$ is highly desirable, since this is a relatively short time at which the Hamiltonian simulation benchmark is nevertheless now guaranteed to solve the sXHOG problem even for a very small circuit fidelity. 
Our numerical results indicate that the values of $t^*$ and $t^{\text{opt}}$ depend only weakly on $n$, and their values are nearly converged for $n$ as small as $12$. 
Therefore this value of $t^{\text{opt}}$ can be used in a future quantum simulation in the heavy output regime.

\begin{figure*}[htbp]
    \centering
    \includegraphics[width=.8\textwidth]{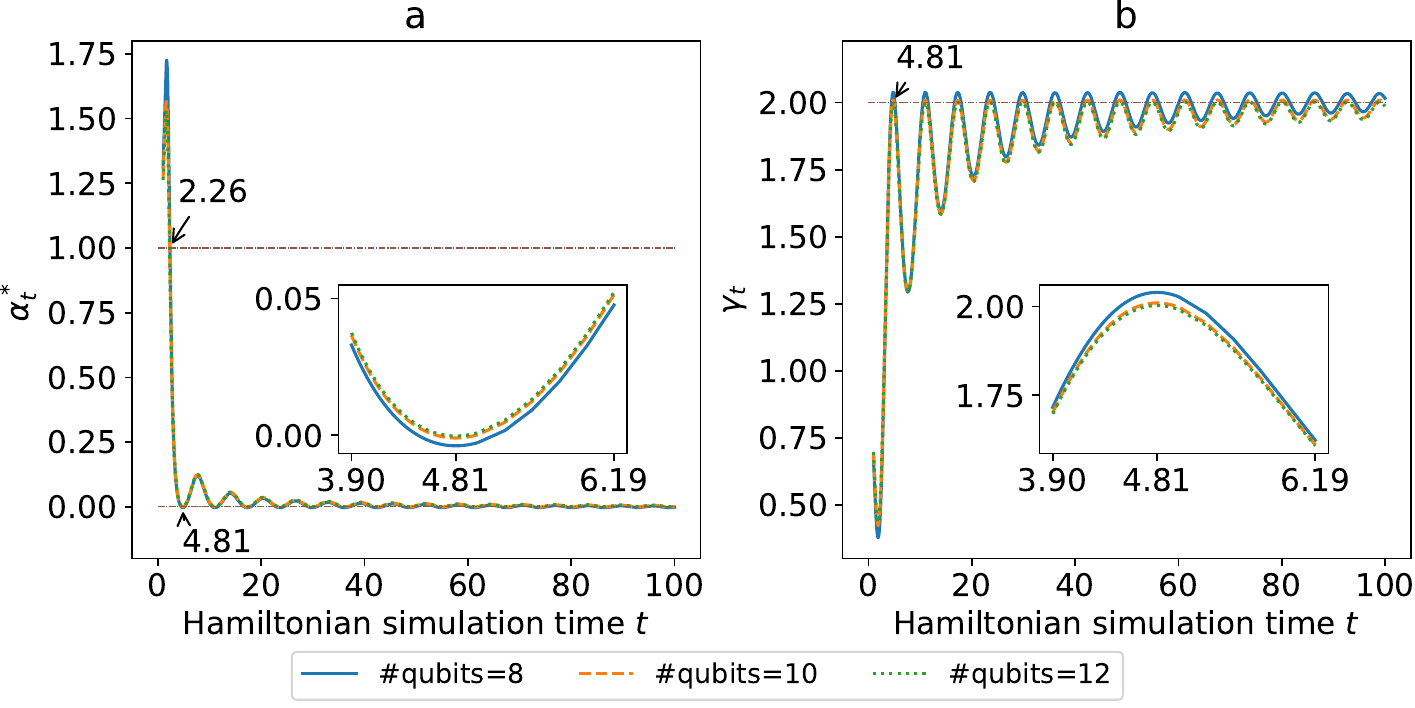}
    \caption{Quantities relevant to the system linear cross-entropy heavy output generation (sXHOG) problem, evaluated using the explicit expressions given in \cref{sec:fidelity_sxhog}.  (a) The threshold fidelity $\alpha^*_t$ as a function of Hamiltonian simulation time $t$. The upper value noted on the plot indicates the time value $t^{\text{thr}} \approx 2.26$ where $\alpha^*(t^{\text{thr}}) = 1$. The lower value noted on the plot indicates the regime at finite time $t^{\text{opt}} \approx 4.81$ with the first minimal value of threshold fidelity. (b) The parameter $\gamma_t$ as a function of Hamiltonian simulation time $t$. The value noted on the plot indicates the value $\gamma_t \approx 2$ at the optimal time $t^{\text{opt}} \approx 4.81$. The averages in (a) and (b) are estimated numerically from $\sim 100$ instances of the mQSVT circuit encoding random Hamiltonians drawn from the Haar measure. Insets in each panel show the behavior of $\alpha^*_t$ and $\gamma_t$ near the optimal time $t^{\text{opt}} \approx 4.81$.}
    \label{fig:supremacy-region-HS}
\end{figure*}

\vspace{1em}
\noindent {\large \textbf{Discussion}}\\

We have presented a quantum benchmark for Hamiltonian simulation on quantum computers.
The Hamiltonian simulation problem is solved using a minimal quantum singular value transformation (mQSVT) circuit. The primary output of the Hamiltonian simulation benchmark is a single number called QUES, which can be verified without any classical computation, even in the regime that is potentially hard for classical computation. 
Therefore the Hamiltonian simulation benchmark provides a scalable benchmark of the circuit fidelity under the global depolarized error model, and can be executed and verified on future quantum devices with a large number of qubits.

As the current quantum computing technologies advance, the possibility of implementing some error correction is improving~\cite{chen2021exponential}.  Here the highly structured mQSVT circuit provides useful indications of where best to implement error correction under limited resources for this.  Recall that the mQSVT circuit consists of a series of repetitions of a random circuit $U_A$ and its conjugate $U_A^\dagger$, interleaved with single-qubit Z rotation operators characterized by carefully selected phase factors. 
Thus given a specific random Hamiltonian block encoded in $U_A$, the time dependent evolution operator for this Hamiltonian is defined entirely in terms of the phase angles for the single-qubit Z rotation operators.
Since these phases should moreover be precisely determined, this suggests that on near-term quantum devices that may allow for some error correction but have overall limited resources, quantum error correction for these single-qubit rotations should be prioritized.

It is also useful to consider here the applicability of this Hamiltonian simulation approach to general Hamiltonians, i.e., not restricted to random Hamiltonians, on near-term quantum computers.  Unfortunately it appears that for current quantum technologies there is potentially a large gap between the feasible simulation of a H-RACBEM given in this work and that of a general Hamiltonian relevant to e.g., molecular or solid-state physics. 
The main reason is that the block encoding of most Hamiltonians of practical interest will involve significant numbers of ancilla qubits, as well as multi-qubit control gates, all of which are extremely expensive on near-term quantum devices.
In contrast to this general situation, the construction of H-RACBEM uses only whatever one-qubit and two-qubit gates are available for a given quantum device and is thus considerably easier.
Nevertheless, it is possible that undertaking Hamiltonian simulation with H-RACBEM may also yield interesting physical applications to the various settings in which quantum chaotic dynamics are relevant.
One immediate possibility in this direction is to use H-RACBEM to simulate the dynamics of quantum scrambling or quantum chaos in strongly interacting quantum systems.  Scrambling dynamics can be studied by simulating  out-of-time-order correlators (OTOCs) for effective Hamiltonians that can be defined implicitly in terms of a random circuit for time $t$ (see e.g.,~\cite{MiRoushanQuintanaEtAl2021}).
We note that one can easily perform a Hamiltonian simulation backward in time, merely by reversing the sign of $t$, so the mQSVT circuit for an OTOC at any time $t$ of a random Hamiltonian encoded in H-RACBEM can be readily constructed by adding local operators between forward and backward implementations of the mQSVT.  Evaluation of the circuit at different times $t$ can be implemented either by reevaluating the phase factors (which may required building a longer circuit depending on the accuracy required). The circuits can also be adapted to Hamiltonian simulation at finite temperatures and hence also to scrambling at finite temperatures.  From a theoretical perspective it would also be useful to explore to what extent the structure of the H-RACBEM influences the speed of scrambling~\cite{brown2012scrambling}.

Our theoretical analysis of the sensitivity of the Hamiltonian benchmarking scheme in this work was based on a fully depolarized noise model, which is often assumed to be a good model for superconducting qubits~\cite{AruteAryaBabbushEtAl2019}.  
In general the Pauli stochastic noise model on which is based may be biased or non-uniform across qubits.  In addition, thermal noise and coherent errors are important for some qubit architectures.  It will be useful to extend the current analysis to more general noise models, and some of these aspects are discussed in \cref{sec:error_model}.

Finally, we note that while this Hamiltonian simulation benchmark is restricted to the specific class of random Hamiltonians, it might also provide information relevant to more general Hamiltonian simulations. Efforts to analyze the complexity of analog Hamiltonian simulations have often focused on the relation of such simulations to classical sampling tasks~\cite{AaronsonArkhipov2011,BremnerMontanaroShepherd2016,haferkamp2020closing}, and are closely related to the cross-entropy analysis for sampling of random quantum circuits of~\cite{BoixoIsakovSmelyanskiyEtAl2018,AruteAryaBabbushEtAl2019}. As noted recently~\cite{haferkamp2020closing}, the classical hardness can be shown for certain classes of analog quantum Hamiltonian simulation~\cite{gao2017quantum,bermejo2018architectures}. Note that the potential classical hardness of the original XHOG problem corresponding to Google's supremacy experiment is justified by a reduction to a complexity-theoretic conjecture called XQUATH \cite{AaronsonGunn2019}.
However, a recent paper \cite{GaoKalinowskiChouEtAl2021} that appeared after submission of the current work has provided evidence that can refute XQUATH, at least for some classes of quantum circuits. Therefore it is possible that our sXQUATH assumption can be refuted on the same basis. 
It could be useful to explore generalizations of other classical sampling tasks to the QSVT setting, as was done here for the cross-entropy heavy output generation, to help guide the search for Hamiltonians whose simulation by QSVT can exhibit quantum advantages. Finally, the current approach of analysis  of alternative fidelity measures under Hamiltonian simulation using mQSVT may provide useful for analysis of recent fidelity based experimental studies of analog Hamiltonian simulations that followed the emergent random nature of a projected ensemble of states~\cite{choi2021emergent}.

\vspace{1em}\noindent {\large \textbf{Methods}}\\
\noindent {\textbf{Details of numerical simulations}}

All numerical tests are implemented in \textsf{python3.7} and \textsf{Qiskit} \cite{Qiskit}. Quantum circuits in \cref{fig:ques-ibmq} are optimized by the \textsf{transpiler} provided by \textsf{Qiskit} before being executed on a real quantum device. The number of measurements (shots) is fixed to be $1,000$ for the experiments on real quantum devices in \cref{fig:ques-ibmq}, and it is set to $1,000,000$ for those on classical simulators in \cref{fig:hsbenchmark}. The classical generation of Haar random unitaries in \cref{fig:supremacy-region-HS} is performed by QR factorization to random complex matrices with i.i.d. Gaussian entries according to the recipe in \cite{Mezzadri2006}.

\bigskip

\noindent{\large \textbf{Data availability}}\\
The experimental data that support the finding are available from the authors upon request.\\

\noindent{\large \textbf{Code availability}}\\
The codes that support the finding are available from the authors upon request.\\

\noindent {\large \textbf{Acknowledgments}}\\
This work was partially supported by 
the U.S. Department of Energy, Office of Science, National Quantum Information Science Research Centers, Quantum Systems Accelerator (BW, LL)
and a Google Quantum Research Award (YD,BW,LL),  by the Department of Energy under grant DE-SC0017867,  and by the Department of Energy under the Center for Advanced Mathematics for Energy Research Applications (CAMERA) program (LL). We thank Andr\'as Gily\'en, Xun Gao, Yunchao Liu, Murphy Niu, and Jiahao Yao for helpful discussions, and in particular thank Timothy Proctor for extended discussions on noise models.
\\

\noindent {\large \textbf{Author contributions}}\\
YD and LL designed the Hamiltonian simulation benchmark and proved its theoretical properties. YD, BW and LL designed the experiments. YD carried out classical simulations and IBM-Q experiments. All authors contributed to the discussion of results and writing of the manuscript.\\

\noindent{\large \textbf{Competing interests}}\\
\noindent The authors declare no competing interests.

\bibliographystyle{abbrvnat}
\bibliographystyle{unsrt}

\widetext
\clearpage

\setcounter{equation}{0}
\setcounter{figure}{0}
\setcounter{table}{0}
\renewcommand{\figurename}{Supplementary Figure}
\renewcommand{\tablename}{Supplementary Table}
\renewcommand{\thetable}{\arabic{table}}
\renewcommand{\thesection}{\SP\ \arabic{section}}

\begin{center}
    {\Large \bf \SP s}
\end{center}

\section{Notations}

We first introduce the definition of block encoding. 
Let $A\in \CC^{N\times N}$ be an $n$-qubit Hermitian matrix ($N=2^n$). If we can find an $(n+1)$-qubit unitary matrix $U_A$ such that ($*$ stands for a matrix block whose entries are not of interest)
\begin{equation}
U_A=\left(\begin{array}{cc}
{A} & {*} \\
{*} & {*}
\end{array}\right)
\label{eqn:block_encode_exact_matrix}
\end{equation}
holds, i.e. $A$ is the upper-left matrix block of $U_A$, then we may get access to the action of $A$ on an $n$-qubit state $\ket{\psi}$ via the unitary matrix $U_A$ by  
\[
U_A\ket{0}\ket{\psi}=\ket{0}(A\ket{\psi})+\ket{\perp},
\]
where $\ket{\perp}$ is an unnormalized $(n+1)$-qubit state not of interest and  satisfies $(\ket{0}\bra{0}\otimes I_n)\ket{\perp}=0$. 
Here we follow the row-major order convention. For instance, 
\[
\ket{0}\ket{\psi}\equiv \begin{pmatrix}
\psi\\
0^n
\end{pmatrix}, \quad 
\ket{1}\ket{\psi}\equiv \begin{pmatrix}
0^n\\
\psi
\end{pmatrix},
\]
and \cref{eqn:block_encode_exact_matrix} can also be written as $A=\left(\langle 0 | \otimes I_n\right) U_A \left( | 0 \rangle \otimes I_n \right)$. 

Clearly when the operator norm $\norm{A}_2$ is larger than $1$, $A$ cannot be encoded by any unitary $U_A$ as in \cref{eqn:block_encode_exact_matrix}. Generally if we can find $\alpha, \epsilon' \in \mathbb{R}_+$, and an $(m+n)$-qubit matrix $U_A$ such that
\begin{equation}
\norm{ A - \alpha \left(\langle 0^m | \otimes I_n\right) U_A \left( | 0^m \rangle \otimes I_n \right) }_2 \leq \epsilon',
\label{eqn:block_encoding}
\end{equation}
then $U_A$ is called an $(\alpha, m, \epsilon')$-block-encoding of $A$. Here $m$ is called the number of ancilla qubits for block encoding. 
We refer to \cite{GilyenSuLowEtAl2019} for more details on block encoding.
When the block encoding is exact with $\epsilon'=0$, $U_A$ is called an $(\alpha, m)$-block-encoding of $A$. The special case of the $(1,1)$-block-encoding may also be called a $1$-block-encoding.

In the \SP, for notational simplicity, we may use $U$ without a subscript to represent a $(n+1)$-qubit quantum circuit drawn at random from a certain probability distribution. Unless otherwise noted, $A$ denotes the upper-left $n$-qubit submatrix of $U$, i.e. $U$ is the $1$-block-encoding of $A$.  This matrix $A$ is also called a random circuit block encoded matrix (RACBEM), and $\mf{H} = A^\dagger A$ is corresponding Hermitian random circuit block encoded matrix (H-RACBEM)~\cite{DongLin2021}. 

We use $N = 2^n$ to represent the dimension of the Hilbert space of the system qubits, and $I_n$ to denote the $n$-qubit identity matrix. For a complex square matrix $A$ with singular value decomposition (SVD) $A = W \Sigma V^\dagger$, its singular value transformation through an even function $g$ is defined as $g^\svt(A) = V g(\Sigma) V^\dagger$. Here, the right triangle in the notation means only the right singular vectors $V$ are kept in the transformation. If we consider $\abs{A} := \sqrt{A^\dagger A} = V \Sigma V^\dagger$, then the singular value transformation of $A$ is equal to the eigenvalue transformation of the Hermitian matrix $\abs{A}$, namely, $g^\svt(A) = g(\abs{A})$. Furthermore, due to the even parity of $g$, there is a function $f$ so that $g(x) = f(x^2)$ and $g^\svt(A) = f(\abs{A}^2)=f(\mf{H})$.
In particular, when $g_t(x)$ is an even polynomial approximation to $s_t(x)=e^{-\I t x^2}$, we can define $g_t(x)=f_t(x^2)$. Hence $f_t(x)$ approximates $e^{-\I t x}$, and $g_t^\svt(A) = f_t(\mf{H})$ approximates the Hamiltonian evolution $e^{-\I t \mf{H}}$.

We use $\circqsvt_{f,U}$ to represent the minimal quantum singular value transformation (mQSVT) circuit in \cref{fig:qsvt_circuit}, which has only a single ancilla qubit, $m=1$. For any $n$-qubit input state $\ket{\psi}$, the mQSVT circuit performs the following transformation of the input quantum state,
\begin{equation*}
    \circqsvt_{f,U} \ket{0} \otimes \ket{\psi} = \ket{0} \otimes \left( g^\svt(A) \ket{\psi} \right) + \ket{1} \otimes \ket{\bot},
\end{equation*}
where $\ket{\bot} \in \CC^N$ is an unnormalized quantum state. 
In other words, $\circqsvt_{f,U}$ is the $1$-block-encoding of $f(\mf{H})\equiv g^\svt(A)$.
$\norm{A}_2 := \sigma_\maxl(A)$ is the operator norm of a matrix which is equal to its maximal singular value. $\norm{f}_\infty := \max_{x \in [-1,1]} \abs{f(x)}$ is the infinity norm of continuous functions on $[-1,1]$. $\expt{\cdot}$ stands for the average over the random matrix ensemble (most commonly, the ensemble of $U$). Both $\overline{z}$ and $z^*$ stands for the complex conjugate of a complex number $z$. For a complex polynomial $P(x) = \sum_i c_i x^i \in \CC[x]$, its complex conjugate as $P^*(x) = \sum_i c_i^* x^i$. For a matrix $A$, the transpose, Hermitian conjugate and complex conjugate are denoted by $A^{\top}$, $A^{\dag}$, $A^*$, respectively. Without otherwise noted, an $n$-bit binary string $x \in \{0,1\}^n$ is identified to its decimal representation. Specifically, when an $n$-bit binary string appears in the subscript of a matrix or a vector, it is identified to be its decimal representation  (we use a zero-based indexing). For example,  $A_{0^n, 1^n} := A_{0, 2^n-1}$.

\cref{tab:prob-symbol} summarizes the main notations used in the \SP.
In the context of Hamiltonian simulation, many quantities depend on the value of the simulation time $t$.
Such a $t$-dependence is usually added as a subscript such as $p_t(U,x)$.
Most of the discussion focuses on the simulation at a fixed time $t$.
Therefore when the context is clear, for simplicity we may drop the $t$ dependence.

\begin{table}[htbp]
    \centering
    \begin{tabular}{ll}
    \hline\hline
    Symbol & Definition\\\hline
    $\circqsvt_{f, U}$ & mQSVT circuit in \cref{fig:qsvt_circuit} implementing a $1$-block-encoding of $f(\mf{H})$\\\hline
    $A$ & Upper-left $n$-qubit submatrix of a $(n+1)$-qubit random unitary matrix $U$ \\\hline
    $\mf{H}$ & $A^{\dag}A$, also called a H-RACBEM\\\hline
    $s_t(x)$           & $e^{-\I t x^2}$\\\hline
    $g_t(x)$           & an even polynomial approximation to $s_t(x)$, also denoted by $P(x,\Phi)$ with phase factor $\Phi$\\\hline
    $f_t(x)$           & $g_t(x^2)$, which is a polynomial approximation to $e^{-\I t x}$\\\hline
    $\bP(\cdot)$     &  probability density function of random quantum circuits
    \\\hline
    $\bP_\expl(\cdot)$     &  probability density function associated with the noisy implementation of random quantum circuits\\\hline
    $p_j$     & \begin{tabular}{@{}l@{}} probability associated with the matrix element at \\the $0$-th column and the $j$-th row of a unitary matrix $V$, i.e., $p_j := \abs{V_{j0}}^2$\end{tabular}\\\hline
    $p_{ij}$    & \begin{tabular}{@{}l@{}} probability associated with the matrix element at\\ the $i$-th column and the $j$-th row of a unitary matrix $V$, i.e., $p_j := \abs{V_{ij}}^2$\end{tabular}\\\hline
    $p(U,x)$ & \begin{tabular}{@{}l@{}}noiseless bitstring probability of measuring $\circqsvt_{f, U}$ with outcome $0$ in the ancilla qubit\\and an $n$-bit binary string $x$ in the $n$ system qubits (dependence on $f$ is omitted)\end{tabular}\\\hline
    $P(U)$ & \begin{tabular}{@{}l@{}}noiseless probability of measuring $\circqsvt_{f, U}$ with outcome $0$ in the ancilla qubit,\\ satisfying $P(U) = \sum_x p(U,x)$ (dependence on $f$ is omitted)\end{tabular}\\\hline
    $p_\expl(U,x)$ & \begin{tabular}{@{}l@{}}bitstring probability of measuring the noisy implementation of $\circqsvt_{f, U}$ with outcome $0$ \\in the ancilla qubit and an $n$-bit binary string $x$ in the $n$ system qubits  (dependence on $f$ is omitted)\end{tabular}\\\hline
    $P_\expl(U)$ & \begin{tabular}{@{}l@{}}probability of measuring the noisy implementation of $\circqsvt_{f, U}$ with outcome $0$  \\with $0$ in the ancilla qubit, satisfying $P_\expl(U) = \sum_x p_\expl(U,x)$ (dependence on $f$ is omitted)\end{tabular}\\\hline    
    \end{tabular}
    \caption{Summary of notations used in the \SP.}
    \label{tab:prob-symbol}
\end{table}

\section{Equivalence between minimal and standard QSVT circuits}\label{sec:simplify_qsvt}

The standard implementation of the QSVT circuit \cite{GilyenSuLowEtAl2019} uses one extra ancilla qubit, called the signal ancilla qubit. In this section, we show that when the system matrix $A$ is block encoded with only one ancilla qubit, the signal ancilla qubit is no longer needed.
Therefore the only ancilla qubit is due to the block encoding of $A$, and the circuit is called the minimal QSVT (mQSVT) circuit in \cref{fig:qsvt_circuit}. Furthermore, \cref{fig:qsvt_circuit} removes all two-qubit and multi-qubit gates outside of the unitary $U$, which greatly simplifies the implementation for a given quantum device. We prove the equivalence between the mQSVT and the standard QSVT circuits in this section for completeness.

For any $(n+1)$-qubit unitary $U$, let the singular value decomposition of its upper-left $n$-qubit submatrix $A$ be $A = W_1 \Sigma V_1^\dagger$. Following the cosine-sine decomposition (CSD), there exists $n$-qubit unitaries $W_2, V_2$ so that $U$ can be decomposed as follows,
\begin{equation*}
    U =\left(\begin{array}{cc}
        A & * \\
        * & B
    \end{array}\right) = \left(\begin{array}{cc}
        W_1 & 0 \\
        0 & W_2
    \end{array}\right) \left(\begin{array}{cc}
        \Sigma & S \\
        -S & \Sigma
    \end{array}\right) \left(\begin{array}{cc}
        V_1 & 0 \\
        0 & V_2
    \end{array}\right)^\dagger,
\end{equation*}
where $S = \sqrt{I - \Sigma^2}$. This decomposition also implies any $n$-qubit non-unitary matrix $A$, up to a scaling factor, can in principle be block encoded using only one ancilla qubit.

Then, the unitary matrix representation of the quantum circuit in \cref{fig:qsvt_circuit} is 
\begin{equation*}
\begin{split}
    \text{Mat. Rep.} = \left(\begin{array}{cc}
        V_1 & 0 \\
        0 & V_2
    \end{array}\right) \left(\begin{array}{cc}
        e^{\I \varphi_0} I_n & 0 \\
        0 & e^{-\I \varphi_0} I_n
    \end{array}\right) \prod_{j=1}^d & \left[ \left(\begin{array}{cc}
        \Sigma & -S \\
        S & \Sigma
    \end{array}\right) \left(\begin{array}{cc}
        e^{\I \varphi_{2j-1}} I_n & 0 \\
        0 & e^{-\I \varphi_{2j-1}} I_n
    \end{array}\right)\right.\\
    &\left.\left(\begin{array}{cc}
        \Sigma & S \\
        -S & \Sigma
    \end{array}\right) \left(\begin{array}{cc}
        e^{\I \varphi_{2j}} I_n & 0 \\
        0 & e^{-\I \varphi_{2j}} I_n
    \end{array}\right) \right]\left(\begin{array}{cc}
        V_1 & 0 \\
        0 & V_2
    \end{array}\right)^\dagger.
\end{split}
\end{equation*}
Let $K$ be the permutation matrix permuting the $j$-th and the $(N+j)$-th columns, and $V = \diag\{V_1,\ V_2\}$. The multiplicand is simplified as a direct sum of $N$ $2$-by-$2$ blocks upon conjugating $\wt{V} := VK$, i.e.
\begin{equation*}
    \wt{V}^\dagger \left(\text{Mat. Rep.}\right) \wt{V} = \bigoplus_{q=0}^{N-1} e^{\I \varphi_0 Z} \prod_{j=1}^d R_q e^{\I \varphi_{2j-1}Z} R_q^\top e^{\I \varphi_{2j} Z},
\end{equation*}
where 
$$R_q = \left(\begin{array}{cc}
    \sigma_q & -\sqrt{1-\sigma_q^2} \\
    \sqrt{1-\sigma_q^2} & \sigma_q
\end{array}\right) = e^{\I \frac{\pi}{4} Z} e^{\I \arccos \left(\sigma_q \right) X} e^{-\I \frac{\pi}{4} Z}.$$
Let $W(x) := e^{\I \arccos(x) X}$, and 
\begin{equation*}
\wt{\varphi}_i=\begin{cases}
\varphi_i+\frac{\pi}{4}, & i=0 \text{ or } 2d,\\
\varphi_i+\frac{\pi}{2}, & i = 2, 4, \cdots, 2d-2,\\
\varphi_i-\frac{\pi}{2}, & i = 1, 3, \cdots, 2d-1.\\
\end{cases}
\end{equation*}
The matrix representation is then
\begin{equation*}
    \wt{V}^\dagger \left(\text{Mat. Rep.}\right) \wt{V} = \bigoplus_{q=0}^{N-1} e^{\I \wt{\varphi}_0 Z} \prod_{j=1}^{2d} W(\sigma_q) e^{\I \wt{\varphi}_j Z}.
\end{equation*}
It is straightforward to show that the following mapping from $[-1,1]$ to $\text{SU}(2)$
\begin{equation*}
x \mapsto e^{\I \wt{\varphi}_0 Z} \prod_{j=1}^{2d} W(x) e^{\I \wt{\varphi}_j Z} = \left(\begin{array}{cc}
    P(x) & \I \sqrt{1-x^2} Q(x) \\
    \I \sqrt{1-x^2} Q^*(x) & P^*(x)
\end{array}\right)
\end{equation*}
defines an even polynomial $P(x)$ of degree at most $2d$, and an odd polynomial $Q(x)$ of degree at most $2d-1$, so that $\abs{P(x)}^2 + (1-x^2) \abs{Q(x)}^2 = 1$ holds for any $x \in [-1,1]$. 

Then, the matrix representation of the quantum circuit is
\begin{equation*}
\begin{split}
    \text{Mat. Rep.} &= V \left(\begin{array}{cc}
        P(\Sigma) & \I \sqrt{I_n - \Sigma^2} Q(\Sigma) \\
        \I \sqrt{I_n - \Sigma^2} Q^*(\Sigma) & P^*(\Sigma)
    \end{array}\right) V^\dagger = \left(\begin{array}{cc}
        P^\svt(A) & * \\
        * & (P^\svt(A))^{\dag}
    \end{array}\right).
\end{split}
\end{equation*}
For example, when $g_t(x)$ is an even polynomial approximation to $s_t(x)=e^{-\I t x^2}$, we can define $g_t(x)=f_t(x^2)$, and the diagonal  $n$-qubit submatrices are $g_t^\svt(A)=f_t(A^\dagger A)$ and $(g_t^\svt(A))^{\dag}=(f_t(A^\dagger A))^{\dag}$ respectively.

This remarkably simple structure of the QSVT circuit is due to the use of $1$-block-encoding. In general, if an $n$-qubit matrix $A$ is block encoded in an $(n+m)$-qubit unitary $U$, the standard QSVT circuit has the structure in \cref{fig:qsvt_circuit_general}.
In particular, even when $m=1$, two CNOT gates are needed to implement each phase rotation.
This introduces additional errors and can be practically cumbersome on near term devices the default two-qubit gate is not CNOT (e.g. $\sqrt{\text{iSWAP}}$).

\begin{figure*}[htbp]
    \centering
        \begin{center}
        \[\scalebox{1}{
                \Qcircuit @C=0.8em @R=1.em {
                        \lstick{\ket{0}}& \targ & \gate{e^{-\I \varphi_{2d} \mathrm{Z}}} & \targ & \qw & \targ & \gate{e^{-\I \varphi_{2d-1} \mathrm{Z}}} & \targ & \qw & \qw &\raisebox{0em}{$\cdots$}&&\qw   &\targ & \gate{e^{-\I \varphi_0 \mathrm{Z}}} & \targ & \qw&\qw & \meter \\
                        \lstick{\ket{0^m}}&\ctrlo{-1} & \qw  & \ctrlo{-1} & \multigate{1}{U} & \ctrlo{-1} & \qw & \ctrlo{-1} & \multigate{1}{U^{\dag}} &\qw &\raisebox{0em}{$\cdots$} &&\qw    &\ctrlo{-1} & \qw & \ctrlo{-1} & \qw& \qw & \meter \\
                        \lstick{\ket{0^n}}&\qw &\qw&\qw &\ghost{U} &\qw&\qw&\qw&\ghost{U^{\dag}}&\qw &\raisebox{0em}{$\cdots$} &&\qw&\qw&\qw   & \qw&\qw & \rstick{P^\svt(A) \ket{0^n}} \gategroup{1}{2}{2}{4}{.7em}{--} \gategroup{1}{6}{2}{8}{.7em}{--} \gategroup{1}{14}{2}{16}{.7em}{--}
        }}
        \]
        \end{center}
    \caption{Quantum circuit for quantum singular value transformation (QSVT) of an even complex matrix polynomial $P$ of degree $2d$. The dashed boxes represent the controlled rotation with phase factor $\varphi_j$ where two $m$-qubit Toffoli gates controlled at $\ket{0^m}$ are used. The QSVT circuit queries the $(n+m)$-qubit quantum circuit $U$ and its inverse recursively for $d$ times. For the given QSVT circuit, by measuring ancilla qubits with outcome $00^m$, the action on the system qubits is the matrix polynomial.}
    \label{fig:qsvt_circuit_general}
\end{figure*}
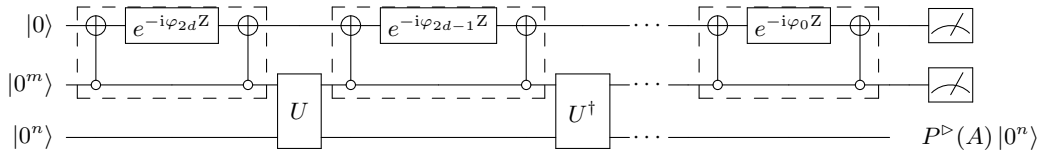

In the mQSVT circuit in \cref{fig:qsvt_circuit}, the phase factors $(\varphi_0, \cdots, \varphi_{2d})$ are determined by an optimization procedure that provides an even polynomial $g_t(x)$ satisfying  $\norm{g_t(x) - s_t(x)}_\infty \le \epsilon$
for some given precision parameter $\epsilon$ (see \cref{sec:opt_phase}), and the evolution time $t$ is encoded in the choice of phase factors. 
We then measure the top ancilla qubit and post-select on the $0$ outcome of this measurement.  This then ensures that the action on the lower $n$ system qubits approximates the Hamiltonian evolution $e^{-\I t \mf{H}}\ket{0^n}\approx f_t(\mf{H})\ket{0^n}$, where $\mf{H} = A^\dagger A$.  Here $f_t(\mf{H})$ is a matrix polynomial, and the approximation error in the operator norm is upper bounded by $\norm{f_t(\mf{H}) - e^{-\I t \mf{H}}}_2 \le \epsilon$.

In the absence of quantum errors the probability of measuring the top ancilla qubit with outcome $0$, i.e. the $P_t(U) := \norm{f_t(\mf{H})\ket{\psi}}^2$, will be close to $1$.  Specifically, 
  by the triangle inequality, the probability of measuring the top ancilla qubit with outcome $0$ is lower bounded:
\begin{equation}
\begin{split}
P_t(U) =& \norm{g_t(\Sigma) V^\dagger\ket{0^n}}_2^2 = \sum_{j=0}^{\Nsys-1} \abs{g_{t}(\sigma_j)}^2 p_j = \left|1 + \sum_{j=0}^{\Nsys-1} \left( \left(g_t(\sigma_j) - s_t(\sigma_j)\right) \overline{g_t(\sigma_j)} + s_t(\sigma_j) \overline{\left( g_t(\sigma_j) - s_t(\sigma_j) \right)} \right) p_j\right|\\
\ge& 1 - 2 \epsilon \sum_{j=0}^{\Nsys-1} p_j = 1 - 2\epsilon,
\end{split}
\label{eqn:PU_bound}
\end{equation}
where $p_j = \abs{V_{0,j}}^2$. 

We also find that the probability $p_t(U, x) = \abs{\braket{0x | \circqsvt_{f_t,U} | 00^{n}}}^2 \approx \abs{\braket{x | \exp(-\I t \mf{H}) | 0^n}}^2$ will characterize the dynamics of the propagation from $0^n$ to $x$ for any $n$-bit string $x \in \{0, 1\}^n$. 
If the simulation time $t$ is short, then $\exp(-\I t \mf{H})\approx I$, and $p_{t}(U,0^n)$ can be much larger than $p_t(U,x)$ for any bitstring $x\ne 0^n$. 
This issue will be particularly important when defining the `heavy weight samples' in later discussions.  
Therefore we shall primarily focus on the case when $x\ne 0^n$.

\begin{figure*}[htbp]
\centering
\includegraphics[width=0.9\textwidth]{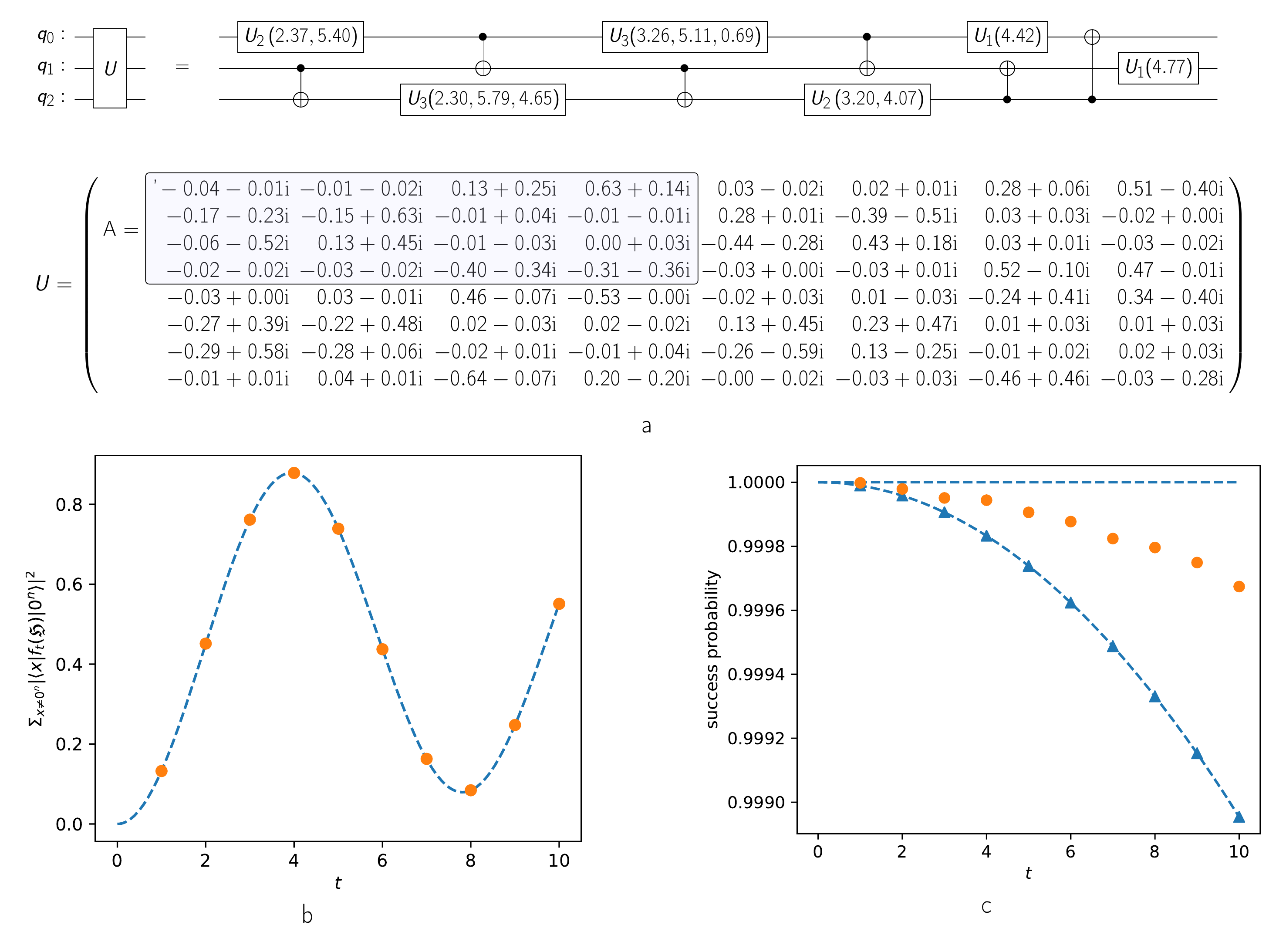}
\caption{An illustrative implementation from the quantum Hamiltonian simulation benchmark. (a) Top: A $3$-qubit random quantum circuit constructed from the basic gate set $\{ \text{U}_1, \text{U}_2, \text{U}_3, \text{CNOT} \}$. Bottom: The $3$-qubit unitary matrix representation $U$ of the quantum circuit and its upper-left $2$-qubit submatrix $A$ (in the shaded area). The top qubit $q_0$ is the ancilla qubit for block encoding. (b) Dynamics of the evolution away from the initial condition $\sum_{x\ne 0^n}p_t(U,x)$ implemented using the mQSVT circuit, compared to the reference solution $\sum_{x\ne 0^n}|\braket{x|e^{-\I t \mf{H}}|0^n}|^2$.   (c) The probability $P_t(U)$ obtained from $10^6$ noiseless measurements (orange dots) and theoretical bounds (blue dashed). }
\label{fig:prop_qsvt_circuit}
\end{figure*}

As an illustrative example, \cref{fig:prop_qsvt_circuit} shows a quantum circuit implementing a $2$-qubit matrix $A$ encoded by a $3$-qubit unitary matrix $U$. The construction uses only the basic gate set $\{ \text{U}_1, \text{U}_2, \text{U}_3, \text{CNOT} \}$.
In \cref{sec:opt_phase} we describe an optimization based method to 
obtain the phase factors for a relatively short time $t$ to a small approximation error $\epsilon$. 
In this example we set $t=1$.
To obtain a theoretical error bound at large $t$, we can use the phase factor concatenation technique in \cref{sec:long_phase} to obtain the simulation at $t$ from $2$ to $10$, and the error bound $\epsilon_t = \epsilon t^2$ is given by \cref{thm:phase_concatenate}.
Using these circuits, we may measure the outcome of the system qubits in the computational basis, and follow the dynamics of the probability $\sum_{x\ne 0^n} p_t(U,x)$, i.e. that of the quantum state moving away from the initial state $\ket{0^n}$.  
\cref{fig:prop_qsvt_circuit}(b) shows that this agrees very well with the result using the exact dynamics $\sum_{x\ne 0^n}|\braket{x|e^{-\I t \mf{H}}|0^n}|^2$.
Furthermore, according to \cref{eqn:PU_bound}, the probability $P_t(U)$ in \cref{fig:prop_qsvt_circuit}(c) satisfies the theoretical bounds $1-2\epsilon_t \le P_t(U) \le 1$ and is very close to $1$.

\section{Optimization based method for finding phase factors in the Hamiltonian simulation benchmark}\label{sec:opt_phase}

In order to implement the Hamiltonian simulation benchmark at time $t$, we need to find the phase factors $\Phi$ that generates an even polynomial $P(x,\Phi)=g_t(x)$ so that $\norm{g_t(x) - s_t(x)}_\infty\le \epsilon$ for a sufficiently small $\epsilon$. 
For a large class of polynomials, the existence of such phase factors is established in \cite[Theorem 4]{GilyenSuLowEtAl2019}, and summarized in \cref{thm:qsp}.
\begin{theorem}[\textbf{Quantum signal processing in SU(2)}]\label{thm:qsp}
    For any $P, Q \in \mathbb{C}[x]$ and a positive integer $d$ such that (1) $\deg(P) \leq d, \deg(Q) \leq d-1$, (2) $P$ has parity $(d\mod2)$ and $Q$ has parity $(d-1 \mod 2)$, (3) $|P(x)|^2 + (1-x^2) |Q(x)|^2 = 1, \forall x \in [-1, 1]$. Then, there exists a set of phase factors $\Phi := (\phi_0, \cdots, \phi_d) \in \RR^{d+1}$ such that
\begin{equation}
\label{eq:qsp-gslw}
\begin{split}
        U(x, \Phi) &= e^{\I \phi_0 Z} \prod_{j=1}^{d} \left[ W(x) e^{\I \phi_j Z} \right] = \left( \begin{array}{cc}
        P(x) & \I Q(x) \sqrt{1 - x^2}\\
        \I Q^*(x) \sqrt{1 - x^2} & P^*(x)
        \end{array} \right)
\end{split}
\end{equation}
where 
\begin{displaymath}
W(x) = e^{\I \arccos(x) X}=\left(\begin{array}{cc}{x} & {\I \sqrt{1-x^{2}}} \\ {\I \sqrt{1-x^{2}}} & {x}\end{array}\right).
\end{displaymath}
\end{theorem}
The phase factors in the theorem and those used in the quantum circuit in the main text is related by the following relation
\begin{equation}
\varphi_i=\begin{cases}
\phi_0+\frac{\pi}{4}, & i=0,\\
\phi_i+\frac{\pi}{2}, & 1\le i\le d-1,\\
\phi_d+\frac{\pi}{4}, & i=d.\\
\end{cases}
\label{eqn:phi_varphi_relation}
\end{equation}

In order to find the phase factors, the standard practice follows a two-step procedure. 
We first identify the approximate polynomial $P(x)$.
Then the phase factors for $P(x)$ are computed following a recursive relation described in \cite[Theorem 4]{GilyenSuLowEtAl2019}.
In the case of the Hamiltonian simulation benchmark, it is highly nontrivial to find an approximate polynomial $P(x)$ satisfying the conditions in \cref{thm:qsp} while approximating the function $e^{-\I t x^2}$ sufficiently well.
Therefore we cannot follow the standard procedure to evaluate the phase factors.

On the other hand, the recently developed optimization based method \cite{DongMengWhaleyEtAl2021} provides an alternative route to streamline this process. 
Instead of following a two-step procedure, the optimization based method allows one to find both the approximate polynomial and the phase factor sequence in a single step.
Note that the optimization procedure in \cite{DongMengWhaleyEtAl2021} only addresses the case when the target function is real. 
Here the target function $s_t(x)$ is complex.
Below we present a modified optimization procedure to find the phase factor sequence for complex polynomials.

Specifically, given an arbitrary set of phase factors $\Phi \in \RR^{d+1}$, \cref{thm:qsp} defines a mapping $\RR^{d+1} \rightarrow \CC[x]$ giving a complex polynomial of degree at most $d$ via $P(x,\Phi) := \braket{0 | U(x, \Phi) | 0}$. Note that given the parity constraint, the number of (complex) degrees of freedom is $\wt{d}$ where $\wt{d} := \lceil \frac{d+1}{2} \rceil$. Therefore, to fix the polynomial, one needs to sample $\wt{d}$ points given the polynomial is complex valued. In practice, to ensure numerical stability, we sample the function on $x_k=\cos\left(\frac{2k-1}{4\wt{d}}\pi\right), k = 1, \cdots, \wt{d}$, which are the positive Chebyshev nodes of $T_{2\wt{d}}(x)$. The optimization based method view $P$ as a nonlinear approximation ansatz. Define the objective function as
\begin{equation}
    F(\Phi) := \frac{1}{\wt{d}} \sum_{k=1}^{\wt{d}} \abs{P(x_k, \Phi) - s_t(x_k)}^2.
\end{equation}
Taking the $(2\pi)$-periodicity into account, the optimization problem is
\begin{equation}\label{eqn:optprob}
    \Phi^* = \argminl_{\Phi \in [-\pi,\pi)^{d+1}} F(\Phi).
\end{equation}

The optimization problem is numerically solved by a quasi-Newton method.
\cref{tab:numerical-precision} describes the approximation error for polynomials  measured by $\norm{P(x,\Phi^*) - s_t(x)}_\infty$ at simulation time $t=1$, for polynomial degrees between $6$ and $20$. 
When the polynomial degree is $20$, the approximation error is as small as $10^{-8}$, which demonstrates the effectiveness of the optimization based method.

\begin{table}[htbp]
    \centering
    \begin{tabular}{@{} *{2}{c} @{}}\midrule
    degree $d$ &        approximation error\\\midrule
    $6$ &       $5.543\times 10^{-03}$\\
    $8$ &       $5.805\times 10^{-04}$\\
    $10$&       $5.230\times 10^{-06}$\\
    $14$&       $3.332\times 10^{-06}$\\
    $18$&       $9.535\times 10^{-08}$\\
    $20$&       $1.107\times 10^{-08}$\\\midrule
    \end{tabular}
    \caption{Approximation error $\norm{P(x,\Phi^*) - e^{-\I tx^2}}_\infty$ at time $t=1$ with different polynomial degrees $d$.}
    \label{tab:numerical-precision}
\end{table}

According to \cref{fig:supremacy-region-HS}, there exists an `optimal' simulation time $t^\text{opt} = 4.8096$, for which the threshold fidelity $\alpha^*(t^\text{opt})\approx 0$ (the derivation is in  \cref{sec:topt_derive}). \cref{tab:phase-factor-t-opt} describes the phase factor sequences that can be directly used in \cref{fig:qsvt_circuit} to perform Hamiltonian simulation at time $t^\text{opt}$.
In order to reach low ($3.0\times 10^{-2}$), medium ($9.4\times 10^{-5}$), and high ($1.6\times 10^{-6}$) accuracy, the degrees of the polynomial found by the optimization procedure are $10,18,26$, respectively.
\cref{fig:app-err-t-top} further shows the pointwise approximate error on the interval $[0,1]$ (the error on $[-1,0]$ is the same due to the even parity).
Compared to \cref{tab:numerical-precision}, in order to reach precision $\epsilon=3.3\times 10^{-6}$ at simulation time $t=1$, the polynomial degree still needs to be $14$.
So even though $t^\text{opt}$ is nearly $5$ times larger,  the polynomial degree only increases by less than twofold to reach similar accuracy. Since $t^\text{opt}$ is still relatively small, this does not violate the `no-fast-forwarding' theorem of Hamiltonian simulation \cite{BerryAhokasCleveEtAl2007}.

\begin{table}[htbp]
\centering
\begin{tabular}{@{} *{7}{c} @{}}\midrule
\multicolumn{7}{c@{}}{approximation error $\norm{P(x,\Phi) - e^{-i t x^2}}_\infty = 3.027\times 10^{-2}$}\\\midrule
$\varphi_{0}$ & $\varphi_{1}$ & $\varphi_{2}$ & $\varphi_{3}$ & $\varphi_{4}$ & $\varphi_{5}$ & $\varphi_{6}$ \\
-2.7731963 & 2.7942520 & -1.5707963 & 2.5930970 & -1.5707963 & -0.6434012 & -1.5707963 \\\midrule
$\varphi_{7}$ & $\varphi_{8}$ & $\varphi_{9}$ & $\varphi_{10}$ &  &  &  \\
2.5930970 & -1.5707963 & 2.7942520 & -2.7731963 &  &  &  \\\midrule
\end{tabular}

\bigskip
\begin{tabular}{@{} *{7}{c} @{}}\midrule
\multicolumn{7}{c@{}}{approximation error $\norm{P(x,\Phi) - e^{-i t x^2}}_\infty = 9.406\times 10^{-5}$}\\\midrule
$\varphi_{0}$ & $\varphi_{1}$ & $\varphi_{2}$ & $\varphi_{3}$ & $\varphi_{4}$ & $\varphi_{5}$ & $\varphi_{6}$ \\
-2.7731963 & 2.8229351 & -1.5707963 & -2.5716144 & -1.5707963 & -3.1056796 & -1.5707963 \\\midrule
$\varphi_{7}$ & $\varphi_{8}$ & $\varphi_{9}$ & $\varphi_{10}$ & $\varphi_{11}$ & $\varphi_{12}$ & $\varphi_{13}$ \\
-1.1677625 & 1.5707963 & -0.6437954 & 1.5707963 & -1.1677625 & -1.5707963 & -3.1056796 \\\midrule
$\varphi_{14}$ & $\varphi_{15}$ & $\varphi_{16}$ & $\varphi_{17}$ & $\varphi_{18}$ &  &  \\
-1.5707963 & -2.5716144 & -1.5707963 & 2.8229351 & -2.7731963 &  &  \\\midrule
\end{tabular}

\bigskip
\begin{tabular}{@{} *{7}{c} @{}}\midrule
\multicolumn{7}{c@{}}{approximation error $\norm{P(x,\Phi) - e^{-i t x^2}}_\infty = 1.644\times 10^{-6}$}\\\midrule
$\varphi_{0}$ & $\varphi_{1}$ & $\varphi_{2}$ & $\varphi_{3}$ & $\varphi_{4}$ & $\varphi_{5}$ & $\varphi_{6}$ \\
-1.5893341 & -0.3207550 & 2.8668325 & -2.9662972 & -1.1921175 & -0.4528806 & 1.5270366 \\\midrule
$\varphi_{7}$ & $\varphi_{8}$ & $\varphi_{9}$ & $\varphi_{10}$ & $\varphi_{11}$ & $\varphi_{12}$ & $\varphi_{13}$ \\
1.6658052 & -0.2379487 & -2.9130657 & 0.3245889 & 0.7863552 & -1.3306612 & -0.2863103 \\\midrule
$\varphi_{14}$ & $\varphi_{15}$ & $\varphi_{16}$ & $\varphi_{17}$ & $\varphi_{18}$ & $\varphi_{19}$ & $\varphi_{20}$ \\
-1.3306612 & 0.7863552 & 0.3245889 & -2.9130657 & -0.2379487 & 1.6658052 & 1.5270366 \\\midrule
$\varphi_{21}$ & $\varphi_{22}$ & $\varphi_{23}$ & $\varphi_{24}$ & $\varphi_{25}$ & $\varphi_{26}$ &  \\
-0.4528806 & -1.1921175 & -2.9662972 & 2.8668325 & -0.3207550 & -1.5893341 &  \\\midrule
\end{tabular}
\caption{Phase factors for Hamiltonian simulation at time $t^\text{opt} = 4.8096$. The table lists three sets of phase factors with different approximation errors.}
\label{tab:phase-factor-t-opt}
\end{table}

\begin{figure*}[htbp]
    \centering
    \vspace{-10pt}
    \subfigure[]{
        \includegraphics[width=.3\textwidth]{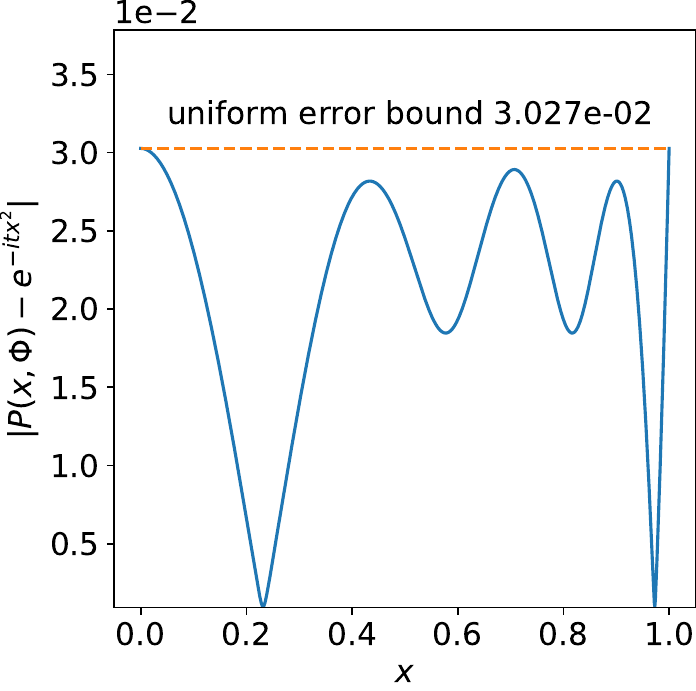}
    }
    \subfigure[]{
        \includegraphics[width=0.3\textwidth]{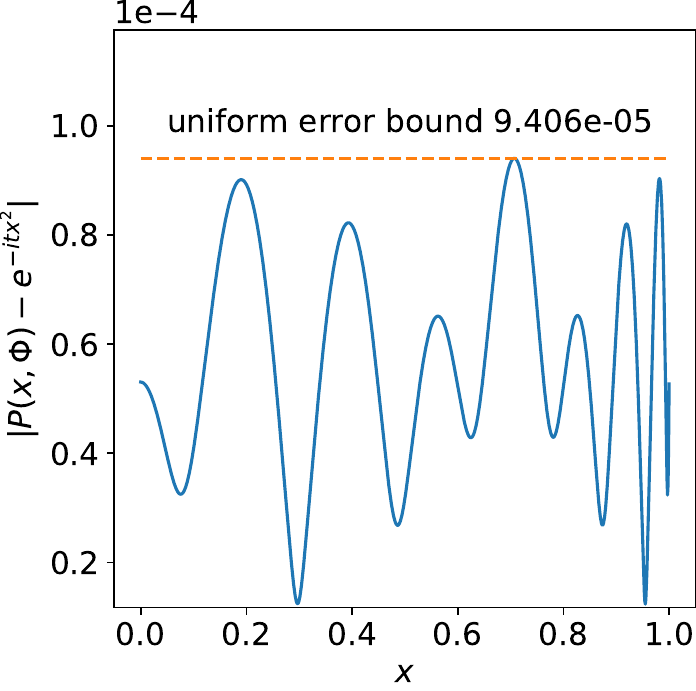}
    }
    \subfigure[]{
        \includegraphics[width=0.3\textwidth]{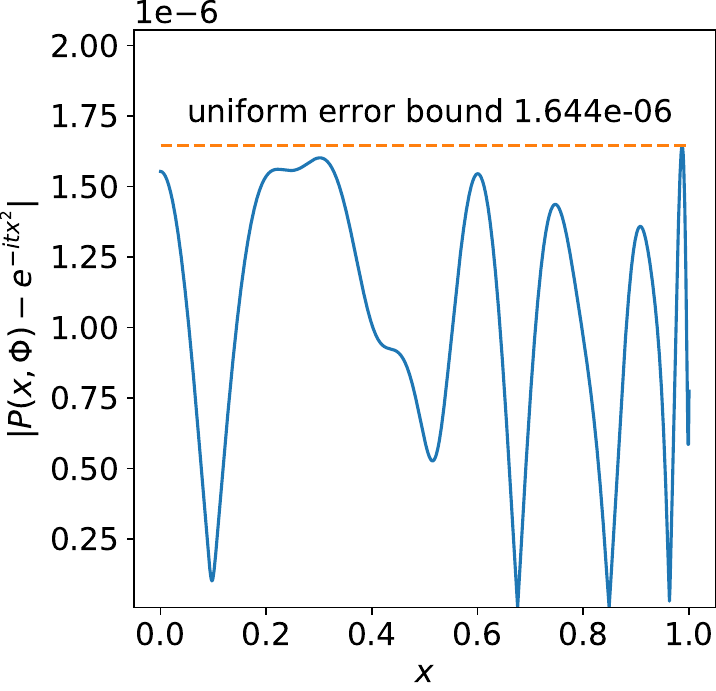}
    }
    \vspace{-10pt}
    \caption{Point-wise approximation error of phase factors for Hamiltonian simulation at time $t^\text{opt} = 4.8096$. The corresponding sets of phase factors are listed in \cref{tab:phase-factor-t-opt}. }
    \label{fig:app-err-t-top}
\end{figure*}

\section{Noise model}\label{sec:error_model}
In the experimental setting, the density matrix after the application of the quantum circuit $\circqsvt_{f, U}$ can be written as
\begin{equation}\label{eqn:rho_channel}
    \varrho_\expl = \alpha \ket{\psi_{f, U}}\bra{\psi_{f, U}} + (1-\alpha) \chi_{f, U}.
\end{equation}
Here, $\ket{\psi_{f,U}} := \circqsvt_{f, U} \ket{0} \otimes \ket{\psi}$ is the ideal quantum state generated by the exact implementation of the quantum circuit, and the operator $\chi_{f, U}$ is due to the noise channel. 
Under the global depolarizing noise model for the mQSVT circuit, we have $\chi_{f, U}=I/2^{n+1}$ and the diagonal entries of~\cref{eqn:rho_channel} then yield the probabilities of
~\cref{eqn:succ-prob-fidelity}. 

However, in practice $\chi_{f, U}$ may not be a scaled identity matrix, or even a diagonal matrix.
If so, the system linear cross-entropy score (sXES) in \cref{eqn:sXES} should be expressed more generally as
\begin{equation}
\text{sXES}(U) = \sum_{x\ne 0^n} p(U, x) p_\expl(U,x)=
\alpha \sum_{x\ne 0^n} p^2(U, x) + (1-\alpha) \sum_{x\ne 0^n} p(U, x)\braket{x|\chi_{f, U}|x}.
\end{equation}
Now we can write 
\begin{equation}\label{eqn:sxes_noise}
\sum_{x\ne 0^n}p(U, x)\braket{x|\chi_{f, U}|x}=\frac{1}{2^{n+1}}\sum_{x\ne 0^n}p(U, x)+\epsilon_{\chi},
\end{equation}
where $\epsilon_\chi$ represents the effects of correlations between noise channels. Under the global depolarized noise channel we have $\epsilon_{\chi}=0$.
Refs.~\cite{BoixoIsakovSmelyanskiyEtAl2018} and~\cite[Supplementary information IV.B]{AruteAryaBabbushEtAl2019} argue that when the noise $\braket{x|\chi_{f, U}|x}$ is uncorrelated with the signal $p(U, x)$, the statistical fluctuation error $\epsilon_{\chi}$ can be of higher order, as a result of the concentration of measure phenomenon and Levy's lemma in high dimensional spaces~\cite{Ledoux2001}.

On the other hand, even though each component $U$ in the mQSVT circuit can be taken to be the same as the random circuit in the supremacy experiment of Ref.~\cite{AruteAryaBabbushEtAl2019}, the overall mQSVT circuit exhibits additional structure due to the presence of the $R_z$ (phase) gates and the multiple repetitions of each $U$ and $U^{\dag}$ pair. Thus 
the global depolarized error model may not hold in practice.
To overcome this difficulty, the randomized compilation method in Ref.~\cite{WallmanEmerson2016} can be applied to the mQSVT circuits to convert all noise in the circuit into stochastic Pauli errors. The theory of error randomization can then be used to guarantee that the effect of these Pauli errors can be accurately modeled by global depolarization~\cite{BoixoIsakovSmelyanskiyEtAl2018,ProctorCarignanDugas2019}.
Specifically, the method of~\cite{WallmanEmerson2016} divides a universal gate set into a set of `easy' gates and a set of `hard' gates and then uses randomization applied to the easy gates to twirl the errors on the hard gates. Ref.~\cite{WallmanEmerson2016} shows that for a set of  elementary gates consisting of Clifford plus T gates, making the easy set equal to one-qubit Clifford operations allows one to tailor general multi-qubit noise into Pauli noise on the hard gates. In the mQSVT circuit, thanks to the full flexibility in generating the random circuit $U$, we may explicitly construct $U$ so that most gates are taken directly from an easy gate set. This enables the twirling operations to be applied independently within all of the $U,U^{\dag}$  blocks and thereby reduces the noise in each layer of gates in an mQSVT circuit into stochastic Pauli channels. Since the $U$ circuits consist of random layers that generate a 2-design, multiple layers of this circuit implement an approximate 2-design. A 2-design twirls any errors into global depolarization (see also~\cite{MagesanGambettaEmerson2011,GellerZhou2013,CaiXuBenjamin2020}), and so the overall error on a single $U$ or $U^{\dag}$ can be approximated by a global depolarizing channel~\cite{BoixoIsakovSmelyanskiyEtAl2018}. The repeated structure of $U$ and $U^{\dag}$ subroutines in a mQSVT means that an error that is randomized by $U$ is then `de-randomized' by $U^{\dag}$. To mitigate this effect, it is possible to apply the randomized compilation method of Ref.~\cite{WallmanEmerson2016} to each $U,U^{\dag}$ independently. The impact of this type of effects has been studied in the setting of randomized mirror circuit benchmarks~\cite{ProctorCarignanDugas2019,ProctorSeritanRudingerEtAl2021,ProctorRudingerYoungEtAl2022}, which also use circuits with a $U,U^{\dag}$ structure and employ a form of randomized compilation. These studies have explicitly shown that with this approach the overall error can still be modeled by a global depolarizing channel.

\section{Structure of the probability space of measuring noisy random quantum circuits and sXES}\label{sec:structure_prob}

There are two sources of randomness when measuring noisy random quantum circuits. The first is due to the random choice of $U$ with probability density $\bP(U)$. The second is due to the Monte Carlo nature of the quantum measurement.
Specifically, given the choice of $U$ and a noisy implementation of the quantum circuit, the probability to obtain $0x$ as the measurement outcome is $p_\expl(U,x) := \braket{0x|\varrho_\expl|0x}$. 
The joint probability of $U$ and the measurement outcome $0x$ is
\begin{equation*}
    \bP_\expl(U,x) = p_\expl(U,x) \bP(U).
\end{equation*}
Here for simplicity, we only focus on the measurement result whose ancilla qubit is measured with outcome $0$. 
When $\bP(U)$ is given by the Haar measure, the noise channel is depolarized, and the noisy (experimental) bitstring probability is
\begin{equation}
p_\expl(U,x) = \alpha p(U,x) + \frac{1-\alpha}{2N},
\end{equation}
which is the convex combination of the exact bitstring probability and the uniform distribution \cite{BoixoIsakovSmelyanskiyEtAl2018}. 

The information of the circuit fidelity can then be encoded in the experimental average of various quantities. For example, the bitstring probability for nonzero bitstrings is given by
\begin{equation}
    \exptwrt{\expl}{p(U,x) ; x\ne 0^n} = \expt{\sum_{x\ne 0^n} p(U,x) p_\expl(U,x)} = \sum_{x\ne 0^n} \alpha \expt{p(U,x)^2} + \frac{1-\alpha}{2N} \expt{p(U,x)}.
\label{eqn:exp_pu}
\end{equation}
\cref{eqn:exp_pu} connects quantities evaluated from quantum experiments and classical computation on the left hand side and those from classical computations on the right hand side. 
The left-hand side is given by the system linear cross-entropy score (sXES) in \cref{eqn:sXES}, which can be evaluated from multiplying the bitstring frequency $p_\expl(U,x)$ obtained from the quantum experiment and the bitstring probability $p(U, x)$ computed classically. 
The quantities on the right-hand side can be evaluated fully classically. The circuit fidelity $\alpha$ is then the only unknown and can be solved for by substituting the quantum experimental and classically computed quantities into \cref{eqn:circuit-fidelity-XES} of the main text.

\section{Estimating circuit fidelity from quantum unitary evolution score}\label{sec:ques_fidelity}

The experimental average of the probability of measuring the ancilla qubit with outcome $0$ is 
\begin{equation}
    \exptwrt{\expl}{P(U)} = \expt{\sum_x P(U) p_\expl(U,x)} = \expt{P(U) P_\expl(U)} = \alpha \expt{P(U)^2} + \frac{1-\alpha}{2} \expt{P(U)}.
\label{eqn:PU_alpha}
\end{equation}
Here $P(U) := \sum_x p(U,x)$, and $P_\expl(U)$ is the probability which can be approximately determined by the bit frequency of the measurement outcome in the experiment. 
Rearranging the terms in \cref{eqn:PU_alpha}, the circuit fidelity can be alternatively estimated via 
\begin{equation}
        \alpha = \frac{\expt{P(U) P_\expl(U)} - \frac{1}{2} \expt{P(U)}}{\expt{P(U)^2} - \frac{1}{2} \expt{P(U)}}.
\end{equation}

From \cref{eqn:PU_bound}, we have $P(U) \in [1-2\epsilon, 1]$. Hence a lower and upper bound on the fidelity follows:
\begin{equation}\label{eqn:est-fid-HS-ideal}
\begin{split}
    \underline{\alpha} &:= \frac{2 \left( 1 - 2\epsilon \right) \expt{P_\expl(U)} - 1}{1+2\epsilon} \le \alpha \le \frac{2 \expt{P_\expl(U)} - \left(1-2\epsilon\right)}{1 - 8 \epsilon} =: \overline{\alpha}.
\end{split}
\end{equation}
The difference between the upper and lower bound is 
\begin{equation}
\overline{\alpha} - \underline{\alpha} \le 16 \epsilon + \Or(\epsilon^2).
\end{equation}
Therefore, $\lim_{\epsilon\to 0}(\overline{\alpha} - \underline{\alpha})=0$ and the derived bounds are tight. Let us choose the form of the estimate as $\epsilon$-independent
\begin{equation}\label{eqn:app:fid-estimate-param-indep}
    \alpha_\text{QUES} := 2 \expt{P_\expl(U)} - 1 \in [\underline{\alpha}, \overline{\alpha}].
\end{equation}
Then, $\abs{\alpha_\text{QUES} - \alpha} \le \overline{\alpha} - \underline{\alpha} \le 16 \epsilon + \Or(\epsilon^2)$. Furthermore, the estimate can be determined using only the experimentally measurable quantity $P_\expl(U)$ and is independent of the classical computation of $P(U)$, which may be hard to evaluate for large $n$. This remarkable fact, namely the evaluation of circuit fidelity without any classical computation, arises from the approximate implementation of Hamiltonian simulation of the overall circuit. Since only one ancilla qubit is measured, the QUES defined in \cref{eqn:app:fid-estimate-param-indep} cannot entirely capture whether the circuit is implemented correctly. However, when the assumption that the noise channel is depolarized and when the polynomial approximation to $s_t(x)$ is sufficiently accurate, $\alpha_\text{QUES}$ provides a very good estimate to the circuit fidelity.

\section{Algorithm for constructing random quantum circuits and  numerical convergence to Haar measure}\label{sec:converge-to-Haar}

In order to theoretically analyze the circuit fidelity, we need the additional assumption that $\bP(U)$ is the Haar measure.
This has the advantage that several terms in \cref{eqn:circuit-fidelity-XES} can be evaluate analytically. Using the Haar measure, the statistics of an ensemble of random Hamiltonians is much simplified and can be computed by the statistics of the truncation matrix of Haar unitaries \cite{ZyczkowskiSommers2000,PetzReffy2004,MastrodonatoTumulka2007,Collins2003,CollinsSniady2006}.
Details of the statistics are given in \cref{sec:stat-inherit-from-Haar}.
Furthermore, if $U$ is Haar-distributed, then the noise effect of directly sampling $U$ is well captured by a fully depolarized error channel, due to the nearly maximal entanglement in the output state~\cite{BoixoIsakovSmelyanskiyEtAl2018,AruteAryaBabbushEtAl2019}. The need to choose an appropriate circuit depth $\ell$ such that the circuit statistics approximate those of Haar unitaries motivates an investigation of the statistics of Haar random quantum circuits of finite number of qubits.

\begin{algorithm}[htbp]
\label{alg:rqc}
\hrule
\caption{Constructing random quantum circuits}
\begin{algorithmic}[0]
\STATE \textbf{Input:} Coupling map $G = \langle V, E\rangle$ where $V$ is the set of $n$ qubits, $E$ is the set of qubit pairs on which CNOT is available, basic gates $\Gamma = \{ \mathrm{U1}, \mathrm{U2}, \mathrm{U3}, \mathrm{CNOT} \}$, the number of total one-qubit gates $g_1$, and the density of one-qubit gates $p_1 \in (0, 1)$.
\vspace{1em}
\STATE Set the number of two-qubit gates to $g_2 = \lceil \frac{1-p_1}{2p_1} g_1 \rceil$.
\STATE Set the maximal number of two-qubit gates in each layer to $y_2 = \lceil \frac{1-p_1}{2} n \rceil$.
\STATE Set $m_1 = m_2 = 0$, initialize an empty quantum circuit $\mc{C}$.
\WHILE{$m_1 \le g_1$ \textbf{and} $m_2 \le g_2$}
\STATE Draw $x_2 \le y_2$ pairs of qubits from $E$ so that each pair $(q_1, q_2)$ and its permutation $(q_2, q_1)$ are not selected in the previous layer. The choice of $x_2$ also satisfies $m_2 + x_2 \le g_2$.
\STATE Draw $x_1 = \mathrm{min}\{n - 2x_2, g_1 - m_1\}$ one-qubit gates uniformly at random from $\Gamma\backslash\{\mathrm{CNOT}\}$ and act them on the rest of qubits in this layer.
\STATE Update the numbers of one- and two-qubit gates, $m_1 \leftarrow m_1 + x_1$ and $m_2 \leftarrow m_2 + x_2$.
\ENDWHILE

\IF{$m_1 < g_1$}
\STATE Append layers of random $g_1 - m_1$ one-qubit gates sampled uniformly at random from $\Gamma\backslash\{\mathrm{CNOT}\}$.
\ELSIF{$m_2 < g_2$}
\STATE Append layers of $g_2 - m_2$ CNOT gates acting on random operands.
\ENDIF
\STATE \textbf{Return:} A random quantum circuit $\mc{C}$ with $g_1$ one-qubit gates and $g_2$ two-qubit gates.
\end{algorithmic}
\end{algorithm}
We first construct random quantum circuits by using the algorithm given in \cref{alg:rqc}. It follows a similar recipe in Ref. \cite{DongLin2021}. We set the basic one-qubit gates to U1, U2 and U3 gates. Up to a global phase factor, the U3 gate is
\[
\mathrm{U}_3(\theta,\phi,\lambda) = R_z(\phi+3\pi)R_x(\pi/2)R_z(\theta+\pi)R_x(\pi/2)R_z(\lambda),
\]
which is a generic single-qubit operation parameterized by three Euler angles. The U1 and U2 gates are defined by restricting to one or two Z-rotation angles respectively, i.e. 
\[
\mathrm{U}_1(\lambda) = R_z(\lambda), \quad \mathrm{U}_2(\phi, \lambda) = R_z(\phi + \pi/2) R_x(\pi/2) R_z(\lambda-\pi/2).
\]
Taken together with the CNOT gate, these form a continuously parameterized gate set that is universal.

Although we specify the choice of one-qubit gates and the use of the CNOT gate here, \cref{alg:rqc} can be directly generalized to an arbitrary basic gate set. The random quantum circuit generated by the algorithm respects the architecture of a quantum computer. In practice, we set the density of one-qubit gates to $p_1 = 0.5$. Then, for an $n$-qubit random quantum circuit with $\ell$ layers, the number of one-qubit gates is $g_1 = \frac{\ell n}{2}$ and that of two-qubit gates is $g_2 = \frac{\ell n}{4}$. 

To measure the numerical convergence of random circuits to the Haar measure, we first summarize some of the statistical properties of the Haar measure.
Given an $n$-qubit Haar-distributed unitary $U$, we denote $p_{ij} := \abs{U_{ij}}^2$. As a special case of the more general \cref{thm:pdf-matrix} (to be presented in \cref{sec:stat-inherit-from-Haar}), the $p_{ij}$'s are identically Beta-distributed.
\begin{theorem}\label{thm:Haar-succ-prob}
    The probability density of $p_{ij}$ is $\mathrm{Beta}(1,N-1)$,
    \begin{equation*}
        \bP(p_{ij}) = (N-1) (1-p_{ij})^{N-2} \II_{0 \le p_{ij} \le 1}.
    \end{equation*}
\end{theorem}
\begin{proof}
    Let the submatrix of interest be the upper left $1$-by-$1$ block, namely, a single matrix element $a := U_{00}$. Note that $p_{00} := \abs{a}^2$. Then, \cref{eqn:matrix-distribution-general} indicates that the probability density of $a$ is
    \begin{equation*}
        \bP(a) \propto \left(1-p_{00}\right)^{N-2}\II_{0 \le p_{00} \le 1}.
    \end{equation*}
    The polar decomposition of the complex number $a = r e^{\I \theta}$ yields the Jacobian $\rd a = r \rd r \rd \theta \propto \rd p_{00} \rd \theta$. Then, integrating with respect to $\rd \theta$, the marginal distribution of $p_{00}$ is
    \begin{equation*}
        \bP(p_{00}) = (N-1) \left(1-p_{00}\right)^{N-2}\II_{0 \le p_{00} \le 1}.
    \end{equation*}
    This is the $\mathrm{Beta}(1,N-1)$ distribution. When $i \ne 0$ or $j \ne 0$, let $K_1$ be the matrix permuting the $i$-th row and the $0$-th row by left multiplication, and let $K_2$ be the matrix permuting the $j$-th column and the $0$-th column by right multiplication. Then, $\wt{U} := K_1 U K_2$ is Haar distributed by the bi-invariance of the Haar measure. Furthermore, $\wt{U}_{00} = U_{ij}$. Therefore, the previous proof shows that $p_{ij}$ is also $\mathrm{Beta}(1,N-1)$ distributed.
\end{proof}
 
Note that in the limit $N \gg 1$, the distribution of $p_{ij}$ is well approximated by the exponential distribution $\text{Exp}(N)$, a.k.a. the Porter-Thomas distribution derived in \cite{BoixoIsakovSmelyanskiyEtAl2018}. The statistics follows straightforward computation by integrating with respect to the probability density.
\begin{theorem}\label{thm:stat-Haar}
    Let $M_k := \sum_{i=0}^{N-1} p_{ij}^k$ be the $k$-th moment,  $S := \sum_{i=0}^{N-1} - p_{ij} \ln(p_{ij})$ be the entropy. Their averages with respect to the Haar distribution take the form
    \begin{equation*}
        M_k^\mathrm{Haar} := \expt{M_k} = \prod_{i=1}^{k-1} \frac{1+i}{N+i},\quad S^\mathrm{Haar} := \expt{S} = \sum_{i=2}^N i^{-1}.
    \end{equation*}
    The variance of the $k$-th moment $V_k^\mathrm{Haar} := \sum_{i=0}^{N-1} \mathrm{Var}(p_{ij}^k)$ is
    \begin{equation*}
        V_k^\mathrm{Haar} = \left(\frac{1}{N}\binom{2k}{k} + \frac{N-1}{N}\right) \prod_{i=k}^{2k-1} \frac{N-k+i}{N+i} - 1.
    \end{equation*}
\end{theorem}
We remark that in the limit $N \gg 1$, $M_k^\mathrm{Haar} \approx \frac{k!}{N^{k-1}}$ and $S^\mathrm{Haar} \approx \ln(N) + \gamma - 1$ where $\gamma$ is Euler's constant. The asymptotic results are the same as those derived in \cite{BoixoIsakovSmelyanskiyEtAl2018}. The variance is asymptotically $V_k^\mathrm{Haar} \approx \frac{1}{N} \left(\binom{2k}{k} - 1\right)$. From the variance, we conclude two important features about the statistics. Given $N$,  the variance (i.e. fluctuation) increases with respect to the order of the moment. For each moment, the statistics becomes concentrated as $N$ increases, namely the variance $V_k^\mathrm{Haar}$ vanishes as $N \rightarrow \infty$. By Taylor expansion, the entropy has the same concentrated behavior which can be numerically observed in \cref{fig:converge-to-Haar}. 

\cref{fig:converge-to-Haar} presents the statistics of random quantum circuits for several different structures of the circuit coupling map and shows that for all three coupling maps studied in the main text, the distribution of circuits of sufficient depth converges to the Haar measure. In \cref{fig:converge-to-Haar}(a), we plot the normalized entropy $S/S^\text{Haar}$. The figure shows that the entropy converges to that of Haar measure, i.e., $S/S^\text{Haar} \rightarrow 1$ after the circuit depth of $U$ increases beyond specific values that depend only weakly on the number of qubits $n$. Since the random quantum circuit is constructed by combining layers of  random one- and two-qubit gates, we test the convergence of random quantum circuits for different coupling maps, thereby varying the qubit pairs on which the two-qubit gates can act. The numerical results in \cref{fig:converge-to-Haar}(a) show that coupling maps with greater connectivity converge significantly faster to the Haar measure. We attribute this to the larger number of possible allocations of two-qubit gates enhancing state entanglement within the system and thereby leading to faster mixing of information. 

In addition to showing the convergence in terms of circuit entropy, we also quantify the convergence to the Haar measure for the first five moments in  \cref{fig:converge-to-Haar}(b). The minimal depth to achieve approximate Haar random circuits deduced from the convergence in moments is highly consistent with that derived from the convergence in entropy. We list the depth used in the computation of the sXES in \cref{tab:convergence-depth}. 

\begin{table}[htbp]
    \centering
    \begin{tabular}{@{} *{4}{c} @{}}\midrule
    \headercell{coupling map} & \multicolumn{3}{c@{}}{$n$ (number of system qubits)}\\\cmidrule(l){2-4}
    & 7 & 9 & 11\\\midrule
    linear & 140 & 160 & 160\\
    rectangular & 76 & 94 & 100\\
    fully connected & 60 & 60 & 60\\\midrule
    \end{tabular}
    \caption{Depth for random quantum circuits used in the computation of the system linear cross-entropy score. Each depth is chosen so that both entropy and moments are close to these derived from the Haar measure.}
    \label{tab:convergence-depth}
\end{table}

\begin{figure*}[htbp]
    \centering
    \vspace{-10pt}
    \subfigure[\label{fig:converge-to-Haar-entropy}]{
        \includegraphics[width=0.75\textwidth]{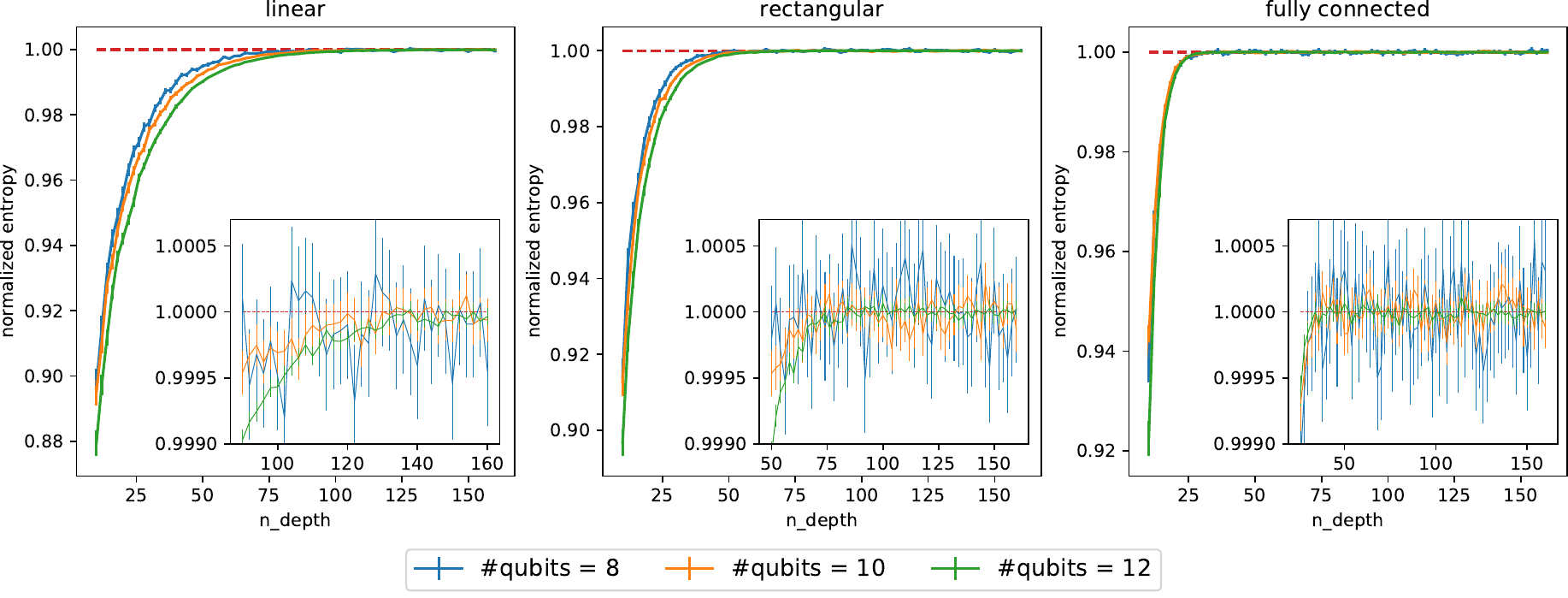}
    }
    
    \subfigure[\label{fig:converge-to-Haar-moment}]{
        \includegraphics[width=0.75\textwidth]{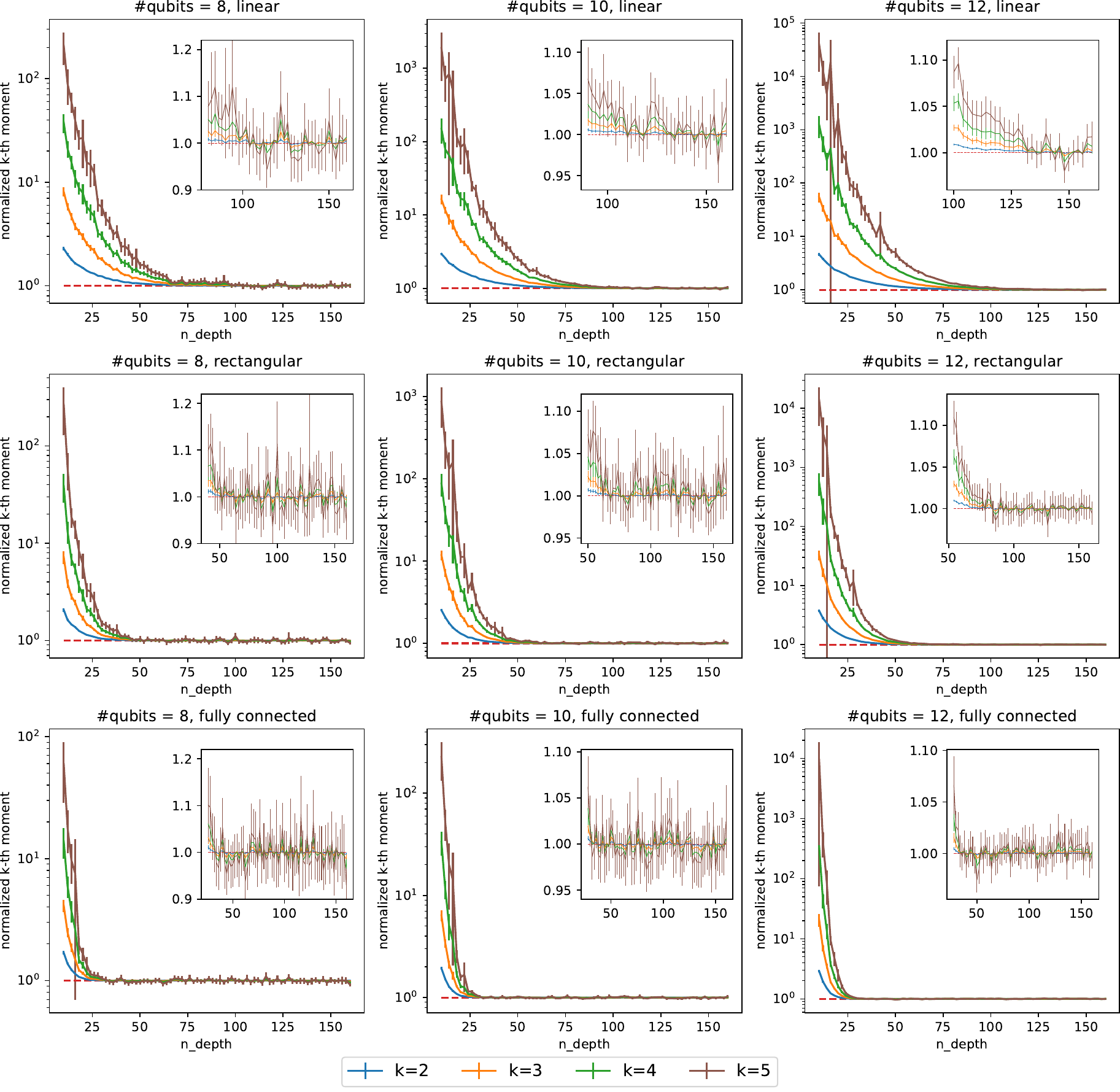}
    }
    \vspace{-10pt}
    \caption{Convergence to the Haar measure. (a) Convergence in terms of entropy. Normalized entropy $S / S^\mathrm{Haar}$ as a function of the depth for random quantum circuits with different number of system qubits and coupling map. (b) Convergence in terms of moments. Each panel is the first five normalized moments $M_k / M_k^\mathrm{Haar}$ as a function of the depth for random quantum circuits with different number of system qubits and coupling map. The convergence of curves to the dashed line at $1$ shows that the random quantum circuit with a modest circuit depth can well approximate the Haar measure. Error bars correspond to the $95\%$ confidence interval estimated from $\sim 1000$ circuit instances.}
    \label{fig:converge-to-Haar}
\end{figure*}

\section{Circuit fidelity from the system linear cross-entropy score}\label{sec:xses}
The system linear cross-entropy score is based on Hamiltonian simulation. Note that in the circuit fidelity of \cref{eqn:circuit-fidelity-XES}, only the terms involved the system linear cross entropy sXES in the numerator contain quantities that must be experimentally evaluated. All other terms can be simplified by using the statistical property of an ensemble of random matrices inherited from the Haar measure of $U$ in \cref{sec:stat-inherit-from-Haar} (in particular \cref{thm:HS-avg}). 

The procedure of computing the system linear cross-entropy score can be summarized as follows.
\begin{enumerate}
    \item Draw quantum circuits $U_i$'s approximately from the Haar measure at random.
    \item For each $U_i$, build the mQSVT circuit for the Hamiltonian simulation benchmark, and measure all qubits to count the bitstring frequencies of $0x$ where $x \in \{0, 1\}^n$. The bitstring frequency is an estimate to $p_\expl(U_i, x)$. Furthermore, the sum of bitstring frequencies for all $x$'s is an estimate to $P_\expl(U_i) = \sum_{x \in \{0,1\}^n} p_\expl(U_i, x)$.
    \item For each $U_i$ and bitstring $0x$, compute the noiseless bitstring probability $p(U_i, x)$ on classical computers.
    \item Compute estimates of the fidelity according to \cref{eqn:circuit-fidelity-QUES,eqn:circuit-fidelity-XES}.
\end{enumerate}

We list the circuit fidelity estimated by different methods in \cref{tab:HSBenchmark}. The agreement shows the consistency of the quantum Hamiltonian simulation benchmark. Here, the theoretical reference value is estimated from the depolarization noise model. Given $U$ with a total of $g_1$ one-qubit gates and $g_2$ two-qubit gates, the value $\alpha_\text{ref} := (1-r_1)^{2d(g_1+1)} (1-r_2)^{2d g_2}$ follows approximately assuming each quantum error fully mixes the quantum state.

\begin{table}[htbp]
\centering
\begin{tabular}{@{} *{7}{c} @{}}\midrule
\headercell{two-qubit gate\\error rate $r_2$} & \multicolumn{6}{c@{}}{QSVT degree parameter $2d$}\\\cmidrule(l){2-7}
& 6 & 8 & 10 & 14 & 18 & 20\\\midrule
\multirow{3}{*}{$4.00\times 10^{-5}$} & \multicolumn{1}{r}{0.95} & \multicolumn{1}{r}{0.92} & \multicolumn{1}{r}{0.90} & \multicolumn{1}{r}{0.85} & \multicolumn{1}{r}{0.80} & \multicolumn{1}{r}{0.79}\\ 
 & \multicolumn{1}{r}{0.92} & \multicolumn{1}{r}{0.90} & \multicolumn{1}{r}{0.87} & \multicolumn{1}{r}{0.83} & \multicolumn{1}{r}{0.79} & \multicolumn{1}{r}{0.76}\\ 
 & \multicolumn{1}{r}{0.93} & \multicolumn{1}{r}{0.92} & \multicolumn{1}{r}{0.89} & \multicolumn{1}{r}{0.86} & \multicolumn{1}{r}{0.82} & \multicolumn{1}{r}{0.80}\\ \midrule
\multirow{3}{*}{$1.30\times 10^{-4}$} & \multicolumn{1}{r}{0.78} & \multicolumn{1}{r}{0.71} & \multicolumn{1}{r}{0.67} & \multicolumn{1}{r}{0.55} & \multicolumn{1}{r}{0.45} & \multicolumn{1}{r}{0.44}\\ 
 & \multicolumn{1}{r}{0.77} & \multicolumn{1}{r}{0.70} & \multicolumn{1}{r}{0.65} & \multicolumn{1}{r}{0.54} & \multicolumn{1}{r}{0.46} & \multicolumn{1}{r}{0.42}\\ 
 & \multicolumn{1}{r}{0.81} & \multicolumn{1}{r}{0.76} & \multicolumn{1}{r}{0.69} & \multicolumn{1}{r}{0.61} & \multicolumn{1}{r}{0.53} & \multicolumn{1}{r}{0.48}\\ \midrule
\multirow{3}{*}{$2.20\times 10^{-4}$} & \multicolumn{1}{r}{0.65} & \multicolumn{1}{r}{0.56} & \multicolumn{1}{r}{0.50} & \multicolumn{1}{r}{0.36} & \multicolumn{1}{r}{0.26} & \multicolumn{1}{r}{0.25}\\ 
 & \multicolumn{1}{r}{0.64} & \multicolumn{1}{r}{0.55} & \multicolumn{1}{r}{0.48} & \multicolumn{1}{r}{0.36} & \multicolumn{1}{r}{0.26} & \multicolumn{1}{r}{0.23}\\ 
 & \multicolumn{1}{r}{0.70} & \multicolumn{1}{r}{0.62} & \multicolumn{1}{r}{0.54} & \multicolumn{1}{r}{0.43} & \multicolumn{1}{r}{0.34} & \multicolumn{1}{r}{0.28}\\ \midrule
\multirow{3}{*}{$3.10\times 10^{-4}$} & \multicolumn{1}{r}{0.54} & \multicolumn{1}{r}{0.43} & \multicolumn{1}{r}{0.38} & \multicolumn{1}{r}{0.24} & \multicolumn{1}{r}{0.15} & \multicolumn{1}{r}{0.15}\\ 
 & \multicolumn{1}{r}{0.54} & \multicolumn{1}{r}{0.43} & \multicolumn{1}{r}{0.35} & \multicolumn{1}{r}{0.23} & \multicolumn{1}{r}{0.15} & \multicolumn{1}{r}{0.12}\\ 
 & \multicolumn{1}{r}{0.60} & \multicolumn{1}{r}{0.51} & \multicolumn{1}{r}{0.41} & \multicolumn{1}{r}{0.30} & \multicolumn{1}{r}{0.22} & \multicolumn{1}{r}{0.17}\\ \midrule
\multirow{3}{*}{$4.00\times 10^{-4}$} & \multicolumn{1}{r}{0.45} & \multicolumn{1}{r}{0.34} & \multicolumn{1}{r}{0.28} & \multicolumn{1}{r}{0.16} & \multicolumn{1}{r}{0.087} & \multicolumn{1}{r}{0.090}\\ 
 & \multicolumn{1}{r}{0.45} & \multicolumn{1}{r}{0.34} & \multicolumn{1}{r}{0.26} & \multicolumn{1}{r}{0.15} & \multicolumn{1}{r}{0.089} & \multicolumn{1}{r}{0.068}\\ 
 & \multicolumn{1}{r}{0.52} & \multicolumn{1}{r}{0.42} & \multicolumn{1}{r}{0.32} & \multicolumn{1}{r}{0.21} & \multicolumn{1}{r}{0.14} & \multicolumn{1}{r}{0.10}\\ \midrule
\end{tabular}
\caption{Circuit fidelity estimated from quantum Hamiltonian simulation benchmark. The total number of qubits is $8$, namely, there are $7$ system qubits and $1$ ancilla qubit. The coupling map is linear. In each cell of the table, the top data is estimated from sXES, the middle data is the theoretical reference value,  and the bottom data is estimated from QUES.
}
\label{tab:HSBenchmark}
\end{table}

\section{Classical hardness of sXHOG}\label{sec:hardness_sxhog}

\begin{definition}[sXHOG, or System Linear Cross-entropy Heavy Output Generation]\label{def:XHOG}
    Given as input a number $b > 1$, a random $(n+1)$-qubit unitary $U$,
    and the mQSVT circuit for the Hamiltonian simulation benchmark with sufficiently small approximation error $\epsilon$,  output nonzero bitstrings $x_1, x_2, \cdots, x_k \in \{0,1\}^n\backslash\{0^n\}$ so that 
\begin{equation}
\label{eqn:def_b_sXHOG}
\frac{1}{k} \sum_{j=1}^k p(U, x) \ge b \times 2^{-n}.
\end{equation}
\end{definition}

The classical hardness of the XEB experiment is justified by reducing the XHOG problem to a complexity assumption referred to as Linear Cross-entropy Quantum Threshold Assumption (XQUATH)\cite{AaronsonGunn2019}. Similarly, the hardness of the sXHOG problem can be reduced to an assumption that we refer to by analogy as sXQUATH.

\begin{definition}[sXQUATH, or System Linear Cross-entropy Quantum Threshold Assumption]\label{def:XQUATH}
        Given a random $(n+1)$-qubit unitary $U$,
    and the mQSVT circuit for the Hamiltonian simulation benchmark with sufficiently small approximation error $\epsilon$, for a uniformly random $x \in \{0,1\}^n\backslash\{0^n\}$, there is no polynomial-time classical algorithm that produces an estimate $p$ of $p_x := p(U, x)$ so that 
        \begin{equation*}
        \expt{(p_x-p)^2} = \expt{(p_x-2^{-n})^2} - \Omega(2^{-3n}).
        \end{equation*}
        Here, the expectation is taken over random circuits $U$, the internal randomness of the algorithm, and the random bitstring $x$.
\end{definition}

The reduction of the XHOG problem is given in the following theorem, which is directly parallel to that in \cite[Theorem 1]{AaronsonGunn2019}.

\begin{theorem}[Classical hardness of sXHOG]
\label{thm:reduction-XHOG-XQUATH}
Assuming sXQUATH, no polynomial-time classical algorithm can solve the XHOG problem in \cref{def:XHOG} with probability $s > \frac{1}{2} + \frac{1}{2b}$, and 
        \begin{equation*}
        k \ge \frac{1}{\left( (2s-1) b - 1 \right) (b - 1)}.
        \end{equation*}
\end{theorem}

\begin{proof}
        Suppose that $\mathsf{A}$ is a classical algorithm solving sXHOG in \cref{def:XHOG} with a success probability $s$ as stated in the theorem. Given $U$ and the mQSVT circuit as the input of $\mathsf{A}$, it outputs $\mathsf{S} := \{x_i \neq 0^n : i = 1, \cdots, k\}$.  
        When $\mathsf{A}$ successfully solves the sXHOG problem, the set $\mathsf{S}$ satisfies \cref{eqn:def_b_sXHOG}. 
        Specifically, let $x \in \{0,1\}^n\backslash \{0^n\}$ be a bitstring sampled uniformly at random.
        We now construct an algorithm to produce an estimate $p$ of $p(U, x)$. 
        Given such a bitstring $x$, the algorithm outputs an estimate $p = b 2^{-n}$ if $x \in \mathsf{S}$ and $p = 2^{-n}$ if $x \notin \mathsf{S}$.
        
        Consider a random variable
        $$X(U, x) := \left(p(U,x)-2^{-n}\right)^2 - \left(p(U,x)-p\right)^2 = \left( 2 p(U,x) - (p + 2^{-n}) \right)\left( p - 2^{-n} \right).$$
        Here, the randomness comes from the uniformly random bitstring $x$, the random unitary $U$ and its corresponding mQSVT circuit, and whether the classical algorithm $\mathsf{A}$ succeeds. We write them explicitly as the subscript of the expectation. Furthermore, we denote by $\mathsf{S}_U^{(s)}:= \mathsf{S}$ when $\mathsf{A}$ succeeds, and $\mathsf{S}_U^{(f)}:= \mathsf{S}$ when $\mathsf{A}$ fails. Let $\II_E$ be the indicator function which gives $1$ if the condition $E$ is satisfied and gives $0$ otherwise. According to the algorithm,
        \begin{equation}
        \begin{split}
        \exptwrt{x,U}{X(U,x) \II_{x \in \mathsf{S}_U^{(s)}} | \mathsf{A}\text{ succeeded}} =& 2 \cdot 2^{-n}(b-1) \cdot \exptwrt{x,U}{p(U,x) \II_{x \in \mathsf{S}_U^{(s)}} \bigg| \mathsf{A}\text{ succeeded}}\\
        &+ 2^{-2n}(1-b^2) \exptwrt{x,U}{\II_{x \in \mathsf{S}_U^{(s)}}},\\
        \exptwrt{x,U}{X(U,x) \II_{x \in \mathsf{S}_U^{(f)}} \bigg| \mathsf{A}\text{ failed}} =& 2 \cdot 2^{-n}(b-1) \cdot \exptwrt{x,U}{p(U,x) \II_{x \in \mathsf{S}_U^{(f)}} \bigg| \mathsf{A}\text{ failed}}\\
        & + 2^{-2n}(1-b^2) \exptwrt{x,U}{\II_{x \in \mathsf{S}_U^{(f)}}}\\
        \ge& 2^{-2n}(1-b^2) \exptwrt{x,U}{\II_{x \in \mathsf{S}_U^{(f)}}}.
        \end{split}
        \end{equation}
        Furthermore, $X(U,x) \equiv 0$ if $x \notin \mathsf{S}$, regardless of  whether $\mathsf{A}$ succeeds or not. By the law of total expectation,
        \begin{equation}
        \begin{split}
                &\exptwrt{x, U, \mathsf{A}}{X(U,x)} = s \cdot \exptwrt{x,U}{X(U,x) | \mathsf{A}\text{ succeeded}} + (1-s) \cdot \exptwrt{x,U}{X(U,x) | \mathsf{A}\text{ failed}}\\
                =& s \cdot \exptwrt{x,U}{X(U,x) \II_{x \in \mathsf{S}^{(s)}_U} \bigg| \mathsf{A}\text{ succeeded}} + (1-s) \cdot \exptwrt{x,U}{X(U,x) \II_{x \in \mathsf{S}^{(f)}_U} \bigg| \mathsf{A}\text{ failed}}\\
                \ge& s \cdot \left(2 \cdot 2^{-n}(b-1) \exptwrt{x,U}{p(U,x) \II_{x \in \mathsf{S}^{(s)}_U} | \mathsf{A}\text{ succeeded}} + 2^{-2n} (1-b^2) \exptwrt{x,U}{\II_{x \in \mathsf{S}_U^{(s)}}}\right)\\
                & + (1-s) \cdot 2^{-2n}(1-b^2) \exptwrt{x,U}{\II_{x \in \mathsf{S}_U^{(f)}}}.
        \end{split}
        \end{equation}
        Here
        \begin{equation*}
                \exptwrt{x,U}{p(U,x) \II_{x \in \mathsf{S}^{(s)}_U} | \mathsf{A}\text{ succeeded}} =\frac{k}{2^n-1}  \exptwrt{U}{\frac{1}{k} \sum_{x \in \mathsf{S}_U^{(s)}} p(U,x)} \ge \frac{bk 2^{-n}}{2^n-1}.
        \end{equation*}
        Note that $\mathsf{S}_U^{(s)}$ and $\mathsf{S}_U^{(f)}$ are sets of $k$ distinct bitstrings. Following that $x$ is uniformly distributed, we have
        \begin{equation*}
                \exptwrt{x, U}{\II_{x \in \mathsf{S}_U^{(s)}}} = \exptwrt{U}{\exptwrt{x}{\II_{x \in \mathsf{S}_U^{(s)}}}} = \frac{k}{2^n-1} \text{ and } \exptwrt{x, U}{\II_{x \in \mathsf{S}_U^{(f)}}} = \exptwrt{U}{\exptwrt{x}{\II_{x \in \mathsf{S}_U^{(f)}}}} = \frac{k}{2^n-1}.
        \end{equation*}
        Then
        \begin{equation}
        \begin{split}
                \exptwrt{x, U, \mathsf{A}}{X(U,x)}& \ge \frac{2^{-2n}}{2^{n}-1} \left( k s \cdot (b-1)^2 + k (1-s) \cdot (1-b^2) \right)\\
                & \ge 2^{-3n} k \left( (2s-1)b - 1 \right) ( b - 1) = \Omega(2^{-3n})
        \end{split}
        \end{equation}
        when $k \ge \frac{1}{\left( (2s-1) b - 1 \right) (b - 1)}$.
        This violates sXQUATH and thereby proves the classical hardness of sHOG.
\end{proof}

\section{Circuit fidelity and sXHOG}\label{sec:fidelity_sxhog}

In this section we demonstrate that the success of sXHOG can be verified by experimental evaluation of the circuit fidelity. Due to the relation between the circuit fidelity and QUES in \cref{sec:xses}, it means that the success of sXHOG can be verified by QUES, which does not involve any classical computation.

First, since we are only interested in the measurement outcome whose ancilla qubit returns $0$, we normalize the bitstring probability as a conditional probability 
\begin{equation}
p_\expl(U,x | \text{ancilla} = 0) := p_\expl(U, x) / P_\expl(U).
\end{equation}
The physical interpretation of the normalization of the bitstring probability is to discard the measurement result whose ancilla is measured with $1$. We also remark that when the Hamiltonian simulation benchmark circuit is sufficiently accurate, we have $P(U) \approx 1$, and it is not necessary to normalize the noiseless bitstring probability in \cref{eqn:def_b_sXHOG}. 

The probability that the experimental measurement on the ancilla qubit outputs $0$ is
\begin{equation*}
\begin{split}
        \mathds{P}_{0} &:= \sum_{x \in \{0,1\}^n} \int p_\expl(U,x) \rd U = \sum_{x \in \{0,1\}^n} \int \alpha p(U,x) + \frac{1-\alpha}{2N} \rd U = \frac{1+\alpha}{2} \approx P_\expl(U).
\end{split}
\end{equation*}
Given the circuit fidelity $\alpha$, we denote the conditional probability density as
\begin{equation*}
 \mathds{Q}_\alpha(U, x) := \frac{1}{\mathds{P}_{0}} p_\expl(U,x),
\end{equation*}
and the corresponding expectation is denoted as $\exptwrt{\mathds{Q}_\alpha}{\cdot}$. Then, the conditional average bitstring probability, which is directly related to the parameter $b$ in determining the sXHOG problem as
\begin{equation}
   b(\alpha) := N\exptwrt{\mathds{Q}_\alpha}{\sum_{x\ne 0^n} p(U,x)} = \frac{\expt{\mathrm{sXES}(U)}}{ \mathds{P}_{0}} = \left( 1 + \frac{\alpha(2-5\mc{H}_1+4\mc{H}_2) - \mc{H}_1}{\alpha + 1} \right) + \Or\left( \frac{1}{N} \right).
\end{equation}
The last equality is derived using results in \cref{sec:stat-inherit-from-Haar}.
Here
\begin{equation}
\mc{H}_1 = \int \bP_\eig^{(2)}(\lambda_1, \lambda_2) \cos\left(t\left(\lambda_1-\lambda_2\right)\right) \rd \lambda_1\rd\lambda_2,
\end{equation}
and
\begin{equation}
\mc{H}_2 = \int \bP_\eig^{(4)}(\lambda_1, \lambda_2, \lambda_3, \lambda_4) \cos\left(t\left(\lambda_1-\lambda_2+\lambda_3-\lambda_4\right)\right) \rd \lambda_1\rd\lambda_2\rd\lambda_3\rd\lambda_4
\end{equation}
are cosine transformations of $\bP_\eig^{(2)}$ and $\bP_\eig^{(4)}$, which are the $2$-marginal and the $4$-marginal distribution of eigenvalues corresponding to the ensemble of random Hermitian matrices, respectively.  The values of $\mc{H}_1$ and $\mc{H}_2$ can be evaluated on classical computers according to \cref{thm:Hl-expansion} and \cref{alg:Fijk}.

Thus for large $n$, we have
\begin{equation}
b(\alpha) \approx 1 + \frac{\alpha(2-5\mc{H}_1+4\mc{H}_2) - \mc{H}_1}{\alpha + 1}=:1+\frac{\gamma(\alpha-\alpha^*)}{\alpha+1}.
\label{eqn:b_fidelity}
\end{equation}
Here $\gamma=2-5\mc{H}_1+4\mc{H}_2$ and $\alpha^*=\mc{H}_1/\gamma$. 
When $\mc{H}_1, \mc{H}_2$ are sufficiently small, $b(\alpha)$ is monotonically increasing. 

The hardness of classical spoofing also requires $b(\alpha) \ge 1$ (see \cref{thm:reduction-XHOG-XQUATH}). Thus, we define the threshold $\alpha^* := \frac{\mc{H}_1}{2-5\mc{H}_1+4\mc{H}_2}$ be the fidelity so that $b(\alpha^*) = 1$. To achieve supremacy, the fidelity is required to satisfy $\alpha \ge \alpha^*$. To see the existence of the threshold fidelity, let us consider a fully contaminated noise where $\alpha = 0$. Then, the average bitstring probability is
\begin{equation}
    \exptwrt{\mathds{Q}_0}{\sum_{x\ne 0^n} p(U,x)} = \frac{1}{N} \left( 1 - \expt{p(U,0^n)} \right) \Rightarrow b(\alpha)|_{\alpha = 0} = 1 - \expt{p(U,0^n)} \le 1.
\end{equation}
By continuity, a threshold fidelity $\alpha^*$ exists to ensure $b(\alpha) > 1$. It also indicates that the threshold is very close to zero when the diagonal elements of the time evolution vanish simultaneously in the ensemble, namely $\expt{p(U,0^n)} \approx 0$. The threshold can be suppressed by choosing a larger simulation time $t$ since $\alpha^* \rightarrow 0$ as $t \rightarrow \infty$. Furthermore, when $\alpha^* \ll 1$, the conditional average bitstring probability is
\begin{equation*}
    \exptwrt{\mathds{Q}_\alpha}{\sum_{x\ne 0^n} p(U,x)} = \frac{1}{N}\left(1 + \frac{2(\alpha-\alpha^*)}{1+\alpha}\right) + \Or\left(\frac{1}{N^2}\right).
\end{equation*}
Note that at $t^\text{opt}$, the threshold $\alpha^*|_{t^\text{opt}} \approx \mc{H}_1 / 2$. \cref{eqn:HS-avg-1ord} implies that $\mc{H}_1 \ge - \frac{2}{N-1}$. Therefore, when the number of qubits $n$ is not sufficiently large, $\alpha^*|_{t^\text{opt}}$ can possibly be negative. However, as $n$ increases, $\alpha^*|_{t^\text{opt}}$ converges to $0$ exponentially fast because the lower bound $- \frac{2}{N-1} \to 0$ in the large $n$-limit. This agrees with the numerical behavior of the threshold $\alpha^*$ shown in \cref{fig:supremacy-region-HS}. 

\section{Analytic estimation of $t^{\text{opt}}$ for large $n$}\label{sec:topt_derive}

According to the result in \cref{fig:supremacy-region-HS}, at $t=t^{\text{opt}}$, we have
\[
\expt{p_t(U,0^n)}=\mathbb{E}|\braket{0^n|e^{-\I \mf{H} t}|0^n}|^2\approx 0.
\]
Jensen's inequality gives
\[
\abs{\mathbb{E}\braket{0^n|e^{-\I \mf{H} t}|0^n}}^2\le \mathbb{E}|\braket{0^n|e^{-\I \mf{H} t}|0^n}|^2\approx 0.
\]
If $\mf{H}$ is a H-RACBEM and the corresponding $U$ is drawn from the Haar measure, then
\[
\mathbb{E}\braket{0^n|e^{-\I \mf{H} t}|0^n}=\int_{0}^1 e^{-\I \lambda t} \bP_\eig^{(1)}(\lambda)\ud \lambda.
\]
Here $\bP_\eig^{(1)}(\lambda)$ is defined in defined in \cref{eqn:Peig_ell} with $\ell=1$.
It is also the $1$-marginal (a.k.a. the level density) of the joint probability distribution of all eigenvalues in \cref{eqn:P_joint}, which is called a $\beta$-Jacobi ensemble with $\beta=2$~\cite{DumitriuEdelman2002}. 
With the block encoding of one ancilla qubit (i.e. $M=2$), the level density follows the $\operatorname{Beta}(0.5,0.5)$ distribution in the large $n$-limit
\cite{Leff1964}, i.e. in the sense of weak convergence, we have
\[
\lim_{n\to \infty} \bP_\eig^{(1)}(\lambda)=
\frac{1}{\pi}\lambda^{-\frac12}(1-\lambda)^{-\frac12}.
\]
Therefore for any $t$, 
\begin{equation}
\begin{aligned}
\lim_{n\to \infty} \int_{0}^1 e^{-\I \lambda t} \bP_\eig^{(1)}(\lambda)\ud \lambda
=&\int_{0}^1 e^{-\I \lambda t} \frac{1}{\pi}\lambda^{-\frac12}(1-\lambda)^{-\frac12}\ud \lambda\\
\overset{\lambda=\sin^2 \left(\frac{\theta}{2}\right)}{=\mathrel{\mkern-5mu}=\mathrel{\mkern-5mu}=\mathrel{\mkern-5mu}=}& \quad\frac{1}{\pi}\int_{0}^{\pi}
\exp\left(-\I t \sin^2 \frac{\theta}{2} \right) \ud \theta\\
=& e^{-\frac{\I t}{2}} \frac{1}{2\pi} \int_{-\pi}^{\pi} e^{\I \frac{t}{2} \cos \theta} \ud \theta\\
=& e^{-\frac{\I t}{2}} J_0(t/2).
\end{aligned}
\label{eqn:topt_derive}
\end{equation}
Here we have used the integral representation of the Bessel function of the first kind
\[
J_0(t/2)=\frac{1}{2\pi}\int_{-\pi}^{\pi} e^{\I \frac{t}{2} \sin \theta} \ud \theta=\frac{1}{2\pi}\int_{-\pi}^{\pi} e^{\I \frac{t}{2} \cos \theta} \ud \theta.
\]
Therefore in the large $n$ limit, $\mathbb{E}\braket{0^n|e^{-\I \mf{H} t}|0^n}$ approximately vanishes at the first node of $J_0(t/2)$, which gives
\[
t^{\text{opt}}\approx 4.81.
\]
This agrees very well with the numerical results in \cref{fig:supremacy-region-HS,fig:mc-justify-topt}.

\begin{figure*}[htbp]
    \centering
    \vspace{-10pt}
    \subfigure[]{
        \includegraphics[width=.3\textwidth]{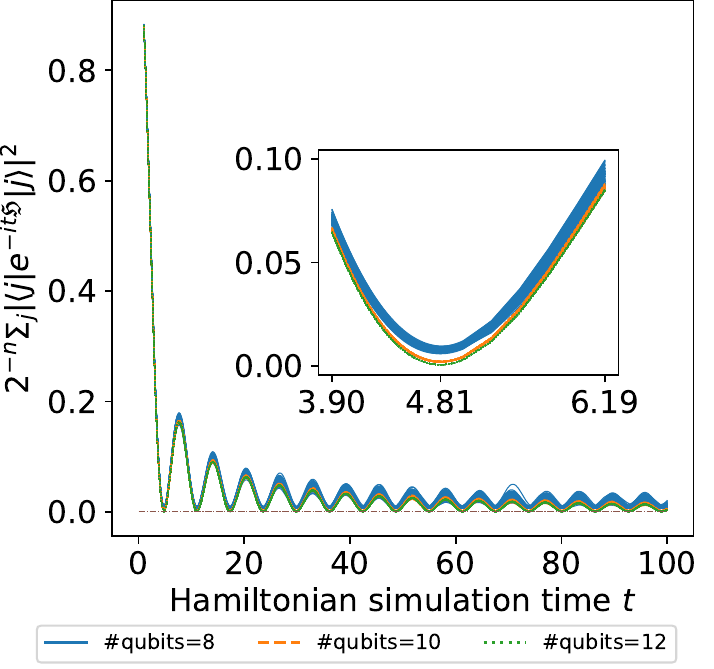}
    }
    \subfigure[]{
        \includegraphics[width=0.46\textwidth]{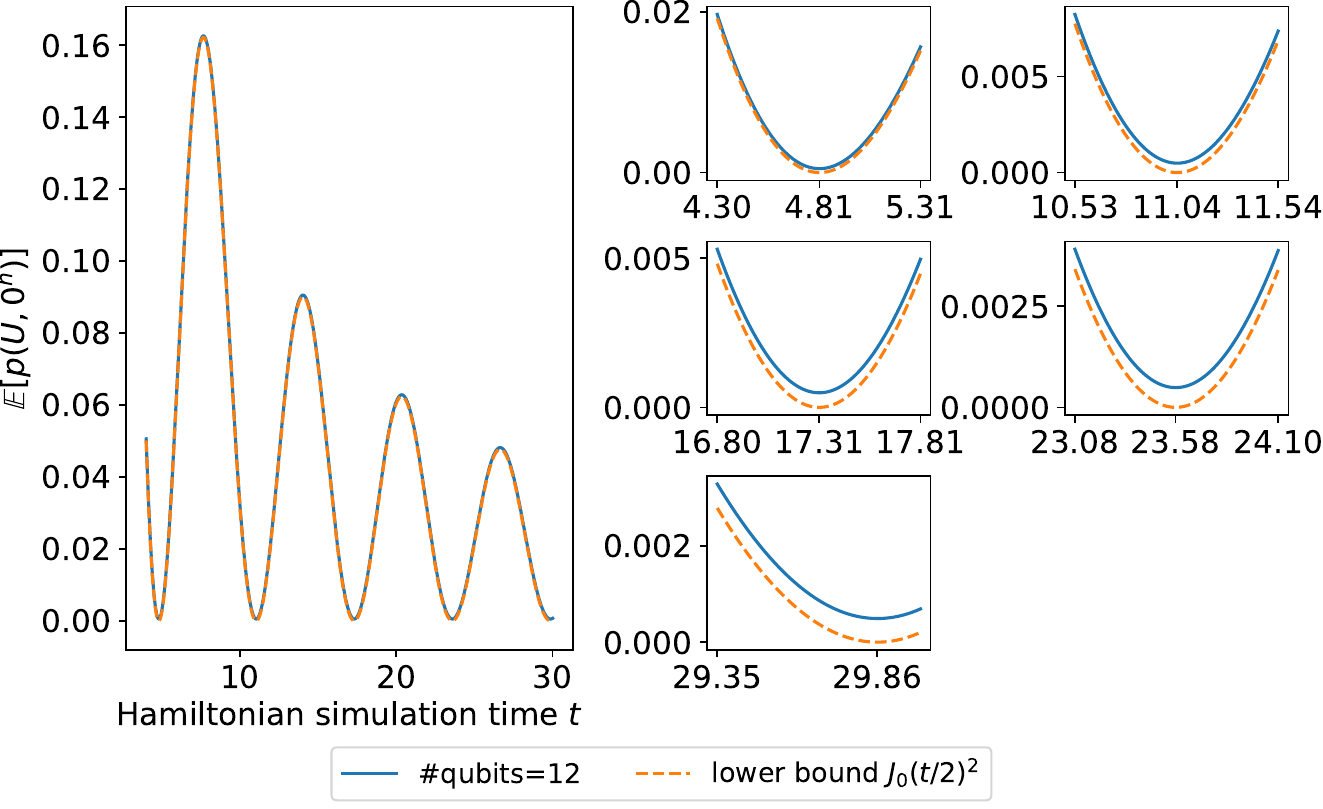}
    }
    \vspace{-10pt}
    \caption{Numerical justification of $t^\mathrm{opt}$. (a) The trajectory of the average diagonal probability $2^{-n} \sum_j \abs{\braket{j | e^{-\I t \mf{H}} | j}}^2$ as a function of Hamiltonian simulation time $t$. The broadening is due to plotting $\sim 100$ instances individually. The subfigure in the box shows the behavior near $t^{\text{opt}} \approx 4.81$. (b) The average probability by measuring all qubits with $0$ and the analytical lower bound on it. Zooming into the first five zeros of the Bessel function, the minima of the average probability well agree these zeros.}
    \label{fig:mc-justify-topt}
\end{figure*}

\section{Additional analytical computations and proofs}
\subsection{Statistical property of the random-matrix ensemble inherited from Haar measure}\label{sec:stat-inherit-from-Haar}
The solution of the system heavy output generation problem and analytic evaluation of the system linear cross-entropy score require the use of statistical properties of the ensemble of random matrices obtained from the Haar measure. In this section, we  derive the statistical properties of the ensemble. We consider a generic block encoding with $m$ extra ancilla qubits, namely, an $n$-qubit matrix $A$ is a submatrix of an $(n+m)$-qubit unitary $U$. We use $M = 2^m$ to represent the dimension of the Hilbert space generated by $m$ ancilla qubits. We assume that  $U$ is drawn from an $(\nsys+\nbe)$-qubit Haar measure. 

Given the identification $\CC^{\Nsys \times \Nsys} \simeq \CC^{\Nsys^2}$, the uniform measure on the space of complex matrices is identified as the pushforward of the Lebesgue measure on $\CC^{\Nsys^2}$, for example, by taking the coordinate system as matrix elements. We denote this uniform measure as $\rd A$. Assuming that $A$ is an $\nsys$-qubit submatrix of a Haar-distributed $(\nsys+\nbe)$-qubit unitary $U$, the first theorem gives a characterization of the induced probability distribution of $A$.
\begin{theorem}[{\cite[Theorem 1.3.1]{Collins2003}}]
        \label{thm:pdf-matrix}
        Let $A \in \CC^{\Nsys\times\Nsys}$ be a submatrix block encoded in a Haar unitary. Then the probability density is 
\begin{equation}
\bP(A) = \mc{Z}^{-1} \det \left( I - A^\dagger A \right)^{\Nsys(\Nbe-2)} \II_{\norm{A}_2 \leq 1},
\end{equation}
where $\mc{Z} := \int_{\norm{A}_2 \le 1} \det \left( I - A^\dagger A \right)^{\Nsys(\Nbe-2)} \rd A$ is a normalization constant. Here, $\II$ is an indicator function. It gives $1$ when the condition in the subscript is satisfied, and gives $0$ otherwise. Generically, let $A$ be an $n_1$-by-$n_2$ submatrix of an $n$-by-$n$ Haar-distributed unitary $U$, and $n \ge n_1+n_2$. Then, the probability density is
      \begin{equation}\label{eqn:matrix-distribution-general}
          \bP(A) \propto \det\left(I - A^\dagger A\right)^{n-n_1-n_2} \II_{\norm{A}_2 \le 1}.
      \end{equation}
\end{theorem}
In particular, for $1$-block-encoded matrix with $\nbe = 1$, the exponent of the determinant is $0$, and $A$ is uniformly distributed in the unit ball $\{A \in \CC^{\Nsys \times \Nsys} : \norm{A}_2 \le 1\}$. Let us consider $A = W \Sigma V^\dagger$ where $W \in \ugrp(\Nsys)/\ugrp(1)^\Nsys, V \in \ugrp(\Nsys)$ and $\diag \Sigma = (\sigma_1, \cdots, \sigma_\Nsys)$. The Jacobian of this decomposition is $\rd A \propto \rd V \rd W \left(\Delta(\sigma_1^2, \cdots, \sigma_\Nsys^2)^2 \prod_{j=1}^\Nsys \sigma_j \rd \sigma_j\right)$ where $\rd W, \rd V$ are the Haar measure on their compact manifolds respectively and $\Delta(x_1, \cdots, x_n) := \prod_{i < j} (x_i - x_j)$ is the Vandermonde determinant. Then, the distribution of $V, W, \Sigma$ follow immediately the theorem.
\begin{corollary}\label{cor:prob-dist-svd}
    Let $W, V$ be Haar-distributed. The joint distribution of all singular values has the density
    \begin{equation}
        \begin{split}
            \bP\left( \sigma_1, \cdots, \sigma_\Nsys \right) \propto& \ \Delta(\sigma_1^2, \cdots, \sigma_\Nsys^2)^2 \prod_{i = 1}^\Nsys \sigma_i \left( 1 - \sigma_i^2 \right)^{\Nsys(\Nbe-2)} \II_{\sigma_i \in [0, 1]}.
        \end{split}
    \end{equation}
\end{corollary}
Since $\mf{H} = A^\dagger A$, the eigenvalue $\lambda_j$ of the random Hermitian matrix $\mf{H}$ and the singular value $\sigma_j$ of the complex matrix $A$ is related by $\lambda_j = \sigma_j^2$. By a direct change-of-variable, the joint distribution of all eigenvalues has the density
\begin{equation}
    \begin{split}
        \bP\left( \lambda_1, \cdots, \lambda_\Nsys \right) =& \mc{Z}_\text{eig}^{-1}\Delta(\lambda_1, \cdots, \lambda_\Nsys)^2 \prod_{i = 1}^\Nsys \left( 1 - \lambda_i \right)^{\Nsys(\Nbe-2)} \II_{\lambda_i \in [0, 1]}.
    \end{split}
    \label{eqn:P_joint}
\end{equation}
The normalization constant is precisely given by the Selberg's integral \cite{Selberg1944}
\begin{equation}
\mc{Z}_\text{eig} = \prod_{j=0}^{\Nsys-1} \frac{\Gamma(j + 1) \Gamma(j+2) \Gamma(j + \Nsys(\Nbe-2)+1)}{\Gamma(j+\Nsys(\Nbe-1)+1)}.
\end{equation}
  The distribution is invariant under the relabeling of eigenvalues $(\lambda_1, \cdots, \lambda_\Nsys) \mapsto (\lambda_{\pi(1)}, \cdots, \lambda_{\pi(\Nsys)})$ for any permutation $\pi$. This feature is inherited from the bi-invariance of the Haar measure on the compact Lie group. 

In practice, only the marginal distribution involving a few eigenvalues will be used. However, the Vandermonde determinant in the joint distribution couples all eigenvalues together, which makes it hard to compute the marginal distribution analytically. Nonetheless, a semi-analytical representation by orthogonal polynomial expansion can be derived as follows.

Let $w(x) := (1-x)^{\Nsys(\Nbe-2)}$ be the weight function, $(f, g)_w := \int_0^1 f(x) g(x) w(x) \rd x$ be the weighted inner product on $[0, 1]$, and $\norm{f}_w := \sqrt{(f,f)_w}$ be the weighted norm. 
\begin{theorem}[{\cite[Theorem 5.7.1]{MehtaRM2004}}]
        \label{thm:numeigdist-RACBEM}
        Let $\{C_i(x) : \deg C_i = i,\ i = 0, \cdots, \Nsys-1\}$ be a set of linearly independent monic polynomials such that they are orthogonal with respect to $(\cdot,\cdot)_w$. Let $c_i := \norm{C_i}_w^2$. Define a bivariate function 
\begin{equation}
K(x,y) := w(x) \sum_{i = 0}^{\Nsys-1} \frac{1}{c_i} C_i(x) C_i(y).
\end{equation}
        Then, the joint distribution for $\ell$ eigenvalues follows a determinantal process
        \begin{equation}
                \bP_\eig^{(\ell)}(\lambda_1, \cdots, \lambda_\ell) = \frac{(\Nsys-\ell)!}{\Nsys!} \det\left[ K(\lambda_{j_1}, \lambda_{j_2}) \right]_{j_1,j_2 = 1, \cdots, \ell}.
\label{eqn:Peig_ell}
        \end{equation}
\end{theorem}
The orthogonal polynomial can be generated by $3$-point recursion formula. Specifically, for 1-block-encoding (i.e. $\nbe=1$ and $M=2$), the weight function is $w \equiv 1$, and the orthogonal polynomial is the shifted Legendre polynomial 
\begin{equation}
C_i(x) \propto P_i(2x-1).
\end{equation}

\begin{corollary}\label{cor:1-block-encoding-dist-kernel}
    For 1-block-encoding, the joint eigenvalue distribution can be expressed in terms of the bivariate function
    \begin{equation}
        K(x, y) = \sum_{i=0}^{N-1} (2i+1) P_i(2x-1) P_i(2y-1),
    \end{equation}
    where $P_i$ is the $i$-th Legendre polynomial.
\end{corollary}
With \cref{cor:1-block-encoding-dist-kernel}, the averages of bitstring probability can be evaluated semi-analytically. For a generic complex even polynomial, we define
\begin{equation}
\begin{split}
    &\mc{R}_{k_1,\cdots,k_{\ell_1} | r_1,\cdots,r_{\ell_3}}^{q_1,\cdots,q_{\ell_2}} := \expt{\prod_{j=1}^{\ell_1} g^{k_j}(\sigma_j) \prod_{j=1}^{\ell_2} \overline{g^{q_j}(\sigma_{\ell_1+j})} \prod_{j=1}^{\ell_3} \abs{g^{r_j}(\sigma_{\ell_1+\ell_2+j})}}.
\end{split}
\end{equation}
A complex polynomial $f \in \CC[x]$ can be determined by setting $f(x^2) = g(x)$. Note that $g(\sigma_j) = f(\lambda_j)$ relates the singular value transformation and the eigenvalue transformation. The expectation can be expresses exactly by the integration with joint distribution,
\begin{equation}\label{eqn:tensor-R}
    \mc{R}_{k_1,\cdots,k_{\ell_1} | r_1,\cdots,r_{\ell_3}}^{q_1,\cdots,q_{\ell_2}} = \int_{[0,1]^{\ell_1+\ell_2+\ell_3}} \bP_\eig^{(\ell_1+\ell_2+\ell_3)} \prod_{j=1}^{\ell_1} f^{k_j}(\lambda_j) \prod_{j=1}^{\ell_2} \overline{f^{q_j}(\lambda_{\ell_1+j})} \prod_{j=1}^{\ell_3} \abs{f^{r_j}(\lambda_{\ell_1+\ell_2+j})} \rd \lambda_1 \cdots \rd \lambda_{\ell_1+\ell_2+\ell_3}.
\end{equation}
For Hamiltonian simulation, $e^{-\I t x^2}$ has unit absolute value for all $x\in\RR$. 
For simplicity we assume the approximation error is sufficiently small, and $g(x) = s_t(x)=e^{-\I t x^2}$, or equivalently $f(x) = e^{-\I t x}$. This allows us to omit the terms due to $\abs{f^{r_j}(\lambda_{\ell_1+\ell_2+j})}$. Furthermore, when the upper and lower indices of $\mc{R}$ in \cref{eqn:tensor-R} are the same $k_j = q_j = 1$, the relabeling invariance of eigenvalues implies that the defined quantity is reduced to
\begin{equation}
\begin{split}
    \mc{H}_{\ell}(t) =& \expt{\prod_{j=1}^\ell g(\sigma_j) \overline{g(\sigma_{\ell+j})}} = \int_{[0,1]^{2\ell}} \bP_\eig^{(2\ell)}(\lambda_1,\cdots,\lambda_{2\ell}) \cos\left(t \sum_{j=1}^{\ell} \lambda_j - \lambda_{\ell+j} \right) \rd \lambda_1 \cdots \rd \lambda_{2\ell},
\end{split}
\end{equation}
which is directly related to the Hamiltonian simulation. When the time dependence is irrelevant to the analysis, we drop the $t$ dependence in $\mc{H}_\ell(t)$ by writing it as $\mc{H}_\ell$ for simplicity. We can compute $\mc{H}_{\ell}$ as follows.

\begin{theorem}\label{thm:Hl-expansion}
    Let the degree-$d$ complex polynomial $f(x) = \sum_{q=0}^d c_q P_q(2x-1)$ be decomposed in terms of Legendre polynomials. Then,
    \begin{equation}
    \begin{split}
        \mc{H}_\ell =& \frac{(N-2\ell)!}{N!} \sum_{k_1=0}^{N-1} \cdots \sum_{k_{2\ell}=0}^{N-1} \sum_{\varsigma \in \mathsf{S}_{2\ell}} \operatorname{sgn}(\varsigma) \prod_{j=1}^{2\ell} \left( (2k_j+1) \sum_{q=0}^d C_q^{(j)} F_{q, k_j, k_{\varsigma^{-1}(j)}} \right).
    \end{split}
    \end{equation}
    Here, $\mathsf{S}_{2\ell}$ is the symmetric group,
    \begin{equation*}
        C_q^{(j)} = \left\{\begin{array}{ll}
            c_q &,\text{ if } j \le \ell,\\
            \overline{c_q} &,\text{ otherwise},
        \end{array}\right.
    \end{equation*}
    and
    \begin{equation}
    \begin{split}
        F_{i,j,k} &= \frac{1}{2} \int_{-1}^1 P_i(x) P_j(x) P_k(x) \rd x = \left\{
        \begin{array}{l}
            \frac{(2s-2i)!(2s-2j)!(2s-2k)!}{(2s+1)!} \left(\frac{s!}{(s-i)!(s-j)!(s-k)!}\right)^2,  \\
            \quad  \text{ if } 2s = i+j+k \text{ is even and }\abs{i-j} \le k \le i+j,\\
            0 ,\text{ otherwise}.
        \end{array}
        \right.
    \end{split}
    \label{eqn:F_ijk}
    \end{equation}
\end{theorem}
\begin{proof}
    Let $f^{(j)}(x) = \sum_{q=0}^d C_q^{(j)} P_q(2x-1)$ so that $f^{(j)}(x) = f(x)$ when $j \le \ell$ and $f^{(j)}(x) = \overline{f(x)}$ when $j > \ell$. Then, directly applying \cref{thm:numeigdist-RACBEM} and \cref{cor:1-block-encoding-dist-kernel}, the quantity can be evaluated immediately,
    \begin{equation*}
        \begin{split}
            \frac{N!}{(N-2\ell)!} \mc{H}_\ell =& \sum_{\varsigma \in \mathsf{S}_{2\ell}} \text{sgn}(\varsigma) \int_{[0,1]^{2\ell}} \prod_{j=1}^{2\ell} f^{(j)}(x_j) K(x_j, x_{\varsigma(j)}) \rd x_j\\
            =& \sum_{\varsigma \in \mathsf{S}_{2\ell}} \text{sgn}(\varsigma) \sum_{k_1=0}^{N-1} \cdots \sum_{k_{2\ell}=0}^{N-1} \prod_{j=1}^{2\ell} (2k_j+1) \int_0^1 f^{(j)}(x) P_{k_j}(2x-1) P_{k_{\varsigma^{-1}(j)}}(2x-1) \rd x.
        \end{split}
    \end{equation*}
    The conclusion follows.
\end{proof}
The constraint in \cref{eqn:F_ijk} will be referred to as the triangle rule. Due to the triangle rule, $F_{i,j,k}$ is a sparse tensor. Many terms in the $(2\ell)$-fold summation vanishes, which can be used to accelerate the evaluation. Note that Legendre polynomials are bounded by $1$ on $[-1,1]$, which implies that $\abs{F_{i,j,k}} \le 1$. To circumvent the numerical instability arising from factorials, $F_{i,j,k}$ can be evaluated recursively,
\begin{equation*}
\begin{split}
    & F_{0,0,0} = 1,\ F_{i,j+1,k+1} = \frac{2s+1-2i}{s+1-i} \frac{s+1}{2s+3} F_{i,j,k},\\
    & F_{i,j,k+2} = \frac{2s+1-2i}{s+1-i} \frac{2s+1-2j}{s+1-j} \frac{s-k}{2s-2k-1}\frac{s+1}{2s+3} F_{i,j,k},
\end{split}
\end{equation*}
where $2s = i+j+k$. Using \cref{alg:Fijk}, $F_{i,j,k}$ can be evaluated stably with $s$ recursions.

\begin{algorithm}[htbp]
\hrule
\caption{A stable recursive algorithm for computing $F_{i,j,k}$}
\label{alg:Fijk}
\begin{algorithmic}[0]
\STATE \textbf{Input:} A triplet $(i,j,k)$ satisfying  the triangle rule.
\vspace{1em}
\STATE Order and relabel the triplet so that $i \le j \le k$.
\STATE Set $2s=i+j+k$.
\IF{$k \ge j-i+2$}
\STATE Recursively call the algorithm to compute $F_{i,j,k-2}$. Note $(i,j,k-2)$ preserves the triangle rule.
\STATE \textbf{Return:} $F_{i,j,k} = \frac{2s-1-2i}{s-i}\frac{2s-1-2j}{s-j}\frac{s-k+1}{2s-2k+1}\frac{s}{2s+1} F_{i,j,k-2}$.
\ELSIF{$k<j-i+2$ and $i < j$}
\STATE Recursively call the algorithm to compute $F_{i,j-1,k-1}$. Note $(i,j-1,k-1)$ preserves the triangle rule.
\STATE \textbf{Return:} $F_{i,j,k} = \frac{2s-1-2i}{s-i}\frac{s}{2s+1} F_{i,j-1,k-1}$.
\ELSE{ $i=j=k=0$}
\STATE \textbf{Return:} $F_{0,0,0} = 1$.
\ENDIF
\end{algorithmic}
\end{algorithm}

The measurement on all qubits gives an $(n+1)$-bit string. We are interested in the bitstring $0x$ which means the outcome of the ancilla qubit is $0$ and that of $n$ system qubits is $x \in \{0, 1\}^n$. The  probability of measuring the bitstring $0x$ is
\begin{equation*}
\begin{split}
    p(U,x) &= \abs{\braket{0x | \circqsvt_{f,U} | 00^\nsys}}^2 = \abs{\braket{x|g^\svt(A)|0^\nsys}}^2 = \sum_{j,k = 0}^{\Nsys-1} g(\sigma_j) \overline{g(\sigma_k)} V_{j0} \overline{V_{jx}}\, \overline{V_{k0}} V_{kx}.
\end{split}
\end{equation*}
When $x =  0^n\equiv 0$, $p(U,x) = \sum_{j,k} g(\sigma_j) \overline{g(\sigma_k)} \abs{V_{j0}}^2 \abs{V_{k0}}^2$ involves only one column of a Haar unitary $V$. Yet when $x \ne 0^n$, the  probability involves two columns. For $x \ne 0^n \text{ or } 1^n$, let us consider another unitary $\wt{V}$ by permuting the column $1$ and column $x$ of $V$. By the bi-invariance of Haar measure on unitary group, $\wt{V}$ and $V$ are identically distributed. Therefore, we conclude the following lemma.
\begin{lemma}\label{lma:iddist-relabel}
    For any nonzero bitstring $0^n \ne x \in \{0, 1\}^\nsys$, $p(U,x)$ is identically distributed.
\end{lemma}
According to \cref{cor:prob-dist-svd}, the distributions of $\Sigma$ and $V$ are decoupled. The average over the singular values can be evaluated semi-analytically, and the average over the Haar unitary can be analytically computed by using representation theory. We conclude the relevant average values as follows.
\begin{theorem}\label{thm:HS-avg}
    For Hamiltonian simulation benchmark, the averages of bitstring probability are
\begin{equation}\label{eqn:HS-avg-1ord}
\begin{split}
    &\expt{p(U,0^n)} = \frac{\Nsys-1}{\Nsys+1} \mc{H}_1 + \frac{2}{\Nsys+1}, \quad \expt{\sum_{x\ne 0^n} p(U,x)} = \frac{N-1}{N+1} \left(1 - \mc{H}_1\right),
\end{split}
\end{equation}
and 
\begin{equation}
    \begin{split}
        &\expt{p(U,0^n)^2} = \frac{12}{(N+2)(N+3)} + \frac{12 N (N-1) \mc{H}_1}{(N+1)(N+2)(N+3)} + \frac{(N-1)(N-2)(N-3)}{(N+1)(N+2)(N+3)} \mc{H}_{2},\\
        &\expt{\sum_{x\ne 0^n} p(U,x)^2} = \frac{2(N-1)(N^2+3N+6)}{N(N+1)(N+2)(N+3)} - \frac{4(N-1)(N^2-N+6)\mc{H}_1}{N(N+1)(N+2)(N+3)} + \frac{2(N-1)(N-2)(N-3)}{N(N+1)(N+2)(N+3)} \mc{H}_{2}.
    \end{split}
    \label{eqn:prob_secondorder}
\end{equation}
\end{theorem}
\begin{proof}
    We first evaluate the first order moments in \cref{eqn:HS-avg-1ord}. Let $p_j =: \abs{V_{j0}}^2$. Directly applying \cref{thm:pdf-matrix} to the $0$-th column of $V$, the joint probability density of $k$ distinct success probabilities is
    \begin{equation*}
        \bP(p_1,\cdots, p_k) = \prod_{j=0}^{k-1}(N-k+j) \left(1-\sum_{j=1}^k p_j\right)^{N-k-1} \II_{\sum_{j=1}^k p_j < 1}.
    \end{equation*}
    Given an index set $\boldsymbol{\alpha} := (\alpha_1, \cdots, \alpha_k)$ and $\abs{\boldsymbol{\alpha}} := \sum_{j=1}^k \alpha_j$, the average $\mc{I}_{\boldsymbol{\alpha}} := \expt{\prod_{j=1}^k p_j^{\alpha_j}} = \left(\prod_{j=1}^k \alpha_j!\right) \prod_{j=0}^{\abs{\boldsymbol{\alpha}}-1} \frac{1}{N+j}$ follows direct computation,
    \begin{equation*}
        \begin{split}
        & \left(\prod_{j=0}^{k-1} \frac{1}{N-k+j}\right) \mc{I}_{\boldsymbol{\alpha}} =\int_{\sum_{j=1}^{k-1} p_{j}<1} \prod_{j=1}^{k-1} p_{j}^{\alpha_{j}} \rd p_{j} \int_{0}^{1-\sum_{j=1}^{k-1} p_{j}} p_{k}^{\alpha_{k}}\left(1-\sum_{j=1}^{k-1} p_{j}-p_{k}\right)^{N-k-1} \rd p_{k} \\
        &\quad=\alpha_{k} ! \prod_{j=0}^{\alpha_{k}} \frac{1}{N-k+j} \int_{\sum_{j=1}^{k-1} p_{j}<1}\left(1-\sum_{j=1}^{k-1} p_{j}\right)^{N-k+\alpha_{k}} \prod_{j=1}^{k-1} p_{j}^{\alpha_{j}} \rd p_{j} = \cdots = \left(\prod_{j=1}^{k} \alpha_{j} !\right) \prod_{j=0}^{\abs{\boldsymbol{\alpha}}+k-1} \frac{1}{N-k+j}.
        \end{split}
    \end{equation*}
    For the second equal sign, we use the identity
    \begin{equation}
    \int_0^y x^\alpha(y-x)^\beta \rd x = y^{\alpha+\beta+1} \text{B}(\alpha+1,\beta+1) = y^{\alpha+\beta+1}\frac{\alpha!\beta!}{(\alpha+\beta+1)!},
    \end{equation}
    where the Beta function is
    \begin{equation*}
        \text{B}(x,y) := \int_0^1 t^{x-1}(1-t)^{y-1} \rd t = \frac{\Gamma(x) \Gamma(y)}{\Gamma(x+y)}.
    \end{equation*}
By definition, $p(U,0^n) = \sum_{j,k} g(\sigma_j) \overline{g(\sigma_k)} p_j p_k$. Applying these results for average and using the fact that singular values and singular vectors are independent, the average bitstring probability is
    \begin{equation*}
        \expt{p(U,0^n)} = N \expt{\abs{g(\sigma_1)}^2} \mc{I}_{(2)} + N(N-1) \expt{\frac{1}{2} \left(g(\sigma_1)\overline{g(\sigma_2)} + \overline{g(\sigma_1)} g(\sigma_2)\right)} \mc{I}_{(1,1)} = \frac{N-1}{N+1} \mc{H}_1 + \frac{2}{N+1}.
    \end{equation*}
    Then,
    \begin{equation*}
        \expt{\sum_{x\ne 0^n} p(U,x)} = 1 - \expt{p(U,0^n)} = \frac{N-1}{N+1}\left(1 - \mc{H}_1\right).
    \end{equation*}
    
Now we evaluate the second order moments in \cref{eqn:prob_secondorder}. Let $\pi \in \mathsf{S}_4$ be any permutation, and $\mathsf{R}$ be any set of constraints on four $n$-bit binary strings $i,j,k,l$. We denote the action of the symmetric group as $\pi \cdot \mathsf{R} := \{ (\pi(i), \pi(j), \pi(k), \pi(l)) : (i,j,k,l) \in \mathsf{R} \}$. Due to the relabeling invariance of the joint distribution of singular values, we have 
\begin{equation}
        \sum_{(i,j,k,l) \in \pi \cdot \mathsf{R}} \expt{g(\sigma_i)\overline{g(\sigma_j)}g(\sigma_k)\overline{g(\sigma_l)}} = \sum_{(i,j,k,l) \in \mathsf{R}} g(\sigma_{\pi(i)})\overline{g(\sigma_{\pi(j)})}g(\sigma_{\pi(k)})\overline{g(\sigma_{\pi(l)})} = \sum_{(i,j,k,l) \in \mathsf{R}} g(\sigma_i)\overline{g(\sigma_j)}g(\sigma_k)\overline{g(\sigma_l)}.
\label{eqn:relabeling_equality}
    \end{equation}
    For example, let $\mathsf{R}_1 := \{(i=j) \ne (k=l)\}, \mathsf{R}_2 := \{(i=k) \ne (j=l)\}, \mathsf{R}_3 := \{(i=l) \ne (j=k)\}$, and $\pi_1 := \left(\begin{array}{*4c}
        i & j & k & l \\
        i & k & j & l
    \end{array}\right), \pi_2 := \left(\begin{array}{*4c}
        i & j & k & l \\
        i & l & j & k\end{array}\right)$. Then, $\pi_1\cdot\mathsf{R}_1 = \mathsf{R}_2$ and $\pi_2\cdot\mathsf{R}_1 = \mathsf{R}_3$ holds. Therefore
        \begin{equation*}
            \sum_{(i,j,k,l) \in \mathsf{R}_{\ell}} \expt{g(\sigma_i)\overline{g(\sigma_j)}g(\sigma_k)\overline{g(\sigma_l)}} = N(N-1) \expt{\abs{g(\sigma_1)}^2 \abs{g(\sigma_2)}^2} = N(N-1), \quad \ell=1,2,3.
        \end{equation*}
    For bitstring $x=0^n$,
    \begin{equation*}
        \expt{p(U,0^n)^2} = \expt{\sum_{i,j,k,l} g(\sigma_i)\overline{g(\sigma_j)}g(\sigma_k)\overline{g(\sigma_l)} p_ip_jp_kp_l}
    \end{equation*}
        Using \cref{eqn:relabeling_equality},  the four-fold summation can be categorized into five equivalent classes. We define the partition $\boldsymbol{\alpha}$ be the array with at most $4$ entries whose sum is exactly $4$. 
    \begin{enumerate}
        \item All indices are distinct with partition $\boldsymbol{\alpha} = (1,1,1,1)$: $i\ne j\ne k\ne l$. The contribution is $N(N-1)(N-2)(N-3) \mc{H}_2 \mc{I}_{(1,1,1,1)}$ where $\mc{I}_{(1,1,1,1)} = \frac{1}{N(N+1)(N+2)(N+3)}$.
        \item Only three indices are distinct with partition $\boldsymbol{\alpha} = (2,1,1)$: $i=j\ne k\ne l$ and its other five permutations. The contribution is $6N(N-1)(N-2) \mc{H}_1\mc{I}_{(2,1,1)}$ where $\mc{I}_{(2,1,1)} = \frac{2}{N(N+1)(N+2)(N+3)}$.
        \item Only two indices are distinct with partition $\boldsymbol{\alpha} = (3,1)$: $i=j=k\ne l$ and its other three permutations. The contribution is $4N(N-1) \mc{H}_1\mc{I}_{(3,1)}$ where $\mc{I}_{(3,1)} = \frac{6}{N(N+1)(N+2)(N+3)}$.
        \item Only two indices are distinct with partition $\boldsymbol{\alpha} = (2,2)$: $(i=j)\ne(k=l)$ and its other two permutations. The contribution is $3N(N-1)\mc{I}_{(2,2)}$ where $\mc{I}_{(2,2)} = \frac{4}{N(N+1)(N+2)(N+3)}$.
        \item All indices are the same with partition $\boldsymbol{\alpha} = (4)$: $i=j=k=l$. The contribution is $N \mc{I}_{(4)}$ where $\mc{I}_{(4)} = \frac{24}{N(N+1)(N+2)(N+3)}$.
    \end{enumerate}
        
    Collecting all the terms together, we have
    \begin{equation*}
        \expt{p(U,0^n)^2} = \frac{12}{(N+2)(N+3)} + \frac{12 N (N-1) \mc{H}_1}{(N+1)(N+2)(N+3)} + \frac{(N-1)(N-2)(N-3)}{(N+1)(N+2)(N+3)} \mc{H}_{2}.
    \end{equation*}
    
    The evaluation of the second moment of the probability with nonzero bitstrings $x$ follows a similar procedure but is more involved. Using \cref{lma:iddist-relabel}, $\expt{\sum_{x\ne 0^n} p(U,x)^2} = (N-1) \expt{p(U,1^n)^2}$. Let 
    \begin{equation*}
            H(i,j,k,l) := \expt{\overline{V_{1^n,i}} V_{0^n,i} V_{1^n,j} \overline{V_{0^n,j}}\, \overline{V_{1^n,k}} V_{0^n,k} V_{1^n,l} \overline{V_{0^n,l}}}.
    \end{equation*}
    Then, 
    \begin{equation*}
        \expt{p(U,1^n)^2} = \sum_{ijkl}\expt{g(\sigma_i)\overline{g(\sigma_j)}g(\sigma_k)\overline{g(\sigma_l)}} H(i,j,k,l).
    \end{equation*}
    Let us consider a linear map $\phi : \left(\CC^N\right)^{\otimes 4} \rightarrow \left(\CC^N\right)^{\otimes 4}$,
    \begin{equation*}
        \phi_{ijkl} := \expt{\bigotimes_{q \in \{i,j,k,l\}} V \ket{q}\bra{q} V^\dagger},
    \end{equation*}
Then
\begin{equation*}
H(i,j,k,l) = \braket{0^n,1^n,0^n,1^n| \phi_{ijkl} | 1^n,0^n,1^n,0^n}.
\end{equation*}
    Note that the linear map $\phi_{ijkl}$ commutes with the action of the symmetric group $\mathsf{S}_4$ and that of the unitary group $\mathsf{U}(N)$ on the tensor product space $\left(\CC^N\right)^{\otimes 4}$. Then, by Schur's lemma over $\CC$, $\phi_{ijkl}$ is a scalar on each subspace decomposed with respect to Schur-Weyl duality. By taking trace on each subspace, the linear map is determined as the linear combination of projectors. The generic formula was derived in \cite{CollinsSniady2006} by which $H(i,j,k,l)$ is evaluated. Let the four-fold sum in terms of $ijkl$ be broken into five equivalent classes as follows.
    \begin{enumerate}
        \item All indices are distinct with partition $(1,1,1,1)$: $i\ne j\ne k\ne l$. The contribution is $N(N-1)(N-2)(N-3) \mc{H}_2 \frac{2}{N^2\,(N-1)(N+1)(N+2)(N+3)}$.
        \item Only three indices are distinct with partition $(2,1,1)$: $i=j\ne k\ne l$ and its other five permutations. The contribution is $6N(N-1)(N-2) \mc{H}_1\left(-\frac{2\,\left(N-3\right)}{3\,N^2\,(N-1)(N+1)(N+2)(N+3)}\right)$.
        \item Only two indices are distinct with partition $(3,1)$: $i=j=k\ne l$ and its other three permutations. The contribution is $4N(N-1) \mc{H}_1 \left(- \frac{4}{N(N-1)(N+1)(N+2)(N+3)}\right)$.
        \item Only two indices are distinct with partition $(2,2)$: $(i=j)\ne(k=l)$ and its other two permutations. The contribution is $3N(N-1) \frac{2\,\left(N^2+N+6\right)}{3\,N^2\,(N-1)(N+1)(N+2)(N+3)}$.
        \item All indices are the same with partition $(4)$: $i=j=k=l$. The contribution is $N \frac{4}{N\,\left(N+1\right)\,\left(N+2\right)\,\left(N+3\right)}$.
    \end{enumerate}
    To conclude, the second moment of the probability with nonzero bitstrings is
    \begin{equation*}
        \expt{\sum_{x\ne 0^n} p(U,x)^2} = \frac{2(N-1)(N^2+3N+6)}{N(N+1)(N+2)(N+3)} - \frac{4(N-1)(N^2-N+6)\mc{H}_1}{N(N+1)(N+2)(N+3)} + \frac{2(N-1)(N-2)(N-3)}{N(N+1)(N+2)(N+3)} \mc{H}_{2}.
    \end{equation*}
\end{proof}

The asymptotic behavior of the ensemble is discussed in \cref{sec:asymptotic_larget}.

\subsection{Concatenating phase factors for long time Hamiltonian simulation}\label{sec:long_phase}
Although simulation at very long time is certainly beyond the regime of near term applications, the unitarity of Hamiltonian simulation provides an alternative method to obtain phase factors. 
Specifically, given the phase factor sequence at some short time $t$, the phase factor sequence at a long simulation time $rt$ can be easily constructed for some integer $r>1$.
The procedure, called phase factor concatenation, is given in \cref{eqn:phase_concatenate}, and the quality of the approximation is describe in \cref{thm:phase_concatenate}.

\begin{theorem}\label{thm:phase_concatenate}
    Let $\Phi = (\phi_0, \cdots, \phi_d)$ be a set of phase factors so that $\norm{P(x,\Phi) - s_t(x)}_\infty \le \epsilon$. Given an integer $r\ge 1$, define
    \begin{equation}
        \Phi^{(r)} := \left( \phi_0, \phi_1, \cdots, \phi_{d-1}, \underbrace{\phi_d + \phi_0, \phi_1, \cdots, \phi_{d-1}}_{\text{repeat }r-1 \text{ times}}, \phi_d \right).
\label{eqn:phase_concatenate}
    \end{equation}
    Then, 
    \begin{equation}
    \norm{P(x,\Phi^{(r)}) - s_{tr}(x)}_\infty \le r^2 \epsilon.
    \end{equation}
\end{theorem}

\begin{proof}
        For simplicity, let $P(x) := P(x,\Phi)$ and $P^{(r)}(x) := P(x,\Phi^{(r)})$. According to the normalization condition, 
        \begin{equation*}
            \begin{split}
                &(1-x^2) \abs{Q(x)}^2 = 1 - P(x) \overline{P(x)} = \left( s_t(x) - P(x) \right) \overline{P(x)} + s_t(x) \overline{\left(s_t(x) - P(x)\right)}.
            \end{split}
        \end{equation*}
        By the triangle inequality, $\norm{(1-x^2) \abs{Q(x)}^2}_\infty \le 2 \epsilon$.
        
        Let $\epsilon_r$ be the approximation error of $P^{(r)}$, and let $Q^{(r)}$ be the complementing polynomial in as \cref{thm:qsp}. Similarly, we have $\norm{(1-x^2) \abs{Q^{(r)}(x)}^2}_\infty \le 2 \epsilon_r$. By the construction of the phase factor sequence in \cref{eqn:phase_concatenate}, we have
        \begin{equation*}
            P^{(r)}(x) = P^{(r-1)}(x) P(x) - (1-x^2) Q^{(r-1)}(x) \overline{Q(x)}.
        \end{equation*}
        This gives
        \begin{equation*}
        \begin{split}
            & P^{(r)}(x) - s_t^r(x) = \left( P^{(r-1)}(x) - s_t^{r-1}(x) \right) P(x) + s_t^{r-1}(x) \left( P(x) - s_t(x) \right) - (1-x^2) Q^{(r-1)}(x) \overline{Q(x)}.
        \end{split}
        \end{equation*}
        Note that $s_t^r(x)=s_{tr}(x)$.  By the triangle inequality, $\epsilon_r \le \epsilon_{r-1} + \epsilon + 2 \sqrt{\epsilon_{r-1} \epsilon}=(\sqrt{\epsilon_{r-1}}+\sqrt{\epsilon})^2$. Then, $\sqrt{\epsilon_r} \le \sqrt{\epsilon_{r-1}} + \sqrt{\epsilon} \le \cdots \le r\sqrt{\epsilon}$. 
Therefore
\begin{equation}
\epsilon_r\le r^2 \epsilon,
\end{equation}
which proves the theorem.
\end{proof}

Following the construction of \cref{eqn:phase_concatenate}, in order to perform Hamiltonian simulation at time $rt$, the length of the phase factor sequence increases by a factor of $r$, and the error grows at most quadratically with respect to $r$ according to \cref{thm:phase_concatenate}. 
However, from \cref{tab:numerical-precision,tab:phase-factor-t-opt} we observe that this estimate can be significantly improved if we can construct the phase factor sequence at time $rt$ directly using the optimization based method. 
Even when \cref{eqn:phase_concatenate} is used, numerical results in \cref{fig:prop_qsvt_circuit} also indicate that $\epsilon_r\le r^2 \epsilon$ is only an upper bound of the numerical error, which only grows linearly at least until $t=10$.

\subsection{Estimating the number of measurements}\label{sec:meas_error}
In this subsection, we  estimate the number of measurement shots needed to achieve given accuracy. In the experimental implementation, the measurement probability is computed by counting the frequency of a bit or bitstring $\mf{b}$. For the $i$-th measurement, the outcome $\mf{b}_i$ is associated with an indicator $\mc{I}_i(\mf{b}) := \delta_{\mf{b}_i, \mf{b}}$ which is $1$ if the outcome is $b$ and is $0$ otherwise. If the probability of measuring $\mf{b}$ in the experiment is $p_\expl(\mf{b})$, then the indicator $\mc{I}_i(\mf{b})$'s are i.i.d. Bernoulli distributed, namely
\begin{equation*}
    \mc{I}_i(\mf{b})\mathrm{\text{ i.i.d. }} \sim \mathrm{Bernoulli}\left(p_\expl(\mf{b})\right) = \left\{
        \begin{array}{ll}
            1 &\mathrm{\text{, with probability }} p_\expl(\mf{b}),  \\
            0 &\mathrm{\text{, with probability }} 1 - p_\expl(\mf{b}).
        \end{array}
    \right.
\end{equation*}
Then, the frequency becomes an unbiased estimator of the probability $p_\expl(\mf{b})$. If $M_\text{meas}$ measurement shots are used, the bitstring frequency is defined as
\begin{equation*}
     \hat{p}_\expl(\mf{b}) := \frac{1}{M_\text{meas}} \sum_{i=1}^{M_\text{meas}} \mc{I}_i(\mf{b}).
\end{equation*}
Furthermore, using the Bernoulli distribution, the mean and variance of the estimator can be explicitly computed
\begin{equation*}
    \expt{\hat{p}_\expl(\mf{b})} = p_\expl(\mf{b}),\quad \text{Var}\left(\hat{p}_\expl(\mf{b})\right) = \frac{p_\expl(\mf{b})\left(1-p_\expl(\mf{b})\right)}{M_\text{meas}}.
\end{equation*}
The central limit theorem suggests that when $M_\text{meas}$ is sufficiently large,
\begin{equation*}
\bP\left(\abs{\hat{p}_\expl(\mf{b}) - p_\expl(\mf{b})} \le 2\sqrt{2} \sqrt{\text{Var}\left(\hat{p}_\expl(\mf{b})\right)}\right)>0.99.
\end{equation*}
In other words, with confidence level higher than $99\%$, it suffices to bound the deviation as
\begin{equation*}
    \abs{\hat{p}_\expl(\mf{b}) - p_\expl(\mf{b})} \le 2\sqrt{2} \sqrt{\text{Var}\left(\hat{p}_\expl(\mf{b})\right)} \le \delta
\end{equation*}
where $\delta$ is some error control parameter, which gives
\begin{equation*}
    M_\text{meas} \ge \frac{8 p_\expl(\mf{b})\left(1-p_\expl(\mf{b})\right)}{\delta^2}.
\end{equation*}
Note that any probability is bounded $0 \le p_\expl(\mf{b}) \le 1$ which further gives $p_\expl(\mf{b})\left(1-p_\expl(\mf{b})\right) \le \frac{1}{4}$. Hence it suffices to choose
\begin{equation*}
    M_\text{meas} \ge \left\lceil\frac{2}{\delta^2}\right\rceil
\end{equation*}
so that the deviation in the probability is bounded $\abs{\hat{p}_\expl(\mf{b}) - p_\expl(\mf{b})} \le \delta$ with high probability. Then, when computing the QUES by measuring each circuit with $M_\text{meas} \ge \left\lceil\frac{2}{\delta^2}\right\rceil$ shots, which is referred to as $\widehat{\text{QUES}} := \expt{\hat{P}_\expl(U)}$, the statistical error is bounded as
\begin{equation*}
    \abs{\widehat{\text{QUES}} - \text{QUES}} \le \expt{\abs{\hat{P}_\expl(U_i)-P_\expl(U_i)}} \le \delta.
\end{equation*}
The circuit fidelity is estimated by $\hat{\alpha}_\text{QUES} := 2 \times \widehat{\text{QUES}} - 1$. Furthermore, the error is bounded as
\begin{equation}
    \abs{\hat{\alpha}_\text{QUES} - \alpha} \le \abs{\hat{\alpha}_\text{QUES} - \alpha_\text{QUES}} + \abs{\alpha_\text{QUES} - \alpha} = 2 \abs{\widehat{\text{QUES}} - \text{QUES}} + \abs{\alpha_\text{QUES} - \alpha} \le 2 \delta + 16 \epsilon + \Or(\epsilon^2)
\end{equation}
where the last inequality uses \cref{eqn:alpha_relation}. The derived error bound is a generalization of \cref{eqn:alpha_relation} by including the Monte Carlo measurement error due to the finite number of measurement shots.

\subsection{Asymptotic behavior of the long time Hamiltonian simulation benchmark}\label{sec:asymptotic_larget}
The first and second moments of the bitstring probability  are directly relevant to the construction of the system linear cross-entropy benchmarking based on Hamiltonian simulation. For simplicity, we define 
\begin{equation}
    K_1(x, t) := \expt{p_t(U, x)}, \text{ and } K_2(x, t) := \expt{p_t(U,x)^2},
\end{equation}
where the dependence on $t$ is encoded in the implementation of the quantum circuit. In this section, we will investigate the behavior of these moments in different regimes in terms of the Hamiltonian simulation time $t$ when $N$ is sufficiently large.

According to \cref{lma:iddist-relabel}, $K_1(x,t)$ and $K_2(x,t)$ are constant for any nonzero bitstring $x \ne 0^n$. Therefore, by using \cref{thm:HS-avg}, we have
\begin{equation*}
    \begin{split}
        & K_1(0^n, t) = \mc{H}_1(t) + \Or\left(\frac{1}{N}\right),\ K_2(0^n, t) = \mc{H}_2(t) + \frac{12}{N} \left(\mc{H}_1(t) - \mc{H}_2(t)\right) + \Or\left(\frac{1}{N^2}\right),
    \end{split}
\end{equation*}
and for any $x \ne 0^n$,
\begin{equation*}
    \begin{split}
        & K_1(x, t) = \frac{1}{N}\left(1-\mc{H}_1(t)\right) + \Or\left(\frac{1}{N^2}\right),\ K_2(x, t) = \frac{1}{N^2} \left(2 - 4 \mc{H}_1(t) + 2 \mc{H}_2(t)\right) + \Or\left(\frac{1}{N^3}\right).
    \end{split}
\end{equation*}
Note that by definition, $\mc{H}_1(t)$ and $\mc{H}_2(t)$ are the cosine transformation of the joint eigenvalue distribution $\bP_\eig^{(2)}$ and $\bP_\eig^{(4)}$ respectively.  According to \cref{thm:numeigdist-RACBEM}, these joint distributions are polynomials. Then both $\mc{H}_1(t)$ and $\mc{H}_2(t)$ converge to zero as $t\to \infty$. In particular, there exists $t^*$ such that $\mc{H}_1(t), \mc{H}_2(t) = \Or\left(\frac{1}{N}\right)$ for any $t > t^*$. In this regime, for any nonzero bitstring $x \ne 0^n$,
\begin{equation*}
    K_1(x,t) = \frac{1}{N} + \Or\left(\frac{1}{N^2}\right),\ \text{and } K_2(x, t) = \frac{2}{N^2} + \Or\left(\frac{1}{N^3}\right).
\end{equation*}
Note that the bitstring probability of a Haar-distributed unitary $U$, $p_{ij} := \abs{U_{ij}}^2$, has the first and second moment $\expt{p_{ij}} = \frac{1}{N}$ and $\expt{p_{ij}^2} = \frac{2}{N^2} + \Or\left(\frac{1}{N^3}\right)$. Furthermore, by using \cref{lma:iddist-relabel,thm:HS-avg}, the cross moment of two different nonzero bitstrings $x \ne y$ is
\begin{equation*}
    \expt{p_t(U,x)p_t(U,y)} = \frac{1-2\mc{H}_1(t)+\mc{H}_2(t)}{N^2} + \Or\left(\frac{1}{N^3}\right).
\end{equation*}
Thus, in the regime where $\mc{H}_1(t), \mc{H}_2(t) = \Or\left(\frac{1}{N}\right)$, the cross moment is $\frac{1}{N^2} + \Or\left(\frac{1}{N^3}\right)$. Following \cref{thm:Haar-succ-prob}, the cross moment of success probabilities of a Haar-distributed unitary is exactly the same up to higher order $\expt{p_{ij} p_{kj}} = \frac{1}{N(N+1)} = \frac{1}{N^2} + \Or\left(\frac{1}{N^3}\right)$. Remarkably, the correlation between different nonzero bitstrings is small, which is quantified by the covariance $\text{Cov}\left(p(U,x)p(U,y)\right) = \expt{p(U,x)p(U,y)} - \expt{p(U,x)}\expt{p(U,y)} = \Or\left(\frac{1}{N^3}\right)$. We conclude that in the defined regime, the ensemble of the time evolution matrix induced by our construction has approximately the same statistics as that of Haar-distributed unitaries up to at least the second moment. 

Note that the  threshold $\alpha^*_t := \frac{\mc{H}_1(t)}{2-5\mc{H}_1(t)+4\mc{H}_2(t)}$ is directly related to $\mc{H}_1(t)$ and $\mc{H}_2(t)$. When $t = t^\text{opt}$ for small time or $t > t^*$ in the long time regime, we have $\mc{H}_1(t), \mc{H}_2(t) = \Or\left(\frac{1}{N}\right)$ which leads to an exponentially small threshold $\alpha^*_t = \Or\left(\frac{1}{N}\right)$. 
It allows the quantum supremacy to be achieved for a small circuit fidelity $\alpha > \alpha^*_t$.
Note $t^*$ can be impractically large for near-term applications. 
Hence it is crucial that the simulation at $t = t^\text{opt}\approx 4.81$ is equally effective and much more tractable.

\end{document}